\documentclass{article}
\usepackage{mathtools}
\usepackage{amsfonts}
\usepackage{amsmath}
\usepackage{amsthm}
\usepackage{amssymb}
\usepackage{mathrsfs}
\usepackage{graphicx}
\usepackage{caption}
\usepackage{dsfont}
\usepackage{esint}
\usepackage{xcolor}
\usepackage{cases}
\usepackage{subcaption}
\usepackage{booktabs}
\usepackage{url}
\usepackage{hyperref}
\usepackage{makecell}
\usepackage{longtable}
\usepackage{authblk}

\def\dss{\displaystyle}
\PassOptionsToPackage{obeyFinal}{todonotes}
\usepackage[backgroundcolor=green!40, bordercolor = white, linecolor=green!40, colorinlistoftodos, textsize=footnotesize]{todonotes}

\theoremstyle{plain}
\newtheorem{theorem}{Theorem}[section]

\newtheorem{corollary}[theorem]{Corollary}
\newtheorem{lemma}[theorem]{Lemma}

\newtheorem{proposition}[theorem]{Proposition}

\newtheorem{conjecture}[theorem]{Conjecture}
\theoremstyle{definition}

\newtheorem{remark}[theorem]{Remark}

\newcommand{\Rq}{\mathcal R_0^{\mathrm{q}}}

\newcommand{\subjclass}[2][2020]{
  \noindent MSC#1: #2\par
}

\newcommand{\keywords}[1]{
  \noindent \textit{Keywords}: #1\par
}

\hypersetup{
  colorlinks=true,
  linkcolor=blue,
  citecolor=blue,
  urlcolor=blue
}

\newcommand{\KE}{K_E}

\title{Tipping the Balance: Allee Thresholds, Saddle-Node Bifurcations, and Optimal Sterile-Male Release Strategies for \textit{Anopheles} Mosquitoes}

\author[1,2,3]{Abba Gumel}
\author[1]{C. Alex Safsten} 

\affil[1]{Department of Mathematics, University of Maryland, College Park, MD 20742, USA} 
\affil[2]{Institute for Health Computing, University of Maryland, North Bethesda, Maryland 20852, USA} 
\affil[3]{Department of Mathematics and Applied Mathematics, University of Pretoria, Pretoria 0002, South Africa} 

\date{June 15, 2026}

\begin{document}
	
	\maketitle

    \begin{abstract}
        We formulate and analyze a sex- and stage-structured model for \textit{Anopheles} dynamics under the sterile insect technique (SIT), motivated by the need for tools robust to insecticide resistance and outdoor transmission. The model tracks aquatic stages, adult males, unmated females, and females mated with wild or sterile males; includes egg-laying capacity and larval competition; and uses a refractory period followed by density-dependent mate search. The resulting Holling type-II mating term generates a mate-finding Allee effect. After establishing well-posedness, we prove that this Allee effect makes the mosquito-free equilibrium locally stable for all admissible parameters and globally asymptotically stable when a quick-mate-search reproduction number $R_0^q$ is below one. When $R_0^q>1$, habitat capacity is large, and larval competition is weak, two positive equilibria arise through a saddle--node bifurcation: a stable natural equilibrium and an unstable Allee equilibrium separating persistence from extinction. For a reduced model, a Goh--Volterra Lyapunov functional estimates the persistence basin. We then show how constant and population-responsive sterile-male releases reshape this bistability. Sufficiently large releases annihilate the positive equilibria in a second saddle--node bifurcation, while a sufficiently large constant release drives local elimination from every admissible initial state. Thus SIT need only push the population across the Allee separatrix, after which mate-finding failure can complete extinction. In a free-horizon optimization framework with an Allee-threshold stopping rule, a hybrid release strategy reduces the sterile-male requirement by about $5\%$ relative to the best constant-only strategy and $39\%$ relative to the best population-responsive-only strategy. These results recast the Allee effect as a control lever for vector suppression.

        \medskip
        
    \end{abstract}

	\section{Introduction}
\noindent
\noindent
Malaria remains one of the deadliest vector-borne diseases of humans.
Caused by protozoan \emph{Plasmodium} parasites and transmitted by
infected adult female \emph{Anopheles} mosquitoes
\cite{BatRan2005,WHO2025}, the disease has exacted a persistent global
public-health toll since the spillover of \emph{Plasmodium} from
non-human primates approximately 12,000 years ago \cite{Man2023}. Of
the several \emph{Plasmodium} species infectious to humans,
\emph{P.~falciparum} and \emph{P.~vivax} together account for nearly
the entire global burden \cite{CDC2018,WHO2025}, with
\emph{P.~falciparum} responsible for the most severe and
life-threatening clinical manifestations \cite{GetEtAl2011,JohSmiFid2013}.
Over 2.5~billion people currently reside in areas at risk of
\emph{P.~falciparum} transmission \cite{GetEtAl2011}, and the disease
caused an estimated 282~million cases and 610,000 deaths worldwide in
2024, with mortality heavily concentrated among children under five
years of age in sub-Saharan Africa \cite{WHO2025}.

Substantial progress was achieved in the global fight against malaria
between 2000 and 2015 \cite{GuiEtAl2020,MohGum2019}, driven primarily
by the widespread deployment of insecticide-based vector-control
interventions. Long-lasting insecticidal nets (LLINs) and indoor
residual spraying (IRS) together accounted for approximately $81\,\%$
of the reduction in malaria burden over this period \cite{MohGum2019},
with LLINs alone attributed to a $68\,\%$ reduction in the incidence
of \emph{P.~falciparum} \cite{GuiEtAl2020}. Complementary
contributions from early diagnosis, improved antimalarial therapy, and
strengthened public-health infrastructure further accelerated these
gains \cite{MohGum2019}. This progress galvanized renewed global
ambitions, reflected in the WHO Global Technical Strategy for Malaria
2016--2030---which sets targets of $90\,\%$ reductions in malaria case
incidence and mortality by 2030 relative to 2015 baseline levels
\cite{WHO2015}---and in the ZERO by~40 commitment, a public--private
partnership pledging to deliver a pipeline of novel vector-control
tools with the overarching goal of ending malaria by 2040
\cite{Gat2019,WilHam2018}.

These ambitions are, however, imperiled by several converging threats.
The sustained, large-scale use of chemical insecticides has driven
widespread resistance in \emph{Anopheles} mosquitoes to every compound
class currently deployed in LLINs and IRS
\cite{AloRocDabCoh2017,BarKayChiHas2018,CarEtAl2022,KumKoeCoeMzi2022,%
MohGum2019,RivVezWeiReaGan2010}, threatening to erode the hard-won
gains of the preceding decade. This biological challenge is compounded
by the growing epidemiological importance of outdoor malaria
transmission: existing tools are designed primarily to target
mosquitoes that bite and rest indoors, and offer limited efficacy
against opportunistic species such as \emph{An.~arabiensis} that feed
both indoors and outdoors, or against exclusively outdoor-feeding
species \cite{MunEtAl2022,WHO2015report}. Concurrently,
\emph{Plasmodium} resistance to frontline antimalarial drugs is
advancing in Southeast Asia and sub-Saharan Africa
\cite{BuyElsDur2021,CDC2018,HalBhaSaf2018,ManEtAl2022,Tow2009}, and a
universally safe and highly effective malaria vaccine has yet to be
realized \cite{DamEtAl2019,deV2019,Max2021}. Taken together, these
converging obstacles underscore an urgent need for alternative or
complementary vector-control strategies capable of interrupting both
indoor and outdoor transmission.

Among the most promising biological approaches is the sterile insect
technique (SIT), an area-wide integrated pest- and vector-management
strategy based on the periodic mass release of radiation-sterilized
male insects \cite{DycHenRob2021}. Sterile males compete with wild
males for mates; when a wild female mates with a sterile male, the
resulting eggs are non-viable, causing a progressive reduction in wild
population abundance with each successive generation. Although the
theoretical foundations of SIT were established in the 1930s and 1940s
\cite{DycHenRob2021}, the technique was first deployed operationally in
the United States in the 1950s, achieving the eradication of the
screwworm fly \emph{Cochliomyia hominivorax} \cite{Kni1955,Kni1957}.
Subsequent refinements have yielded successful applications against
several agricultural pests---including fruit flies, moths, and tsetse
flies---and, more recently, against \emph{Aedes aegypti}, the principal
vector of dengue, chikungunya, yellow fever, and Zika viruses
\cite{IboGumTay2020,Kni1979,PatEtAl2015,VreEtAl2000}. Advances in
mass-rearing, irradiation, and release technologies have renewed
interest in extending SIT to \emph{Anopheles}-borne diseases
\cite{BenRob2003}. Pilot field trials in Italy, China, and elsewhere
have produced encouraging results \cite{BelMedBalCar2013,ZheEtAl2019},
though investigations targeting dominant \emph{Anopheles} vectors in
Africa---where the malaria burden is greatest---have only recently been
initiated \cite{HelEtAl2008,IboGumTay2020,MunEtAl2011,MunEtAl2022}.

Large-scale SIT implementation nevertheless faces several
well-recognized biological and operational challenges, including
imperfect sex-sorting prior to release
\cite{DumYat2022,DumYat2024,GuiEtAl2020,PodEtAl2018}, residual
fertility among irradiated males \cite{DumYat2024},
sterilization-associated fitness costs that reduce male mating
competitiveness and survival
\cite{CulEtAl2020,GuiEtAl2020,HelKno2009,PodEtAl2018}, and seasonal
or climate-driven fluctuations in wild mosquito abundance
\cite{DumDup2024,DouEtAl2021,IboGumTay2020,LeGEtAl2019}. These
factors collectively govern whether SIT can achieve local suppression
or elimination of a target population and, critically, how sterile
males should be released in practice. Mathematical modeling has been
central to addressing these questions. Deterministic models with
impulsive releases and environmental forcing have quantified
SIT-induced changes in mosquito abundance and persistence
\cite{CaiAiLi2014,DumTch2012,IboGumTay2020,LiCaiLi2017};
sex-structured models have elucidated the roles of mating dynamics,
residual fertility, and release design
\cite{AroDum2020,BliCarDumVas2019,DumYat2024}; and analyses
incorporating nonlinear mating interactions have demonstrated that
mate-finding Allee effects can induce bistability, making population
elimination achievable once the wild population is suppressed below a
critical density \cite{AngDumYatIvr2020,DumYat2022}. Building on this
bistability structure, optimal control has emerged as a natural
framework for SIT design, since practical success requires not only
biological efficacy but cost-effective scheduling of sterile-male
releases \cite{BliCarDumVas2024,HuaYouLiuSon2021,ThoYanEst2010}.
Spatial aspects of SIT have also been explored in PDE frameworks,
including barrier-release strategies designed to prevent reinvasion of
cleared territory \cite{AlmEstVau2022,EvaBis2014,MaLiCai2026}.

Despite these advances, a pronounced gap remains. Relatively few
studies combine detailed life-cycle structure, nonlinear mating
dynamics, and control design within a single analytically tractable
framework; because these biological components simultaneously elevate
model dimensionality and nonlinearity, rigorous qualitative analysis is
frequently sacrificed in favor of purely numerical investigations.
Moreover, optimal control formulations in the SIT literature have been
developed almost exclusively over fixed time horizons
\cite{BliCarDumVas2024,HuaYouLiuSon2021,ThoYanEst2010}; free-horizon
formulations, in which releases continue dynamically until a
state-dependent ecological criterion is satisfied, remain largely
unexplored. A widely cited operational benchmark recommends maintaining a sterile-to-wild
male overflooding ratio of at least $10{:}1$
\cite{OliEtAl2021,CarEtAl2022,StrBosDum2019}; this guideline, however, rests on
accumulated field experience rather than on a formal optimization of
sterile-male release effort.

The present study addresses these gaps. We analyze a sex- and
stage-structured deterministic model for the \emph{Anopheles} mosquito
population \cite{IboGumTay2020} that incorporates immature aquatic
developmental stages, density-dependent larval mortality, and nonlinear
mating dynamics giving rise to a mate-finding Allee effect, under
continuous non-constant sterile-male releases. The system is analyzed
with progressively increasing complexity: we begin with the SIT-free
regime to establish baseline equilibrium and bifurcation structures,
then incorporate sterile-male releases to characterize how control
modifies these dynamics. Specifically, the following questions are
addressed:
\begin{enumerate}
    \item[(i)] What is the optimal sterile-male release
        strategy---within the class of constant and proportional
        release rates---that drives the wild population to local
        elimination while minimizing the total number of sterile
        males deployed?
    \item[(ii)] Once the wild population has been driven below the
        Allee threshold, can SIT releases be safely discontinued,
        and under what conditions will the natural Allee dynamics
        thereafter guarantee population extinction without further
        intervention?
\end{enumerate}

The paper is organized as follows. The model is formulated and its
well-posedness established in Section~\ref{sec:model}. A rigorous
analysis of the equilibrium structure and bifurcation properties of the
model, in both the SIT-free and controlled regimes, is presented in
Section~\ref{sec:analysis}. Numerical investigations, including the
computation of optimal release strategies and a sensitivity analysis of
the total sterile-male requirement, are given in
Section~\ref{sec:numerics}. Concluding remarks and directions for future work appear in Section~\ref{sec:discussion}.

\section{The Model}\label{sec:model} \subsection{Preamble and Modeling Assumptions}

The model considered in this work is adapted from the sterile insect
technique~(SIT) framework of Iboi, Gumel, and
Taylor~\cite{IboGumTay2020}.  Two simplifications are made relative to
that work.  First, whereas~\cite{IboGumTay2020} resolves the larval
stage into four compartments corresponding to the four instars, the
present model consolidates these into a single larval compartment.
Second, whereas~\cite{IboGumTay2020} employs impulsive (discrete,
periodic) releases of sterile males, the present model admits release
rates that are continuous in time, though not necessarily constant.
Both simplifications are motivated by mathematical tractability: the
reduced structure permits a thorough analytical treatment of
equilibrium existence, bifurcation structure, and stability that
complements the primarily numerical, seasonally-forced analysis
of~\cite{IboGumTay2020}.  The numerical investigations in
Section~\ref{sec:numerics} are directed toward optimizing the
sterile-male release rate so as to drive the wild mosquito population
to extinction while minimizing the total number of sterile males
deployed.

The model is formulated under the following assumptions.

\begin{enumerate}
\item[(A1)] \textit{Spatial homogeneity.}  The mosquito population is
  well mixed within the habitat; spatial heterogeneity is neglected.

\item[(A2)] \textit{Closed habitat.}  Immigration of mosquitoes from
  external sources is absent.

\item[(A3)] \textit{Mate-search kinetics.}  Following a
  post-emergence refractory period of mean duration~$\zeta > 0$
  (see below), an unmated female actively seeks a mate.  The expected
  duration of the mate-search phase is inversely proportional to the
  effective male density $M_w + \eta M_s$, with proportionality
  constant $\gamma > 0$.

\item[(A4)] \textit{Perfect sex sorting and sterilization.}  No female
  mosquitoes are inadvertently released alongside sterile males, and
  all released males carry no residual
  fertility~\cite{DumYat2022,DumYat2024,GuiEtAl2020,PodEtAl2018}.

\item[(A5)] \textit{Constant habitat capacity.}  Environmental
  conditions are sufficiently stable that the carrying capacity of the
  egg-laying habitat remains constant over the timescale of interest.
\end{enumerate}

\subsection{State Variables, Equations, and Parameter
  Descriptions}\label{sec:formulation}

The total mosquito population at time~$t$ is partitioned into eight
mutually exclusive compartments: eggs~($E$), larvae~($L$), pupae~($P$),
unmated adult females~($F_u$), adult females mated with wild
males~($F_{mw}$), adult females mated with sterile males~($F_{ms}$),
wild adult males~($M_w$), and sterile adult males~($M_s$).  The
population dynamics are governed by the following system of nonlinear
ordinary differential equations:
\begin{equation}\label{eq:model_I}
  \begin{cases}
    \dot{E}
       = \phi\!\left(1-\dfrac{E}{\KE}\right)F_{mw}
         -(\sigma_E+\mu_E)\,E,
    \\[10pt]
    \dot{L}
       = \sigma_E E - (\sigma_L+\mu_L+\delta_L L)\,L,
    \\[10pt]
    \dot{P}
       = \sigma_L L - (\sigma_P+\mu_P)\,P,
    \\[10pt]
    \dot{F}_u
       = r\sigma_P P
         - \dfrac{M_w+\eta M_s}{\gamma+\zeta(M_w+\eta M_s)}\,F_u
         - \mu_F F_u,
    \\[10pt]
    \dot{F}_{mw}
       = \dfrac{M_w}{\gamma+\zeta(M_w+\eta M_s)}\,F_u
         - \mu_F F_{mw},
    \\[10pt]
    \dot{F}_{ms}
       = \dfrac{\eta M_s}{\gamma+\zeta(M_w+\eta M_s)}\,F_u
         - \mu_F F_{ms},
    \\[10pt]
    \dot{M}_w
       = (1-r)\sigma_P P - \mu_M M_w,
    \\[6pt]
    \dot{M}_s
       = S(t) - \mu_M M_s.
  \end{cases}
\end{equation}
All parameter values are nonneg\-ative; mortality rates are strictly
positive.  The state variables and parameters are described in
Table~\ref{tab:model_I_parameters}; additional discussion follows.

\medskip
\noindent\textbf{Oviposition and carrying capacity.}
The parameter $\phi > 0$ is a composite oviposition parameter equal to
the product of the mean egg-laying frequency and the mean clutch size
per oviposition event of a mated adult female.  The parameter~$\KE >
0$ is the environmental carrying capacity for eggs, reflecting the
finite availability of suitable oviposition sites; the logistic
factor $(1-E/\KE)$ in the first equation of~\eqref{eq:model_I}
ensures that egg production is density-dependently suppressed as~$E$
approaches~$\KE$.

\medskip
\noindent\textbf{Aquatic-stage dynamics.}
For $x\in\{E,L,P\}$, $\sigma_x > 0$ denotes the per-capita maturation
rate from aquatic stage~$x$ to the subsequent developmental stage, and
$\mu_x > 0$ the corresponding per-capita natural mortality rate.  The
coefficient $\delta_L > 0$ is the density-dependent larval mortality
rate, capturing intraspecific competition for food, space, and other
limiting
resources~\cite{GilLeeSolBen2011,MurCouMboGod2012,YosEtAl2012}.  The
parameter $r\in(0,1)$ is the female sex ratio at eclosion, i.e., the
fraction of newly eclosed adults that are female.

\medskip
\noindent\textbf{Mating dynamics.}
The nonlinear mating terms in~\eqref{eq:model_I} are governed by three
parameters: $\eta$, $\zeta$, and $\gamma$.

The parameter $\eta\in(0,1)$ is a \emph{fitness-cost coefficient}
quantifying the reduced mating competitiveness of sterile males
relative to wild males~\cite{LeeEtAl2015,SutPriWed2021}.  The value
$\eta=1$ corresponds to fully competitive sterile males; smaller
values reflect greater sterilization-induced fitness costs, consistent
with field observations~\cite{CulEtAl2020,HelKno2009,PodEtAl2018}.

Upon eclosion, a female mosquito undergoes a post-emergence refractory
period of mean duration $\zeta > 0$ days, during which she completes
sexual maturation before initiating mate-seeking
behavior~\cite{ChaEtAl2003,OliBenLemGil2011,StoHamFos2011}.  Once
sexually mature, the unmated female searches actively for a mate.  By
assumption~(A3), the expected duration of this mate-search phase is
$\gamma/(M_w + \eta M_s)$, so the mean total pre-mating interval
(refractory period plus mate-search phase) for a single unmated female
is $\zeta + \gamma/(M_w+\eta M_s)$.  The total per-female mating rate
is therefore
\begin{equation}\label{eq:mating_rate}
  R_{\mathrm{mating}}
  = \frac{F_u}{\;\zeta + \dfrac{\gamma}{M_w+\eta M_s}\;}
  = \frac{M_w+\eta M_s}{\gamma+\zeta(M_w+\eta M_s)}\;F_u.
\end{equation}
The parameter $\gamma>0$ is interpreted as the mean first-encounter
time between a single unmated female and a single effective male (i.e.,
evaluated at $M_w+\eta M_s=1$); it depends on the mosquito diffusivity
and the spatial extent of the habitat, and is derived from first
principles via mean first-passage-time
theory~\cite{SzaSchSch1980,LeVYusAbaDen2020} in
Appendix~\ref{app:spatial}.  Decomposing~\eqref{eq:mating_rate} by
partner type gives
\begin{equation}\label{eq:mating_decomp}
  R_{\mathrm{mating,w}}
    = \frac{M_w}{\gamma+\zeta(M_w+\eta M_s)}\;F_u,
  \qquad
  R_{\mathrm{mating,s}}
    = \frac{\eta M_s}{\gamma+\zeta(M_w+\eta M_s)}\;F_u,
\end{equation}
representing mating events with wild and sterile males, respectively.
Setting $\beta = 1/\gamma$ recovers the equivalent formulation with a
leading competitiveness coefficient employed in~\cite{IboGumTay2020}.

A key structural feature of the mating
rate~\eqref{eq:mating_rate} is that the denominator $\gamma +
\zeta(M_w+\eta M_s)$ remains bounded away from zero for all $\gamma
>0$, regardless of male population size.  Consequently, at low
mosquito densities the mating rate is \emph{quadratically} small in
the state variables: it is bilinear in $(M_w + \eta M_s,\,F_u)$ near
the origin, constituting a Holling type-II functional response that
degenerates to a mass-action term~\cite{XiaRua2001}.  This quadratic
vanishing is the mechanistic source of the mate-finding Allee effect
analyzed in Section~\ref{sec:analysis}.

\medskip
\noindent\textbf{Sterile-male release rate.}
The function $S(t)\geq 0$ denotes the per-unit-time release rate of
sterile males into the habitat.  Although field implementations
typically deploy sterile males in discrete, periodic
batches~\cite{IboGumTay2020,DumTch2012,LiAi2020}, such impulsive
schedules are commonly well approximated by continuous functions on
timescales long relative to the inter-release
interval~\cite{EstYan2005,LiCaiLi2017,ThoYanEst2010}.  In the present
study, $S(t)$ is assumed to be continuous, nonneg\-ative, and bounded
for all $t\geq 0$; specific functional forms are analyzed in
Sections~\ref{sec:analysis} and~\ref{sec:numerics}.

\renewcommand{\arraystretch}{1.3}
\begin{longtable}{|c|p{11.2cm}|}
\caption{State variables and parameters of model~\eqref{eq:model_I}.
All parameter values are nonnegative; mortality rates and maturation rates are strictly positive.}
\label{tab:model_I_parameters}
\\
  \hline\hline
  \multicolumn{2}{|c|}{\bfseries State Variables}\\ \hline
  $E$        & Number of mosquito eggs \\ \hline
  $L$        & Number of mosquito larvae \\ \hline
  $P$        & Number of mosquito pupae \\ \hline
  $F_u$      & Number of unmated adult female mosquitoes \\ \hline
  $F_{mw}$   & Number of adult female mosquitoes mated with wild males \\ \hline
  $F_{ms}$   & Number of adult female mosquitoes mated with sterile males \\ \hline
  $M_w$      & Number of wild adult male mosquitoes \\ \hline
  $M_s$      & Number of sterile adult male mosquitoes \\ \hline
  \multicolumn{2}{|c|}{\bfseries Parameters}\\ \hline
  $\phi$     & Composite oviposition parameter: product of the mean
               egg-laying frequency and mean clutch size per oviposition
               event of a mated adult female \\ \hline
  $\KE$      & Environmental carrying capacity for eggs \\ \hline
  $\sigma_E$ & Per-capita egg-to-larva maturation (hatching) rate \\ \hline
  $\mu_E$    & Per-capita natural mortality rate of eggs \\ \hline
  $\sigma_L$ & Per-capita larva-to-pupa maturation rate \\ \hline
  $\mu_L$    & Per-capita natural mortality rate of larvae \\ \hline
  $\delta_L$ & Density-dependent larval mortality coefficient \\ \hline
  $\sigma_P$ & Per-capita pupa-to-adult maturation (eclosion) rate \\ \hline
  $\mu_P$    & Per-capita natural mortality rate of pupae \\ \hline
  $r$        & Female sex ratio at eclosion: fraction of newly eclosed
               adults that are female, $r\in(0,1)$ \\ \hline
  $\eta$     & Fitness-cost coefficient for sterile-male mating
               competitiveness relative to wild males,
               $\eta\in(0,1)$ \\ \hline
  $\gamma$   & Mean first-encounter time between a single unmated female
               and a single effective male in the habitat \\ \hline
  $\zeta$    & Mean duration of the post-emergence refractory period of
               adult females prior to mate-seeking \\ \hline
  $\mu_F$    & Per-capita natural mortality rate of adult female
               mosquitoes \\ \hline
  $\mu_M$    & Per-capita natural mortality rate of adult male
               mosquitoes \\ \hline
  $S(t)$     & Per-unit-time sterile-male release rate (control input) \\ \hline\hline
\end{longtable}

\begingroup
\renewcommand{\arraystretch}{1.25}
\begin{table}[!htbp]
    \centering
    \begin{tabular}{|c|c|p{7.5cm}|}
    \hline\hline
        Parameter & Baseline value & Source(s) \\ \hline\hline
        $\phi$ & $26$ day${}^{-1}$ & \cite{AfrGitYan2012,TakKloCha1998} \\ \hline
        $\KE$ & $10^5$ & See Appendix~\ref{app:spatial} \\ \hline
        $\sigma_E$ & $0.37$ day${}^{-1}$ & \cite{BayLin2003} \\ \hline
        $\mu_E$ & $0.054$ day${}^{-1}$ & \cite{BayLin2003} \\ \hline
        $\sigma_L$ & $0.091$ day${}^{-1}$ & \cite{BayLin2003} \\ \hline
        $\mu_L$ & $0.054$ day${}^{-1}$ & \cite{BayLin2003} \\ \hline
        $\delta_L$ & $5\times 10^{-5}$ day${}^{-1}$ & See Appendix~\ref{app:spatial} \\ \hline
        $\sigma_P$ & $0.37$ day${}^{-1}$ & \cite{BayLin2003} \\ \hline
        $\mu_P$ & $0.054$ day${}^{-1}$ & \cite{BayLin2003} \\ \hline
        $r$ & $0.5$ & \cite{MazKidMyaKwe2019} \\ \hline
        $\eta$ & $0.75$ & \cite{EstYan2005} \\ \hline
        $\gamma$ & $450$ days & See Appendix~\ref{app:spatial} \\ \hline
        $\zeta$ & $1$ day & \cite{OliBenLemGil2011,Fer2022} \\ \hline
        $\mu_F$ & $0.083$ day${}^{-1}$ & \cite{AduGetYew2022,MatBetOse2020} \\ \hline
        $\mu_M$ & $0.15$ day${}^{-1}$ & \cite{LamLeeGabWat2022} \\ \hline\hline
    \end{tabular}
    \caption{Baseline parameter values used in the numerical simulations
    of model~\eqref{eq:model_I}.  Unless otherwise stated, each
    simulation is performed using these values.  The spatially dependent
    parameters $\KE$, $\delta_L$, and $\gamma$ are interpreted relative
    to the reference domain described in Appendix~\ref{app:spatial}.}
    \label{tab:parameter_values}
\end{table}
\endgroup
%

\subsection{Relationship to Some Related Prior Models}\label{sec:prior_models}

\noindent\textbf{Relationship to Iboi--Gumel--Taylor (2020).}
Beyond the simplifications described in Section~\ref{sec:model}, the
present model differs from~\cite{IboGumTay2020} in its treatment of
the release schedule (continuous rather than impulsive) and in the
form of the sterile-male release rate considered in
Section~\ref{sec:analysis}, which combines a constant component~$S_0$
with a component proportional to the current wild adult population.
These choices preserve the autonomous structure of the model and
facilitate the complete bifurcation analysis carried out in
Section~\ref{sec:analysis}.

\medskip
\noindent\textbf{Relationship to Esteva--Yang (2005).}
Model~\eqref{eq:model_I} is also related to the Esteva--Yang
model~\cite{EstYan2005}, which likewise describes SIT dynamics via a
per-female mating rate depending on wild and sterile male abundance.
The two formulations are, however, mechanistically distinct.  Adapting
the Esteva--Yang mating term to the notation of the present model,
their per-female mating rate is
\begin{equation}\label{eq:EY_mating}
  R_{\mathrm{mating,E\text{-}Y}}
  = \frac{M_w + \eta M_s}
         {\zeta(1-\eta)M_s + \zeta(M_w + \eta M_s)}\;F_u.
\end{equation}
This expression superficially resembles~\eqref{eq:mating_rate} with
the constant $\gamma$ replaced by the density-dependent
factor~$\zeta(1-\eta)M_s$.  This substitution, however, introduces a
critical deficiency: the effective encounter-time coefficient
$\zeta(1-\eta)M_s$ vanishes identically whenever sterile males are
absent ($M_s = 0$) or fully competitive ($\eta = 1$), reducing the
expected mate-search duration to zero in precisely those
regimes---implying instantaneous mate-finding even at arbitrarily low
population densities, a biologically untenable conclusion.

In the present model, by contrast, the strictly positive,
density-independent parameter $\gamma > 0$ ensures that mate-finding
limitation is present at all population densities, including in the
complete absence of SIT intervention.  It is this unconditional
positivity of~$\gamma$ that gives rise to the mate-finding Allee
effect established in Theorem~\ref{thm:LAStrivial} and that renders
the mosquito-free equilibrium locally asymptotically stable for all
admissible parameter values.  A formal demonstration that setting
$\gamma = 0$ eliminates the Allee effect---and that the mosquito-free
equilibrium becomes unstable when $\Rq > 1$---is given in
Appendix~\ref{app:linear_mating}.

Both models share the prediction that reduced sterile-male
competitiveness ($\eta < 1$) suppresses population-level reproductive
success.  The present framework offers two additional mechanistic
contributions beyond~\cite{EstYan2005}.  First, the explicit
decomposition of the pre-mating interval into a post-emergence
refractory phase~($\zeta$) and a subsequent mate-search
phase~($\gamma/(M_w + \eta M_s)$) provides a transparent,
two-timescale representation of female reproductive biology.  Second,
the parameter~$\gamma$, derived from mean first-passage-time theory
(Appendix~\ref{app:spatial}), furnishes a mechanistic link between
individual-level spatial encounter dynamics and population-level
reproductive rates, lending the model greater interpretive and
predictive transparency.

The analysis in Section~\ref{sec:analysis} recovers, within this more
detailed framework, the qualitative results established by Esteva and
Yang for the SIT-free system---namely, the existence of two nontrivial
positive equilibria when the net reproductive output exceeds unity,
and their mutual annihilation through a saddle--node bifurcation as
SIT intensity is increased---together with explicit asymptotic formulas
for the bifurcation thresholds and a more complete characterization of
the global stability structure.

\subsection{Well-Posedness of the Model}\label{sec:well_posedness}

The biologically feasible region for model~\eqref{eq:model_I} is
\begin{equation*}
  \Omega
  = \mathbb{R}^8_+
  = \bigl\{(E,L,P,F_u,F_{mw},F_{ms},M_w,M_s)\in\mathbb{R}^8 :
    \text{all components} \geq 0\bigr\}.
\end{equation*}

\begin{proposition}[Local existence and uniqueness]
  \label{prop:local_existence}
  For each initial condition in $\Omega$, model~\eqref{eq:model_I}
  has a unique solution on a maximal time interval $[0,T_{\max})$.
\end{proposition}

\begin{proof}
Write~\eqref{eq:model_I} as $\dot{X} = G(X,t)$, where
$X=(E,L,P,F_u,F_{mw},F_{ms},M_w,M_s)^\top$.  For any initial
condition in $\Omega$, the mating denominators $\gamma +
\zeta(M_w+\eta M_s)$ are bounded away from zero in a neighborhood of
the initial condition (since $\gamma > 0$).  Consequently, $G$ is
continuously differentiable---hence locally Lipschitz---in~$X$
throughout~$\Omega$, and continuous in~$t$.  Local existence and
uniqueness on a maximal interval $[0,T_{\max})$ follow from the
Picard--Lindel\"of theorem~\cite{Teschl2012}.  Global existence is
established in Corollary~\ref{cor:global existence}.
\end{proof}

\begin{remark}
The strict positivity of $\gamma$ (Assumption~(A3)) is essential: it
ensures that the mating terms in~\eqref{eq:model_I} are smooth on all
of~$\Omega$, including at zero mosquito density, avoiding the
degeneracy that would arise if $\gamma = 0$ and $(M_w,M_s) = (0,0)$
simultaneously.
\end{remark}

\begin{proposition}[Positive invariance of $\Omega$]
  \label{prop:positivity}
  If $S(t)\geq 0$ for all $t\geq 0$, then every solution
  of~\eqref{eq:model_I} with initial condition in~$\Omega$ remains
  in~$\Omega$ for all $t\in[0,T_{\max})$.
\end{proposition}

\begin{proof}
It suffices to verify that the vector field $G$ points inward (or is
tangent to the boundary) on each face of~$\Omega$ defined by setting
one state variable to zero, with all remaining variables
nonneg\-ative.  Inspection of each equation of~\eqref{eq:model_I}
confirms this: on every such face the corresponding time derivative is
nonneg\-ative.  (For the $\dot{M}_s$ equation, the condition $S(t)
\geq 0$ is necessary and sufficient.)  Hence no trajectory originating
in~$\Omega$ can exit through any boundary face, and $\Omega$ is
positively invariant.
\end{proof}

\begin{proposition}[Uniform boundedness]
  \label{prop:bounded_solns}
  If $S(t)$ is nonneg\-ative and bounded, then every solution
  of~\eqref{eq:model_I} with initial condition in~$\Omega$ is bounded
  in~$\Omega$ for all $t\in[0,T_{\max})$.
\end{proposition}

\begin{proof}
By Proposition~\ref{prop:positivity}, all solution components are
nonneg\-ative, so it remains to establish upper bounds.  The proof
rests on the following standard comparison fact: if a nonneg\-ative
function~$u$ satisfies $\dot{u}\leq a - bu$ for constants $a,b > 0$,
then the Gr\"onwall inequality~\cite{Teschl2012} gives
\begin{equation}\label{eq:Gronwall_bound}
  u(t) \leq u_{\max} := \max\!\left\{u(0),\,\frac{a}{b}\right\}
  \quad \text{for all } t\in[0,T_{\max}).
\end{equation}

\emph{Bound on $E$.}  If $E \geq \KE$, then $1 - E/\KE \leq 0$, so
$\dot{E} \leq -(\sigma_E + \mu_E)\KE < 0$.  Hence $E$ cannot
increase through~$\KE$, yielding
\[
  0 \leq E(t) \leq E_{\max} := \max\{E(0),\,\KE\}.
\]

\emph{Sequential bounds on remaining variables.}  Since $E$ is now
bounded, the comparison fact~\eqref{eq:Gronwall_bound} yields upper
bounds successively via the following dependency chain:
\begin{itemize}
  \item $L_{\max}$: from $\dot{L}\leq \sigma_E E_{\max} -
    (\sigma_L+\mu_L)L$ (the term $-\delta_L L^2 \leq 0$ only
    strengthens the inequality);
  \item $M_{s,\max}$: from $\dot{M}_s \leq S_{\max} - \mu_M M_s$,
    where $S(t)\leq S_{\max}<\infty$ by hypothesis;
  \item $P_{\max}$: from $\dot{P}\leq \sigma_L L_{\max} -
    (\sigma_P+\mu_P)P$;
  \item $F_{u,\max}$ and $M_{w,\max}$: from the equations for
    $\dot{F}_u$ and $\dot{M}_w$, using $P_{\max}$;
  \item $F_{mw,\max}$ and $F_{ms,\max}$: from the equations for
    $\dot{F}_{mw}$ and $\dot{F}_{ms}$, using $F_{u,\max}$,
    $M_{w,\max}$, and $M_{s,\max}$.
\end{itemize}
Explicit expressions for each upper bound are recorded in
Appendix~\ref{app:ext_proof}.
\end{proof}

\begin{corollary}[Global existence and uniqueness]
  \label{cor:global existence}
  If $S(t)$ is nonneg\-ative and bounded, then every solution
  of~\eqref{eq:model_I} with initial condition in~$\Omega$ exists and
  is unique for all $t\geq 0$; that is, $T_{\max} = \infty$.
\end{corollary}

\begin{proof}
Since $G(X,t)$ is locally Lipschitz in~$X$ and all solutions
in~$\Omega$ are uniformly bounded by
Proposition~\ref{prop:bounded_solns}, global existence and uniqueness
follow from standard continuation arguments~\cite{Teschl2012}.
\end{proof}

\noindent
Propositions~\ref{prop:local_existence}--\ref{prop:bounded_solns} and
Corollary~\ref{cor:global existence} together confirm that
model~\eqref{eq:model_I} is well-posed both mathematically and
ecologically on~$\Omega$: solutions exist globally, are unique,
preserve nonneg\-ativity, and remain bounded for all biologically
admissible initial conditions and release schedules.  The equilibrium
structure and stability properties of the model are analyzed in
Section~\ref{sec:analysis}.

    \section{Existence and Asymptotic Stability of Equilibria of the Model}\label{sec:analysis}

\subsection{Existence and stability of equilibria of SIT-free model}

Before analyzing the full model~\eqref{eq:model_I}, it is
instructive to examine the special case in which no sterile males are
released, obtained by setting $S(t) \equiv 0$
in~\eqref{eq:model_I}.  This SIT-free reduction isolates the
intrinsic dynamics of the wild mosquito population and provides the
analytical foundation for understanding the effects of SIT
intervention in subsequent sections.  The SIT-free model admits a trivial
mosquito free equilibrium (MFE):
\begin{equation*}
    \mathcal E_0
    = (E,\,L,\,P,\,F_u,\,F_{mw},\,F_{ms},\,M_w,\,M_s)
    = (0,\,0,\,0,\,0,\,0,\,0,\,0,\,0).
\end{equation*}
 
\begin{theorem}\label{thm:LAStrivial}
The MFE $\mathcal E_0$ of the SIT-free
model~\eqref{eq:model_I} is locally asymptotically stable for all
admissible parameter values.
\end{theorem}
 
\begin{proof}
Setting $S(t) \equiv 0$, the Jacobian of~\eqref{eq:model_I}
evaluated at $\mathcal E_0$ is
\begin{equation}\label{eq:AI_0}
    A(\mathcal E_0)=\begin{pmatrix}
     -\mu_E-\sigma_E & 0 & 0 & 0 & \phi & 0 & 0 & 0  \\
     \sigma_E & -\mu_L-\sigma_L & 0 & 0 & 0 & 0 & 0 & 0 \\
     0 & \sigma_L & -\mu_P-\sigma_P & 0 & 0 & 0 & 0 & 0 \\
     0 & 0 & r\sigma_P & -\mu_F & 0 & 0 & 0 & 0 \\
     0 & 0 & 0 & 0 & -\mu_F & 0 & 0 & 0 \\
     0 & 0 & 0 & 0 & 0 & -\mu_F & 0 & 0 \\
     0 & 0 & (1-r)\sigma_P & 0 & 0 & 0 & -\mu_M & 0\\
     0 & 0 & 0 & 0 & 0 & 0 & 0 & -\mu_M
    \end{pmatrix}.
\end{equation}
The numerators of the mating terms $\frac{M_w}{\gamma+\zeta(M_w+\eta M_s)}F_u$ and
$\frac{\eta M_s}{\gamma+\zeta(M_w+\eta M_s)}F_u$ are bilinear in
$(M_w,F_u)$ and $(M_s,F_u)$, respectively; since both factors in
each product vanish at $\mathcal E_0$, all partial derivatives of these
terms with respect to any state variable are zero at $\mathcal E_0$.
Consequently, rows~5, 6, and~8 of $A(\mathcal E_0)$---corresponding to
$\dot{F}_{mw}$, $\dot{F}_{ms}$, and $\dot{M}_s$---contain only
their respective diagonal entries, contributing three decoupled
eigenvalues $-\mu_F$ (multiplicity two) and $-\mu_M$ directly.
The remaining $5\times 5$ submatrix, corresponding to the variables
$(E,L,P,F_u,M_w)$, is lower-triangular (the entry $\phi$ appearing
in row~1 of $A(\mathcal E_0)$ at column~5 corresponds to $F_{mw}$,
which lies outside this submatrix and does not affect its spectrum),
with diagonal entries $-(\sigma_E+\mu_E)$, $-(\sigma_L+\mu_L)$,
$-(\sigma_P+\mu_P)$, $-\mu_F$, and $-\mu_M$.  The eight eigenvalues of $A(\mathcal{E}_{0})$
are $-(\sigma_{E}+\mu_{E})$ (row~1), $-(\sigma_{L}+\mu_{L})$
(row~2), $-(\sigma_{P}+\mu_{P})$ (row~3), $-\mu_{F}$ from
rows~4,\,5,\,6 (multiplicity~3), and $-\mu_{M}$ from rows~7
and~8 (multiplicity~2), total $1+1+1+3+2=8$, all strictly
negative, and $\mathcal E_0$ is locally
asymptotically stable for all admissible parameter values.
\end{proof}
 
The local stability of $\mathcal E_0$ admits a natural mechanistic
explanation.  The per-female mating rate in~\eqref{eq:model_I},
namely
\[
    \frac{M_w}{\gamma + \zeta M_w}\,F_u
    \;\sim\;
    \frac{1}{\gamma}\,M_w F_u
    \quad \text{as } (M_w,F_u)\to (0,0),
\]
is \emph{quadratic} (bilinear) in the state variables near $\mathcal E_0$
(a Holling type-II functional response that degenerates to a
mass-action term at low densities \cite{XiaRua2001}).
Consequently, mating events become vanishingly rare as the population
approaches extinction: unmated females cannot reproduce, and the
absence of males near $\mathcal E_0$ ensures that no new mated females
are produced.  This renders the effective reproduction pathway
inactive in a neighborhood of $\mathcal E_0$, and the linearized
dynamics are governed entirely by mortality and developmental
transitions.  Equivalently, the basic reproduction number
$\mathcal{R}_0$ of the SIT-free model~\eqref{eq:model_I} at
$\mathcal E_0$, computed via the next-generation operator
method \cite{DieHeeMet1990,VanWat2002}, equals zero---a consequence of the
fact that the mating term contributes no linear terms to the Jacobian
at $\mathcal E_0$, so that no new mated females are generated in the
linearized system.  The equilibrium $\mathcal E_0$ is therefore
\emph{superstable} in the sense that $\mathcal{R}_0 = 0$.  This
property has been identified in related ecological models employing
Holling type-II mating or predation terms~\cite{BouBer2009}, and is
consistent with the findings of several SIT modeling
studies~\cite{AlmDupPriVau2022,EstYan2005,ThoYanEst2010}, which
establish the same stability under constant, positive sterile-male
release rates.  The present result extends these findings by
demonstrating that the MFE remains locally
asymptotically stable even in the complete absence of SIT
intervention ($S(t) \equiv 0$), a consequence of the Holling type-II
structure of the mating term rather than the release of sterile males.
 
\noindent
The ecological implication of Theorem~\ref{thm:LAStrivial} is that a
sufficiently small wild mosquito population introduced into a
mosquito-free environment will fail to establish itself and will tend
to extinction.  This outcome is a direct consequence of the
mate-finding Allee effect inherent in the mating term: at low
population densities, the probability that an unmated female locates
a mate within her reproductive lifespan becomes negligible, precluding
population growth.
 
Theorem~\ref{thm:LAStrivial} can be extended to a global asymptotic
stability result for $\mathcal E_0$ under a condition on the following
threshold quantity:
\begin{equation}\label{eq:Rlm}
    \Rq
    \;:=\;
    \phi
    \left(\frac{\sigma_E}{\sigma_E+\mu_E}\right)
    \left(\frac{\sigma_L}{\sigma_L+\mu_L}\right)
    \left(\frac{r\sigma_P}{\sigma_P+\mu_P}\right)
    \left(\frac{1}{\mu_F(1+\zeta\mu_F)}\right).
\end{equation}
 
\begin{theorem}\label{thm:GAStrivial}
    The MFE $\mathcal E_0$ of the SIT-free
    model~\eqref{eq:model_I} is globally asymptotically stable in
    $\Omega$ whenever $\Rq < 1$.
\end{theorem}
 
\begin{proof}
Consider the candidate Lyapunov function $V:\Omega\to\mathbb{R}$
defined by
\begin{align*}
    V(E,L,P,F_u,F_{mw},F_{ms},M_w,M_s)
    &= E
     + \left(\frac{\sigma_E+\mu_E}{\sigma_E}\right)L
     + \left[\frac{(\sigma_E+\mu_E)(\sigma_L+\mu_L)}
                  {\sigma_E\sigma_L}\right]P \\
    &\quad
     + \left[\frac{(\sigma_E+\mu_E)(\sigma_L+\mu_L)(\sigma_P+\mu_P)}
                  {r\sigma_E\sigma_L\sigma_P}\right]F_u
     + \left(\frac{\phi}{\mu_F}\right)F_{mw},
\end{align*}
with time derivative along trajectories of~\eqref{eq:model_I} given
by
\begin{equation}\label{eq:Lyapunov_derivative}
\begin{split}
    \dot{V}
    &= \dot{E}
     + \left(\frac{\sigma_E+\mu_E}{\sigma_E}\right)\dot{L}
     + \left[\frac{(\sigma_E+\mu_E)(\sigma_L+\mu_L)}
                  {\sigma_E\sigma_L}\right]\dot{P} \\
    &\quad
     + \left[\frac{(\sigma_E+\mu_E)(\sigma_L+\mu_L)(\sigma_P+\mu_P)}
                  {r\sigma_E\sigma_L\sigma_P}\right]\dot{F}_u
     + \left(\frac{\phi}{\mu_F}\right)\dot{F}_{mw}.
\end{split}
\end{equation}
 
Using the inequality $\phi F_{mw}(1 - E/K_E) \leq \phi F_{mw}$, which
follows from the nonnegativity of $E$ and the equation for $\dot{E}$
in~\eqref{eq:model_I}, equation~\eqref{eq:Lyapunov_derivative}
satisfies
\begin{equation}\label{eq:Vdot_simplified}
    \begin{split}
    \dot{V}
    &\leq
    \frac{(\sigma_E+\mu_E)(\sigma_L+\mu_L)(\sigma_P+\mu_P)}
         {r\sigma_E\sigma_L\sigma_P}
    \Bigl[(\Rq-1)C_m
          +(\Rq C_m\zeta-1)\mu_F\Bigr]F_u \\
    &\quad
    -\delta_L\!\left(\frac{\sigma_E+\mu_E}{\sigma_E}\right)L^2,
    \end{split}
\end{equation}
where $C_m$ denotes the mating coefficient
\begin{equation*}
    C_m = \frac{M_w+\eta M_s}{\gamma+\zeta(M_w+\eta M_s)}.
\end{equation*}
A full derivation of how \eqref{eq:Vdot_simplified} can be obtained from \eqref{eq:Lyapunov_derivative} is found in Appendix \ref{app:ext_proof}. Since all state variables are nonneg\-ative in $\Omega$, we have
$C_m \geq 0$.  Moreover, although $M_w$ and $M_s$ may be arbitrarily
large in $\Omega$, the mating coefficient is bounded above: $C_m \leq
1/\zeta$.  Suppose now that $\Rq < 1$.  Then
$(\Rq - 1) < 0$, so the first term in the bracket is
nonpositive.  For the second term, the bound $C_m \leq 1/\zeta$
gives $\Rq C_m\zeta \leq \Rq < 1$,
so $(\Rq C_m\zeta - 1)\mu_F < 0$ as well.  Hence the
entire bracket is strictly negative, and $\dot{V} \leq 0$ throughout
$\Omega$.  Consequently, $V$ is a Lyapunov function on $\Omega$.
 
It remains to identify the largest invariant set on which $\dot{V} =
0$.  Let
\begin{equation*}
    X(t)=\bigl(E(t),\,L(t),\,P(t),\,F_u(t),\,F_{mw}(t),
    \,F_{ms}(t),\,M_w(t),\,M_s(t)\bigr)
\end{equation*}
be a trajectory satisfying $\dot{V}(X(t)) \equiv 0$.  Then $L(t)
\equiv F_u(t) \equiv 0$.  Since $L(t) \equiv 0$
implies $\dot{L}(t) \equiv 0$, the larval equation
in~\eqref{eq:model_I} gives
\begin{equation*}
    0 = \dot{L} = \sigma_E E - (\sigma_L+\mu_L+\delta_L L)L
      = \sigma_E E,
\end{equation*}
and hence $E(t) \equiv 0$.  Applying the same reasoning sequentially
to the remaining equations of~\eqref{eq:model_I} yields $P(t) \equiv
F_{mw}(t) \equiv F_{ms}(t) \equiv M_w(t) \equiv 0$.  Therefore the
largest invariant set contained in $\{\dot{V} = 0\}$ is 
\begin{equation*}
    \mathcal D=\{(0,0,0,0,0,0,0,M_s):M_s\geq 0\}.
\end{equation*}
By LaSalle's invariance
principle \cite{LaS1976}, the $\omega$-limit set of every
trajectory in $\Omega$ is contained in $\mathcal D$, so
\begin{equation*}
    \lim_{t\to\infty}(E(t),L(t),P(t),F_u(t),F_{mw}(t),F_{ms}(t))=(0,0,0,0,0,0,0).
\end{equation*}
Finally, since $\dot M_s=-\mu_M M_s$, we must have $\lim_{t\to\infty}M_s(t)=0$. Hence, all
trajectories originating in $\Omega$ converge to $\mathcal E_0$.
\end{proof}
 
\noindent
The ecological implication of Theorem~\ref{thm:GAStrivial} is that
reducing and sustaining $\Rq$ below unity is
sufficient to guarantee elimination of the mosquito population from
the habitat, irrespective of its initial size.  An examination of
the expression in~\eqref{eq:Rlm} identifies several actionable
pathways to achieving this condition.  Insecticide-based vector
control measures, such as indoor residual spraying (IRS) and
larviciding, that increase the mortality rates of immature stages
($\mu_E$, $\mu_L$, $\mu_P$) or adult mosquitoes ($\mu_F$, $\mu_M$)
can reduce $\Rq$ below one when applied at sufficient
intensity.  Complementary biological or genetic strategies---such as
those aimed at suppressing the egg-laying rate ($\phi$), reducing the
female sex ratio ($r$), or retarding aquatic-stage maturation
($\sigma_E$, $\sigma_L$, $\sigma_P$)---offer additional levers for
driving $\Rq$ below the elimination threshold.  Taken
together, these results suggest that combining standard
insecticide-based control with biological or genetic interventions
having the properties described above may substantially enhance the
prospect of mosquito population elimination.
 
Theorems~\ref{thm:LAStrivial} and~\ref{thm:GAStrivial} together
reveal a noteworthy feature of the SIT-free model: the trivial
equilibrium is \emph{always} locally asymptotically stable
(Theorem~\ref{thm:LAStrivial}), yet is globally asymptotically
stable only when $\Rq < 1$
(Theorem~\ref{thm:GAStrivial}).  This local-but-not-global stability
structure is a hallmark of a bistable system and is a direct
population-level manifestation of the mate-finding Allee effect
inherent in the mating term: at sufficiently low densities, the
probability that an unmated female locates a mate within her
reproductive lifespan becomes negligible, causing the population to
decline toward extinction regardless of the value of
$\Rq$.
 
The threshold $\Rq$ is precisely the basic
reproduction number of a limiting, simplified version of
model~\eqref{eq:model_I} obtained by setting $\gamma = 0$.  In this
limiting model, mate-finding is instantaneous---unmated females
locate a mate without delay---and no sterile males are present.
Because the mating term reduces to a linear (mass-action) expression
when $\gamma = 0$, this limiting system is referred to as the
\emph{quick mate search model}; hence the superscript $q$ in
$\Rq$.  The equations of the quick mate search model are
given in Appendix~\ref{app:linear_mating} for completeness, together
with a detailed explanation of how the strictly positive
mate-search time ($\gamma > 0$) induces the unconditional local
asymptotic stability of $\mathcal E_0$ observed in
Theorem~\ref{thm:LAStrivial}.  Because mate-finding is instantaneous
in the quick mate search model, adult females transition from the unmated
compartment $F_u$ to the mated compartment $F_{mw}$ without delay,
and are therefore able to lay eggs throughout a greater fraction of
their adult lifespan.  The mosquito abundance predicted by the linear
mating model consequently exceeds that of the full
model~\eqref{eq:model_I} for the same parameter values.  It follows that $\Rq < 1$ is sufficient---though not
necessary---for global extinction in the full model~\eqref{eq:model_I}:
because the quick mate search model overestimates reproductive output
relative to the full model (by setting the mate-search time $\gamma$
to zero), any parameter regime that drives the quick mate search model to
extinction will also drive the full model to extinction.  The
condition $\Rq < 1$ is therefore \emph{conservative}
in the sense that the true extinction threshold of the full
model---accounting for the mate-finding delay and the Holling type-II
nonlinearity of the mating term---lies strictly above the threshold
$\Rq = 1$.  In other words, global extinction in the
full model may occur even when $\Rq \geq 1$.


    

Since Theorem~\ref{thm:GAStrivial} establishes global asymptotic
stability of $\mathcal E_0$ when $\Rq < 1$, it is instructive to also investigate the dynamics of the SIT-free model when
$\Rq > 1$.  In particular, we ask whether nontrivial
(positive) equilibria exist and, if so, how many.
 
\begin{theorem}\label{thm:positive_equilibria}
    Suppose $\Rq > 1$ and $S(t) \equiv 0$
    in~\eqref{eq:model_I}.  Then there exist thresholds
    $K_E^{\ast} > 0$ and $\delta_L^{\ast} > 0$, depending on the
    model parameters, such that whenever $K_E > K_E^{\ast}$ and
    $\delta_L < \delta_L^{\ast}$, the SIT-free model admits exactly
    two positive (component-wise) equilibria $X_{-}^{\ast\ast}$ and
    $X_{+}^{\ast\ast}$, with each component of $X_{+}^{\ast\ast}$
    greater than or equal to the corresponding component of
    $X_{-}^{\ast\ast}$.
\end{theorem}
 
\begin{proof}
Let $X^{\ast\ast} = (E^{\ast\ast},\, L^{\ast\ast},\, P^{\ast\ast},\,
F_u^{\ast\ast},\, F_{mw}^{\ast\ast},\, F_{ms}^{\ast\ast},\,
M_w^{\ast\ast},\, M_s^{\ast\ast})$ denote an equilibrium of
model~\eqref{eq:model_I} with $S(t) \equiv 0$.  Setting $S(t) \equiv
0$ in the equations for $\dot{M}_s$ and $\dot{F}_{ms}$
in~\eqref{eq:model_I} and evaluating at equilibrium gives
$M_s^{\ast\ast} = 0$ and $F_{ms}^{\ast\ast} = 0$.  The equations
governing $F_{ms}$ and $M_s$ therefore decouple from the remaining
six equations at equilibrium, and the subsequent analysis is
restricted to these six equations in the variables
$(E, L, P, F_u, F_{mw}, M_w)$.
 
Solving each of the six remaining equilibrium equations in terms of
$L^{\ast\ast}$ yields
\begin{equation}\label{eq:equil_components}
\begin{split}
    P^{\ast\ast}
        &= \frac{\sigma_L\,L^{\ast\ast}}{\sigma_P+\mu_P},
    \quad
    M_w^{\ast\ast}
        = \frac{(1-r)\sigma_L\sigma_P\,L^{\ast\ast}}
               {\mu_M(\sigma_P+\mu_P)},
    \\[4pt]
    F_u^{\ast\ast}
        &= \frac{r\sigma_L\sigma_P\,L^{\ast\ast}
                 \bigl[(1-r)\zeta\sigma_L\sigma_P\,L^{\ast\ast}
                       + \gamma\mu_M(\sigma_P+\mu_P)\bigr]}
                {(\sigma_P+\mu_P)\,Q(L^{\ast\ast})},
    \\[4pt]
    F_{mw}^{\ast\ast}
        &= \frac{(1-r)r\sigma_L^2\sigma_P^2\,(L^{\ast\ast})^2}
                {\mu_F(\sigma_P+\mu_P)\,Q(L^{\ast\ast})},
    \\[4pt]
    E^{\ast\ast}
        &= \frac{(1-r)r\phi\sigma_L^2\sigma_P^2\,K_E\,(L^{\ast\ast})^2}
                {(1-r)r\phi\sigma_L^2\sigma_P^2\,(L^{\ast\ast})^2
                 + K_E\mu_F(\sigma_E+\mu_E)(\sigma_P+\mu_P)\,Q(L^{\ast\ast})},
\end{split}
\end{equation}
where
\begin{equation*}
    Q(L^{\ast\ast})
    = \gamma\mu_F\mu_M(\sigma_P+\mu_P)
      + (1-r)\sigma_L\sigma_P(1+\zeta\mu_F)\,L^{\ast\ast}.
\end{equation*}
Each expression in~\eqref{eq:equil_components} is a positive,
increasing function of $L^{\ast\ast}$ for $L^{\ast\ast} > 0$, and
each vanishes when $L^{\ast\ast} = 0$ (recovering the trivial
equilibrium $\mathcal E_0$).  It therefore suffices to determine the
number of positive solutions $L^{\ast\ast} > 0$.
 
From the larval equation in~\eqref{eq:model_I}, any equilibrium
value $L^{\ast\ast}$ must satisfy
\begin{equation}\label{eq:L_equilibrium}
    \sigma_E E^{\ast\ast} - (\sigma_L+\mu_L+\delta_L L^{\ast\ast})L^{\ast\ast} = 0.
\end{equation}
Substituting the expression for $E^{\ast\ast}$
from~\eqref{eq:equil_components} into~\eqref{eq:L_equilibrium} and
dividing through by $K_E$, one finds that $L^{\ast\ast}$ satisfies
\begin{equation}\label{eq:quartic}
    L^{\ast\ast}\bigl(a\,(L^{\ast\ast})^3 + b\,(L^{\ast\ast})^2
    + c\,L^{\ast\ast} + d\bigr) = 0,
\end{equation}
where
\begin{align}
    a &= \frac{\delta_L(1-r)\Rq\mu_F\sigma_L\sigma_P
               (\sigma_E+\mu_E)(\sigma_L+\mu_L)(\sigma_P+\mu_P)
               (1+\zeta\mu_F)}
              {K_E\sigma_E},
    \label{eq:a_coef}\\[4pt]
    b &= (1-r)\mu_F\sigma_L\sigma_P(1+\zeta\mu_F)(\sigma_E+\mu_E)
        (\sigma_P+\mu_P)
        \left(\frac{\Rq(\sigma_L+\mu_L)^2}{K_E\sigma_E}
              + \delta_L\right),
    \label{eq:b_coef}\\[4pt]
    c &= \mu_F(\sigma_E+\mu_E)(\sigma_P+\mu_P)
        \Bigl[-(\Rq-1)(1-r)\sigma_L\sigma_P
               (\sigma_L+\mu_L)(1+\zeta\mu_F)
              + \gamma\delta_L\mu_F\mu_M(\sigma_P+\mu_P)\Bigr],
    \label{eq:c_coef}\\[4pt]
    d &= \gamma\mu_F^2\mu_M(\sigma_E+\mu_E)(\sigma_L+\mu_L)(\sigma_P+\mu_P)^2.
    \label{eq:d_coef}
\end{align}
The factor $L^{\ast\ast} = 0$ in~\eqref{eq:quartic} corresponds to
the MFE $\mathcal E_0$.  Positive equilibria therefore
correspond to positive roots of the cubic polynomial
\begin{equation*}
    p(s) = as^3 + bs^2 + cs + d.
\end{equation*}
Since $p(0) = d > 0$ and
$p(s) \to -\infty$ as $s \to -\infty$, $p(s)$ has at least one negative
real root.  

It remains to show that, for $K_E$ sufficiently large and $\delta_L$ sufficiently small, there are also two positive roots. We write the dependence of $p(s)$ on $\delta_L$ explicitly: $p(s)=p_{\delta_L}(s)$.  In the limit as $\delta_L \to 0$, the cubic $p_{\delta_L}(s)$ approaches a limiting quadratic polynomial 
\begin{equation*}
p_0(s)=b_0 s^2+c_0 s+d_0,
\end{equation*}
where
\begin{align}
    b_0&=
        \frac{\Rq(1-r)\mu_F\sigma_L\sigma_P(1+\zeta\mu_F)
        (\sigma_E+\mu_E)(\sigma_P+\mu_P)(\sigma_L+\mu_L)^2}{K_E\sigma_E}\label{eq:b0}\\
    c_0&=-(\Rq-1)(1-r)\sigma_L\sigma_P\mu_F(\sigma_E+\mu_E)(\sigma_L+\mu_L)(\sigma_P+\mu_P)(1+\zeta\mu_F)\label{eq:c0}\\
    d_0&=\gamma\mu_F^2\mu_M(\sigma_E+\mu_E)(\sigma_L+\mu_L)(\sigma_P+\mu_P)^2.\label{eq:d0}
\end{align}
By Descartes' rule of signs, $p_0$ has no negative roots and $0$ or $2$ positive roots (counting multiplicity), depending on the parameter values. The discriminant of $p_0$ is $c_0^2-4b_0d_0$. Since $b_0\to 0$ as $\KE\to \infty$, there exists $\KE^\ast$ so that if $\KE>\KE^\ast$, then the discriminant of $p_0$ is positive, and so for such $\KE$, $p_0$ has two distinct positive roots denoted $s_{\pm,0}$ with $s_{-,0}<s_{+,0}$. Since these are two distinct roots of a quadratic polynomial, $p_0'(s_{\pm,0})\neq 0$. Therefore, by the implicit function theorem, there exists $\delta_L^\ast>0$ and continuous functions $s_\pm:[0,\delta_L^\ast)\to\mathbb R$ such that $s_\pm(0)=s_{\pm,0}$ and $p_{\delta_L}(s_\pm(\delta_L))=0$ for all $0<\delta_L<\delta_L^\ast$. Taking $\delta_L^\ast$ sufficiently small, we can be sure that $0<s_-(\delta_L)<s_+(\delta_L)$.

For given $\delta_L<\delta_L^\ast$, we define $L_{\pm}^{\ast\ast}=s_\pm(\delta_L)$. We may compute the asymptotic forms of $L_\pm^{\ast\ast}$ (with details of this calculation in Appendix \ref{app:ext_proof}):
\begin{align}
    L_+^{\ast\ast}&=\frac{\KE \left(\Rq-1\right) \sigma_E}{\Rq \left(\mu _L+\sigma _L\right)}-\frac{\gamma  \mu _F \mu _M \left(\mu _P+\sigma _P\right)}{(1-r) \left(\Rq-1\right) \sigma _L \sigma _P \left(\zeta  \mu _F+1\right)}+O(1/\KE,\delta_L),\label{eq:Lstar}\\
    L_-^{\ast\ast}&=\frac{\gamma  \mu _F \mu _M \left(\mu _P+\sigma _P\right)}{(1-r) \left(\Rq-1\right) \sigma _L \sigma _P \left(\zeta  \mu _F+1\right)}+O(1/\KE,\delta_L).\label{eq:Lstarstar}
\end{align}
Recalling that each component
of~\eqref{eq:equil_components} is an increasing function of
$L^{\ast\ast}$, we obtain two positive equilibria $X_-^{\ast\ast}$ and $X_+^{\ast\ast}$ with each component of $X_+^{\ast\ast}$ greater than or equal to the corresponding component of $X_-^{\ast\ast}$.
\end{proof}

Theorem~\ref{thm:positive_equilibria} and its proof show that
$\Rq > 1$ is necessary but not sufficient for the existence of
positive equilibria of model~\eqref{eq:model_I}.  Even when $\Rq >
1$, the additional conditions $K_E > K_E^{\ast}$ and $\delta_L <
\delta_L^{\ast}$ must also hold, and these become increasingly strict or
stringent as $\Rq$ approaches unity from above: the closer $\Rq$ is
to one, the larger the required habitat capacity and the weaker the
required larval competition.  Biologically, this reflects the
following mechanism.  The threshold $\Rq$ measures net reproductive
output under the most favorable mating conditions (instantaneous
mate-finding, $\gamma = 0$, as in the linear mating or ``quick-search" model); when
$\Rq$ barely exceeds one, the population is only marginally capable
of generational replacement even under these ideal conditions.  In
such a regime, the mate-finding Allee effect---which reduces
effective reproduction relative to the linear mating model---can
suppress population growth unless the habitat is sufficiently large
(large $K_E$, ensuring an adequate egg-laying resource) and larval
crowding is sufficiently weak (small $\delta_L$, reducing
density-dependent juvenile mortality).  Together, these conditions
allow the population to attain the density at which mate-finding
becomes sufficiently frequent to sustain persistence.
    
    \begin{figure}
        \centering
        \includegraphics[width=0.6\linewidth]{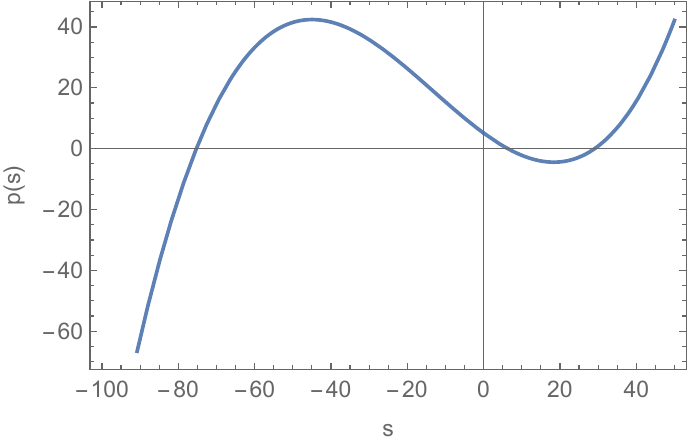}
        \caption{Plot of the cubic polynomial $p(s)$ from the proof of Theorem~\ref{thm:positive_equilibria}, whose positive roots are the $L$-components of the positive equilibria
             of model~\eqref{eq:model_I}.  The cubic always has one
             negative real root; it has two positive real roots when
             $K_E$ is sufficiently large and $\delta_L$ sufficiently
             small.  Parameter values used: $K_E = 1000$,
             $\delta_L = 0.05$; all remaining parameters as in
             Table~\ref{tab:parameter_values}}
        \label{fig:placeholder}
    \end{figure}

\begin{theorem}\label{thm:stable_equilibria}
    Suppose $\Rq > 1$ and $S(t) \equiv 0$ in
    model~\eqref{eq:model_I}, and let $K_E^{\ast}$,
    $\delta_L^{\ast}$, and $X_{\pm}^{\ast\ast}$ be as in
    Theorem~\ref{thm:positive_equilibria}.  Then there exist
    thresholds $K_E^{\ast\ast} \geq K_E^{\ast}$ and
    $\delta_L^{\ast\ast} \leq \delta_L^{\ast}$ such that whenever
    $K_E > K_E^{\ast\ast}$ and $\delta_L < \delta_L^{\ast\ast}$,
    the equilibrium $X_+^{\ast\ast}$ is locally asymptotically
    stable and $X_-^{\ast\ast}$ is unstable.
\end{theorem}
 
\begin{proof}
Suppose the conditions of Theorem~\ref{thm:positive_equilibria}
hold, so that the positive equilibria $X_{\pm}^{\ast\ast}$
of~\eqref{eq:model_I} exist.  We investigate the asymptotic
stability of each equilibrium in turn.  For given $K_E, \delta_L >
0$, let $A_{\pm}(K_E,\delta_L)$ denote the Jacobian
of~\eqref{eq:model_I} evaluated at $X_{\pm}^{\ast\ast}$.
 
\medskip
\noindent\textit{Asymptotic stability of $X_+^{\ast\ast}$.}
To compute the limiting Jacobian at \(X^{\ast\ast}_{+}\), we first record two
asymptotic identities following from \eqref{eq:Lstar}. In the limit
\(\delta_L\to 0\) and \(K_E\to\infty\),
\[
\frac{L^{\ast\ast}_{+}}{K_E}
\longrightarrow
\frac{(\mathcal R_0^q-1)\sigma_E}{\mathcal R_0^q(\sigma_L+\mu_L)}.
\]
Since the larval equilibrium equation gives $
\sigma_EE^{\ast\ast}_{+}
=
(\sigma_L+\mu_L+\delta_LL^{\ast\ast}_{+})L^{\ast\ast}_{+},$
the limiting relation with \(\delta_L=0\) yields
\[
\frac{E^{\ast\ast}_{+}}{K_E}
\longrightarrow
\frac{\sigma_L+\mu_L}{\sigma_E}
\cdot
\frac{(\mathcal R_0^q-1)\sigma_E}
{\mathcal R_0^q(\sigma_L+\mu_L)}
=
1-\frac{1}{\mathcal R_0^q}.
\]
Next, the egg equilibrium equation gives
\[
\phi\left(1-\frac{E^{\ast\ast}_{+}}{K_E}\right)F^{\ast\ast}_{mw,+}
=
(\sigma_E+\mu_E)E^{\ast\ast}_{+}.
\]
Dividing by \(K_E\) and using
\(E^{\ast\ast}_{+}/K_E\to 1-1/\mathcal R_0^q\), we obtain
\[
\frac{\phi F^{\ast\ast}_{mw,+}}{K_E}
\longrightarrow
(\sigma_E+\mu_E)
\frac{1-1/\mathcal R_0^q}{1/\mathcal R_0^q}
=
(\mathcal R_0^q-1)(\sigma_E+\mu_E).
\]
These identities determine the limiting contributions of the logistic
egg-laying term in the first row of the Jacobian.

Define
\begin{equation*}
    \widehat{A}_+
    \;:=\; \lim_{\substack{K_E\to\infty\\\delta_L\to 0}}
           A_+(K_E,\delta_L),
\end{equation*}
which is computed explicitly using the asymptotic expansion of
$X_+^{\ast\ast}$ in~\eqref{eq:Lstar}.  Columns~6 and~7, and row~8,
of $\widehat{A}_+$ each contain only their respective diagonal
entries $-\mu_F$, $-\mu_M$, and $-\mu_M$, contributing three
decoupled, strictly negative eigenvalues.  The remaining eigenvalues
are those of the $5\times 5$ leading submatrix
\begin{equation}\label{eq:Atilde_plus}
    \widetilde{A}_+
    = \begin{pmatrix}
        -\Rq(\sigma_E+\mu_E) & 0 & 0 & 0
            & \dfrac{\phi}{\Rq} \\[8pt]
        \sigma_E & -(\sigma_L+\mu_L) & 0 & 0 & 0 \\[4pt]
        0 & \sigma_L & -(\sigma_P+\mu_P) & 0 & 0 \\[4pt]
        0 & 0 & r\sigma_P
            & -\!\left(\mu_F+\dfrac{1}{\zeta}\right) & 0 \\[8pt]
        0 & 0 & 0 & \dfrac{1}{\zeta} & -\mu_F
      \end{pmatrix}.
\end{equation}
The characteristic polynomial of $\widetilde{A}_+$ is
\begin{equation}\label{eq:char_poly_plus}
    q(\lambda)
    = g(\lambda)
      - \frac{r\phi\sigma_E\sigma_L\sigma_P}{\zeta\,\Rq},
\end{equation}
where
\begin{equation*}
    g(\lambda)
    := \bigl[\lambda + \Rq(\sigma_E+\mu_E)\bigr]
       \bigl[\lambda + (\sigma_L+\mu_L)\bigr]
       \bigl[\lambda + (\sigma_P+\mu_P)\bigr]
       \bigl[\lambda + \!\left(\tfrac{1}{\zeta}+\mu_F\right)\bigr]
       \bigl[\lambda + \mu_F\bigr].
\end{equation*}
Suppose $\lambda\in\mathbb{C}$ satisfies $\mathrm{Re}(\lambda) \geq
0$.  For any $p > 0$, $|\lambda + p| \geq p$, so
\begin{align*}
    |g(\lambda)|
    &\;\geq\;
    \Rq(\sigma_E+\mu_E)(\sigma_L+\mu_L)(\sigma_P+\mu_P)
    \!\left(\frac{1}{\zeta}+\mu_F\right)\mu_F \\[4pt]
    &\;>\;
    (\sigma_E+\mu_E)(\sigma_L+\mu_L)(\sigma_P+\mu_P)
    \!\left(\frac{1}{\zeta}+\mu_F\right)\mu_F
    \;=\;
    \frac{r\phi\sigma_E\sigma_L\sigma_P}{\zeta\,\Rq},
\end{align*}
where the strict inequality uses $\Rq > 1$ and the final equality
follows from the definition~\eqref{eq:Rlm} of $\Rq$.  Hence,
$|g(\lambda)| > |q(\lambda) - g(\lambda)|$, and Rouché's theorem
implies $q(\lambda) \neq 0$ for all $\lambda$ with
$\mathrm{Re}(\lambda) \geq 0$.  Consequently, all eigenvalues of
$\widetilde{A}_+$, and therefore all eigenvalues of $\widehat{A}_+$,
have strictly negative real part.  By continuity of eigenvalues with
respect to matrix entries, there exist thresholds $K_E^{\ast\ast}
\geq K_E^{\ast}$ and $\delta_L^{\ast\ast} \leq \delta_L^{\ast}$
such that all eigenvalues of $A_+(K_E,\delta_L)$ have strictly
negative real part whenever $K_E > K_E^{\ast\ast}$ and $\delta_L <
\delta_L^{\ast\ast}$, establishing local asymptotic stability of
$X_+^{\ast\ast}$.
 
\medskip
\noindent\textit{Instability of $X_-^{\ast\ast}$.}
We will show that
\begin{equation}\label{eq:negative_determinant}
    \lim_{\substack{K_E\to\infty\\\delta_L\to 0}}\det A_-(K_E,\delta_L)<0
\end{equation}
Whenever $\Rq>1$. As the
determinant equals the product of all eight eigenvalues of
$\widehat{A}_-$, their product is negative.  The complex eigenvalues
of a real matrix arise in conjugate pairs, each pair contributing a
strictly positive product.  Therefore, the product of the real
eigenvalues alone is negative, which requires an odd number of
strictly positive real eigenvalues---and, hence, at least one
eigenvalue with strictly positive real part.  By continuity of
eigenvalues, $K_E^{\ast\ast}$ and $\delta_L^{\ast\ast}$ may be
chosen so that $A_-(K_E,\delta_L)$ retains at least one eigenvalue
with strictly positive real part whenever $K_E > K_E^{\ast\ast}$ and
$\delta_L < \delta_L^{\ast\ast}$, establishing the instability of
$X_-^{\ast\ast}$.

To see that \eqref{eq:negative_determinant} holds, we first observe that column 6 and row 8 of $A_-(\KE,\delta_L)$ contain only the diagonal elements $-\mu_F$ and $-\mu_M$ respectively. If $B_-(\KE,\delta_L)$ is the $6\times 6$ submatrix consisting of the remaining rows and columns, it follows that $\det A_-(\KE,\delta_L)=\mu_M\mu_F\det B_-(\KE,\delta_L)$ so it is sufficient to study the determinant of
\begin{equation*}
    B_-(\KE,\delta_L):=\begin{pmatrix}
-(\mu_E+\sigma_E) & 0 & 0 & 0 & \phi & 0 \\[2mm]
\sigma_E & -(\mu_L+\sigma_L+2\delta_L L_{-}^{\ast\ast}) & 0 & 0 & 0 & 0 \\[2mm]
0 & \sigma_L & -(\mu_P+\sigma_P) & 0 & 0 & 0 \\[2mm]
0 & 0 & r\sigma_P
&
-\dfrac{M_{w,-}^{\ast\ast}}{H}-\mu_F
& 0
&
-\dfrac{\gamma F_{w,-}^{\ast\ast}}{H^2}
\\[3mm]
0 & 0 & 0
&
\dfrac{M_{w,-}^{\ast\ast}}{H}
&
-\mu_F
&
\dfrac{\gamma F_{w,-}^{\ast\ast}}{H^2}
\\[3mm]
0 & 0 & (1-r)\sigma_P & 0 & 0 & -\mu_M
\end{pmatrix},
\end{equation*}
where $H=\gamma+\zeta M_w^{\ast\ast}$, and we have used $M_s^{\ast\ast}=0$. The determinant of $B_-(\KE,\delta_L)$ can be computed via cofactor expansion:
\begin{equation}\label{eq:determinant_of_B-}
\begin{split}
    \det B_-(\KE,\delta_L)=\frac{1}{H^2}&\Big(
H\mu_F(M^{\ast\ast}_{w,-}+H\mu_F)\mu_M
(\mu_E+\sigma_E)
(2L^{\ast\ast}_{-}\delta_L+\mu_L+\sigma_L)
(\mu_P+\sigma_P)\\
&
-(1-r)\gamma F^{\ast\ast}_{u,-}\mu_F\sigma_E\sigma_L\sigma_P\phi
-HM^{\ast\ast}_{w,-}r\mu_M\sigma_E\sigma_L\sigma_P\phi
\Big).
\end{split}
\end{equation}
In Appendix \ref{app:ext_proof}, we show that
\begin{equation}\label{eq:limit_of_determinant_of_B-}
    \lim_{\substack{K_E\to\infty\\\delta_L\to 0}}\det B_-(K_E,\delta_L)=-
\frac{
(\Rq-1)\mu_M\mu_F^2(1+\zeta\mu_F)
(\mu_E+\sigma_E)(\mu_L+\sigma_L)(\mu_P+\sigma_P)
}{
\Rq(1+\zeta\mu_F)-1
},
\end{equation}
which is negative provided $\Rq>1$. Since the determinants of $A_-$ and $B_-$ differ only by a positive factor, \eqref{eq:negative_determinant} follows.
\end{proof}
 
\noindent
Theorems~\ref{thm:positive_equilibria} and~\ref{thm:stable_equilibria}
together confirm that the dynamical structure suggested by
Theorems~\ref{thm:LAStrivial} and~\ref{thm:GAStrivial} is consistent
with a classical Allee effect.  In the parameter regime described
above, the MFE $\mathcal E_0$ is locally
asymptotically stable while the model simultaneously admits two
positive equilibria: a larger, locally asymptotically stable
equilibrium $X_+^{\ast\ast}$ and a smaller, unstable equilibrium
$X_-^{\ast\ast}$.  Consequently, sufficiently small mosquito
populations are attracted toward extinction, whereas sufficiently
large populations persist near the stable positive equilibrium.
 
The larger equilibrium $X_+^{\ast\ast}$ is referred to hereafter as
the \emph{natural equilibrium}, since it represents the stable,
persistent mosquito population expected to arise in the absence of
SIT or any other control intervention.  The smaller equilibrium
$X_-^{\ast\ast}$ is referred to as the \emph{Allee equilibrium}:
its instability, combined with the local asymptotic stability of
both $\mathcal E_0$ and $X_+^{\ast\ast}$, produces precisely the
threshold behavior characteristic of an Allee effect---initial
conditions below $X_-^{\ast\ast}$ are drawn toward extinction, while
those above it are drawn toward the natural equilibrium.  The
existence and stability structure of $\mathcal E_0$,
$X_-^{\ast\ast}$, and $X_+^{\ast\ast}$ therefore provide rigorous
confirmation that the SIT-free model exhibits mate-finding Allee
dynamics. Figure \ref{fig:R0q_bifurcation} shows the equilibria $X_\pm^{\ast\ast}$ and the bifurcation which gives rise to them.

\begin{figure}
    \centering
    \includegraphics[width=0.6\linewidth]{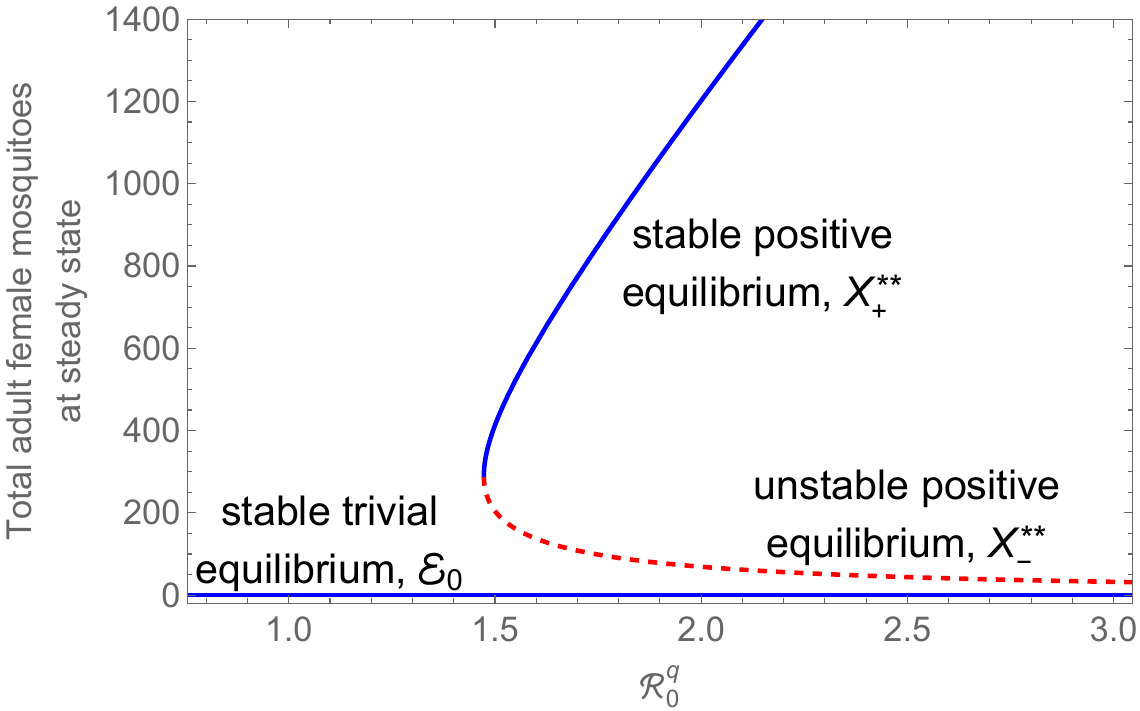}
    \caption{The three nonnegative equilibria of the model \eqref{eq:model_I} when $S(t)\equiv 0$, including the trivial MFE $\mathcal E_0$ which is always stable, the nontrivial Allee equilibrium $X_{-}^{\ast\ast}$, and the natural equilibrium $X_{+}^{\ast\ast}$. Note that while Theorem \ref{thm:positive_equilibria} shows that the nontrivial equilibria may exist when $\Rq>1$, that condition alone is not sufficient. Indeed, the bifurcation shown here occurs near $\Rq\approx 1.5$. This plot uses parameter values as in Table \ref{tab:parameter_values} except for $\phi$ whose value varies in this plot so the value of $\Rq$ is as indicated on the horizontal axis. Note that the adult female mosquito population size shown in this plot is far lower than would normally be expected in an area where mosquito-borne disease is endemic. This is because the expected value of $\Rq$ in an endemic area is far larger than it is near the bifurcation point, as in this plot.}
    \label{fig:R0q_bifurcation}
\end{figure}
 
The condition $\Rq > 1$ is biologically mild and expected to hold
under all realistic parameter values.  The maturation rates
$\sigma_E$, $\sigma_L$, and $\sigma_P$ typically exceed the
corresponding natural mortality rates $\mu_E$, $\mu_L$, and $\mu_P$;
the post-emergence refractory period $\zeta$ is short relative to
the expected adult female lifespan $1/\mu_F$; and the female sex
ratio $r$ is close to $1/2$~\cite{MazKidMyaKwe2019}.  Under these
conditions, each survival fraction $\sigma_x/(\sigma_x+\mu_x)$
exceeds $1/2$ and $(1+\zeta\mu_F)^{-1}$ is close to one, yielding
the rough lower bound
\begin{equation}\label{eq:Rq_lower_bound}
    \Rq \;\geq\; \frac{1}{32}\,\frac{\phi}{\mu_F}.
\end{equation}
Since $\phi/\mu_F$ represents the expected lifetime egg production
of a mated adult female---typically on the order of several hundred
eggs for the relevant mosquito
species~\cite{AfrGitYan2012,TakKloCha1998}---the inequality $\Rq
\gg 1$ holds for all biologically realistic parameter values.  The
SIT-free system therefore lies well within the parameter regime in
which positive mosquito equilibria exist and the Allee structure
described above is present.

\subsection{Global Asymptotic Stability of the natural $X_+^{\ast\ast}$ Equilibrium}
Consider the special case of the SIT-free model with $\zeta=\delta_L=0$,
given by
\begin{equation}\label{eq:reduced}
\left\{
\begin{array}{lcl}
\dot E      &=& \displaystyle \phi\left(1-\frac{E}{\KE}\right)F_{mw}-k_1E,\\[4pt]
\dot L      &=& \displaystyle \sigma_EE-k_2L,\\[4pt]
\dot P      &=& \displaystyle \sigma_LL-k_3P,\\[4pt]
\dot F_u    &=& \displaystyle r\sigma_PP-\frac{M_w}{\gamma}F_u-\mu_FF_u,\\[4pt]
\dot F_{mw} &=& \displaystyle \frac{M_w}{\gamma}F_u-\mu_FF_{mw},\\[4pt]
\dot M_w    &=& \displaystyle (1-r)\sigma_PP-\mu_M M_w,
\end{array}
\right.
\end{equation}
where $\zeta=0$ gives bilinear (mass-action) mating and $\delta_L=0$ removes
density-dependent larval mortality. This reduced model admits a tractable
Lyapunov-based basin analysis while preserving the bistable SIT-free structure
established for model~\eqref{eq:model_I} in
Theorems~\ref{thm:LAStrivial}--\ref{thm:stable_equilibria}. All remaining
biological parameters are positive, and we recall from
Table~\ref{tab:model_I_parameters} that the sex-ratio parameter satisfies
$0<r<1$. It is convenient to write $C(M_w):=M_w/\gamma$ and to define
$k_1:=\sigma_E+\mu_E$, $k_2:=\sigma_L+\mu_L$, and $k_3:=\sigma_P+\mu_P$.

With $\delta_L=0$, the $L$-component of any positive equilibrium
of~\eqref{eq:reduced} satisfies a quadratic equation, obtained by specializing
the equilibrium analysis of Theorem~\ref{thm:positive_equilibria}; this is the
$\delta_L\to0$ limiting quadratic $p_0$ of that proof. Consequently, when
$\Rq>1$ and $\KE>\KE^{\ast}$, where $\KE^{\ast}$ is the threshold of
Theorem~\ref{thm:positive_equilibria}, model~\eqref{eq:reduced} has exactly two
positive equilibria: the natural equilibrium $X_+^{**}$ and the Allee
equilibrium $X_-^{**}$. The second condition $\delta_L<\delta_L^{\ast}$ of
Theorem~\ref{thm:positive_equilibria} holds automatically here, since
$\delta_L=0$; the condition $\Rq>1$ alone is necessary but not sufficient, as it
fixes only the sign structure of $p_0$, whereas $\KE>\KE^{\ast}$ is required for
the two roots to be real and positive. For the baseline values in
Table~\ref{tab:parameter_values} this threshold is small, so the condition is
amply satisfied. Throughout this subsection we assume $\Rq>1$ and
$\KE>\KE^{\ast}$, so that both positive equilibria exist.

Let $(E_+^{\ast\ast},L_+^{\ast\ast},P_+^{\ast\ast},F_{u,+}^{\ast\ast},F_{mw,+}^{\ast\ast},M_{w,+}^{\ast\ast})$ denote the components
of the natural equilibrium $X_+^{**}$. These satisfy
\begin{equation}\label{eq:equil}
\begin{gathered}
 \phi(1-\rho)F_{mw,+}^{\ast\ast}=k_1E_+^{\ast\ast},\qquad
 \sigma_EE_+^{\ast\ast}=k_2L_+^{\ast\ast},\qquad
 \sigma_LL_+^{\ast\ast}=k_3P_+^{\ast\ast},\\[2pt]
 r\sigma_PP_+^{\ast\ast}=k_5F_{u,+}^{\ast\ast},\qquad
 k_4F_{u,+}^{\ast\ast}=\mu_FF_{mw,+}^{\ast\ast},\qquad
 (1-r)\sigma_PP_+^{\ast\ast}=\mu_M M_{w,+}^{\ast\ast},
\end{gathered}
\end{equation}
where it is further convenient to set $k_4:=M_{w,+}^{\ast\ast}/\gamma$,
$k_5:=k_4+\mu_F$, and $\rho:=E_+^{\ast\ast}/\KE\in(0,1)$.

Let
\[
 E_{\max}:=\KE,\qquad
 L_{\max}:=\frac{\sigma_E\KE}{k_2},\qquad
 P_{\max}:=\frac{\sigma_LL_{\max}}{k_3},
\]
\[
\begin{gathered}
 M_{w,\max}:=\frac{(1-r)\sigma_PP_{\max}}{\mu_M},\qquad
 C_{\max}:=\frac{M_{w,\max}}{\gamma},\\[2pt]
 F_{u,\max}:=\frac{r\sigma_PP_{\max}}{\mu_F},\qquad
 F_{mw,\max}:=\frac{C_{\max}F_{u,\max}}{\mu_F}.
\end{gathered}
\]
The assumption $0<r<1$ ensures that the bounds involving $r$ and $1-r$ are
positive. Model~\eqref{eq:reduced} is studied in the six-dimensional
biologically feasible region
\begin{equation}\label{eq:Omegahat}
\begin{split}
\widehat\Omega=\bigl\{(E,L,P,F_u,F_{mw},M_w)\in\Omega:
&\ E\le E_{\max},\ L\le L_{\max},\ P\le P_{\max},\\
&\ F_u\le F_{u,\max},\ F_{mw}\le F_{mw,\max},\ M_w\le M_{w,\max}\bigr\}.
\end{split}
\end{equation}

\begin{proposition}\label{prop:OmegaHat}
The region $\widehat\Omega$ is positively invariant with respect to
model~\eqref{eq:reduced}.
\end{proposition}

\begin{proof}
The vector field associated with~\eqref{eq:reduced} is quasi-positive on the
coordinate hyperplanes, since
\[
\begin{gathered}
\dot E|_{E=0}=\phi F_{mw}\ge0,\qquad
\dot L|_{L=0}=\sigma_EE\ge0,\qquad
\dot P|_{P=0}=\sigma_LL\ge0,\\[2pt]
\dot F_u|_{F_u=0}=r\sigma_PP\ge0,\qquad
\dot F_{mw}|_{F_{mw}=0}=C(M_w)F_u\ge0,\\[2pt]
\dot M_w|_{M_w=0}=(1-r)\sigma_PP\ge0.
\end{gathered}
\]
Hence solutions that start in $\mathbb{R}_+^6$ remain in $\mathbb{R}_+^6$. On
the upper bounding faces of $\widehat\Omega$, using the bounds on the
preceding compartments together with $0\le C(M_w)=M_w/\gamma\le C_{\max}$, it
follows that
\begin{align*}
\dot E|_{E=\KE}&=-k_1\KE\le0,\\
\dot L|_{L=L_{\max}}&=\sigma_EE-k_2L_{\max}
 \le\sigma_E\KE-k_2L_{\max}=0,\\
\dot P|_{P=P_{\max}}&=\sigma_LL-k_3P_{\max}
 \le\sigma_LL_{\max}-k_3P_{\max}=0,\\
\dot M_w|_{M_w=M_{w,\max}}&=(1-r)\sigma_PP-\mu_M M_{w,\max}
 \le(1-r)\sigma_PP_{\max}-\mu_M M_{w,\max}=0,\\
\dot F_u|_{F_u=F_{u,\max}}&=r\sigma_PP-C(M_w)F_{u,\max}-\mu_FF_{u,\max}
   \le r\sigma_PP_{\max}-\mu_FF_{u,\max}=0,\\
\dot F_{mw}|_{F_{mw}=F_{mw,\max}}&=C(M_w)F_u-\mu_FF_{mw,\max}
   \le C_{\max}F_{u,\max}-\mu_FF_{mw,\max}=0.
\end{align*}
Thus the vector field points into $\widehat\Omega$ on each of its bounding
faces, and $\widehat\Omega$ is positively invariant. Since the vector field
in~\eqref{eq:reduced} is polynomial, it is locally Lipschitz on $\mathbb{R}^6$;
together with the boundedness of $\widehat\Omega$, this guarantees that every
solution with initial data in $\widehat\Omega$ exists and remains in
$\widehat\Omega$ for all $t\ge0$.
\end{proof}

For the stability analysis, all ratios are taken relative to the natural
equilibrium. Thus, for $X\in\widehat\Omega$, set
\[
\begin{gathered}
  \eta_E:=\frac{E}{E_+^{\ast\ast}},\qquad
  \eta_L:=\frac{L}{L_+^{\ast\ast}},\qquad
  \eta_P:=\frac{P}{P_+^{\ast\ast}},\\[2pt]
  \eta_{F_u}:=\frac{F_u}{F_{u,+}^{\ast\ast}},\qquad
  \eta_{F_{mw}}:=\frac{F_{mw}}{F_{mw,+}^{\ast\ast}},\qquad
  \eta_{M_w}:=\frac{M_w}{M_{w,+}^{\ast\ast}}.
\end{gathered}
\]
Define
\begin{equation}\label{eq:S5QCE}
\begin{aligned}
  S_5&:=5-\frac{\eta_{F_{mw}}}{\eta_E}-\frac{\eta_E}{\eta_L}
        -\frac{\eta_L}{\eta_P}-\frac{\eta_P}{\eta_{F_u}}
        -\frac{\eta_{F_u}}{\eta_{F_{mw}}},\\[2pt]
  Q&:=1-\frac{\eta_P}{\eta_{M_w}}-\frac{k_4}{k_5}(1-\eta_{F_u})
     +\eta_{F_u}\left(\frac{1}{\eta_{F_{mw}}}-1\right),\\[2pt]
  \mathcal{C}_E&:=\frac{\rho\,\eta_{F_{mw}}}{1-\rho}\,
  \frac{(\eta_E-1)^2}{\eta_E},\\[2pt]
  B&:=\mathcal{C}_E-S_5.
\end{aligned}
\end{equation}
Because the equilibrium transfer rates
\begin{equation}\label{eq:transfer}
\begin{gathered}
  f_1^{**}:=k_1E_+^{\ast\ast},\qquad
  f_2^{**}:=\sigma_EE_+^{\ast\ast},\qquad
  f_3^{**}:=\sigma_LL_+^{\ast\ast},\\
  f_4^{**}:=r\sigma_PP_+^{\ast\ast},\qquad
  f_5^{**}:=k_4F_{u,+}^{\ast\ast},\qquad
  f_6^{**}:=(1-r)\sigma_PP_+^{\ast\ast}
\end{gathered}
\end{equation}
are strictly positive (recall $0<r<1$), we may define the Goh--Volterra-type
Lyapunov function
\begin{equation}\label{eq:LyapF}
\begin{aligned}
  \mathcal{F}
  &=\frac{1}{f_1^{**}}\Bigl(E-E_+^{\ast\ast}-E_+^{\ast\ast}\ln\frac{E}{E_+^{\ast\ast}}\Bigr)
   +\frac{1}{f_2^{**}}\Bigl(L-L_+^{\ast\ast}-L_+^{\ast\ast}\ln\frac{L}{L_+^{\ast\ast}}\Bigr)\\
  &\quad
   +\frac{1}{f_3^{**}}\Bigl(P-P_+^{\ast\ast}-P_+^{\ast\ast}\ln\frac{P}{P_+^{\ast\ast}}\Bigr)
   +\frac{1}{f_4^{**}}\Bigl(F_u-F_{u,+}^{\ast\ast}-F_{u,+}^{\ast\ast}\ln\frac{F_u}{F_{u,+}^{\ast\ast}}\Bigr)\\
  &\quad
   +\frac{1}{f_5^{**}}\Bigl(F_{mw}-F_{mw,+}^{\ast\ast}
      -F_{mw,+}^{\ast\ast}\ln\frac{F_{mw}}{F_{mw,+}^{\ast\ast}}\Bigr)
   +\frac{1}{f_6^{**}}\Bigl(M_w-M_{w,+}^{\ast\ast}-M_{w,+}^{\ast\ast}\ln\frac{M_w}{M_{w,+}^{\ast\ast}}\Bigr).
\end{aligned}
\end{equation}
For $c>0$, define the associated sublevel set
\begin{equation}\label{eq:Dc}
  \mathcal{D}_c:=\{X\in\widehat\Omega:\mathcal{F}(X)\le c\}.
\end{equation}

The following two lemmas record the elementary properties of $\mathcal{F}$ and
the exact form of its derivative along solutions; both are used in the proof
of the main theorem.

\begin{lemma}\label{lem:Fcompact}
The function $\mathcal{F}$ is non-negative on $\widehat\Omega$ and vanishes
only at $X_+^{**}$. Moreover, for every $c>0$ the sublevel set
$\mathcal{D}_c$ is a compact subset of $\widehat\Omega$.
\end{lemma}

\begin{proof}
For each scalar $x>0$ the entropy term satisfies
\begin{equation}\label{eq:entropy}
  x-x^{**}-x^{**}\ln(x/x^{**})\ge0,
\end{equation}
with equality if and only if $x=x^{**}$; hence $\mathcal{F}\ge0$ on
$\widehat\Omega$, with equality only at $X_+^{**}$. Each entropy term is
also non-negative individually, so $\mathcal{F}(X)\le c$ forces
\[
  x-x^{**}-x^{**}\ln(x/x^{**})\le c\,f_i^{**}
\]
for the corresponding coordinate $x$ and rate $f_i^{**}$. Since
\[
  x-x^{**}-x^{**}\ln(x/x^{**})\to+\infty
  \qquad\text{as }x\to0^+,
\]
each coordinate is bounded below on $\mathcal{D}_c$ by a positive constant
depending only on $c$. Hence any convergent sequence in $\mathcal{D}_c$ has a
limit with all coordinates strictly positive; since $\widehat\Omega$ is closed
and $\mathcal{F}$ is continuous on the open positive box, the limit again lies
in $\mathcal{D}_c$. Therefore $\mathcal{D}_c$ is a closed and bounded subset of
$\mathbb{R}^6$ contained in $(0,\infty)^6$, hence a compact subset of
$\widehat\Omega$.
\end{proof}

\begin{lemma}\label{lem:Fdot}
Along solutions of model~\eqref{eq:reduced} in $\widehat\Omega$, the
derivative of $\mathcal{F}$ satisfies the exact identity
\begin{equation}\label{eq:Fdot}
  \dot{\mathcal{F}}=S_5-\mathcal{C}_E+(1-\eta_{M_w})Q.
\end{equation}
Consequently $S_5\le0$, with equality if and only if
\begin{equation}\label{eq:S5eq}
  \frac{\eta_{F_{mw}}}{\eta_E}
  =\frac{\eta_E}{\eta_L}
  =\frac{\eta_L}{\eta_P}
  =\frac{\eta_P}{\eta_{F_u}}
  =\frac{\eta_{F_u}}{\eta_{F_{mw}}}=1,
\end{equation}
$\mathcal{C}_E\ge0$, and therefore $B=\mathcal{C}_E-S_5\ge0$.
\end{lemma}

\begin{proof}
Using
\[
  \frac{d}{dt}\left(x-x^{**}-x^{**}\ln\frac{x}{x^{**}}\right)
  =\left(1-\frac{x^{**}}{x}\right)\dot x,
\]
substitution of the right-hand sides of~\eqref{eq:reduced} into the time
derivative of~\eqref{eq:LyapF}, followed by use of the equilibrium
identities~\eqref{eq:equil} and the relation
\[
  \frac{M_w/\gamma}{k_4}=\frac{M_w}{M_{w,+}^{\ast\ast}}=\eta_{M_w},
\]
gives
\begin{align}\label{eq:Fdot-substitution}
\dot{\mathcal{F}}
&=\left(1-\frac{1}{\eta_E}\right)
   \left[\frac{\eta_{F_{mw}}(1-\rho\eta_E)}{1-\rho}-\eta_E\right]
 +\left(1-\frac{1}{\eta_L}\right)(\eta_E-\eta_L) \notag\\
&\quad
 +\left(1-\frac{1}{\eta_P}\right)(\eta_L-\eta_P)
 +\left(1-\frac{1}{\eta_{F_u}}\right)
   \left[\eta_P-\frac{k_4\eta_{M_w}+\mu_F}{k_5}\eta_{F_u}\right] \notag\\
&\quad
 +\left(1-\frac{1}{\eta_{F_{mw}}}\right)
   (\eta_{M_w}\eta_{F_u}-\eta_{F_{mw}})
 +\left(1-\frac{1}{\eta_{M_w}}\right)(\eta_P-\eta_{M_w}).
\end{align}
The first term in~\eqref{eq:Fdot-substitution} can be written as
\begin{equation}\label{eq:Eterm}
\begin{aligned}
\left(1-\frac{1}{\eta_E}\right)
   \left[\frac{\eta_{F_{mw}}(1-\rho\eta_E)}{1-\rho}-\eta_E\right]
&=1+\eta_{F_{mw}}-\eta_E-\frac{\eta_{F_{mw}}}{\eta_E}
  -\frac{\rho\eta_{F_{mw}}}{1-\rho}\frac{(\eta_E-1)^2}{\eta_E}\\
&=1+\eta_{F_{mw}}-\eta_E-\frac{\eta_{F_{mw}}}{\eta_E}-\mathcal{C}_E.
\end{aligned}
\end{equation}
Similarly,
\begin{align*}
\left(1-\frac{1}{\eta_L}\right)(\eta_E-\eta_L)
&=1+\eta_E-\eta_L-\frac{\eta_E}{\eta_L},\\
\left(1-\frac{1}{\eta_P}\right)(\eta_L-\eta_P)
&=1+\eta_L-\eta_P-\frac{\eta_L}{\eta_P},\\
\left(1-\frac{1}{\eta_{F_u}}\right)
   \left[\eta_P-\frac{k_4\eta_{M_w}+\mu_F}{k_5}\eta_{F_u}\right]
&=1+\eta_P-\eta_{F_u}-\frac{\eta_P}{\eta_{F_u}}
  -\frac{k_4}{k_5}(\eta_{F_u}-1)(\eta_{M_w}-1),\\
\left(1-\frac{1}{\eta_{F_{mw}}}\right)(\eta_{M_w}\eta_{F_u}-\eta_{F_{mw}})
&=1+\eta_{M_w}\eta_{F_u}-\eta_{F_{mw}}
  -\frac{\eta_{M_w}\eta_{F_u}}{\eta_{F_{mw}}},\\
\left(1-\frac{1}{\eta_{M_w}}\right)(\eta_P-\eta_{M_w})
&=1+\eta_P-\eta_{M_w}-\frac{\eta_P}{\eta_{M_w}},
\end{align*}
where in the fourth line we have used
$\dfrac{k_4\eta_{M_w}+\mu_F}{k_5}=1+\dfrac{k_4}{k_5}(\eta_{M_w}-1)$, valid
since $k_5=k_4+\mu_F$. Summing the six contributions, the constant terms add to
$6$ and the linear terms telescope, leaving, after subtracting $S_5-\mathcal{C}_E$,
the residual
\[
  1+\eta_P-\eta_{F_u}+\eta_{M_w}\eta_{F_u}-\eta_{M_w}
  -\frac{\eta_{M_w}\eta_{F_u}}{\eta_{F_{mw}}}
  -\frac{\eta_P}{\eta_{M_w}}
  -\frac{k_4}{k_5}(\eta_{F_u}-1)(\eta_{M_w}-1)
  +\frac{\eta_{F_u}}{\eta_{F_{mw}}},
\]
which equals $(1-\eta_{M_w})Q$. This yields the exact
identity~\eqref{eq:Fdot}.

The five ratios in $S_5$ have product equal to one by telescoping; hence the
arithmetic--geometric mean inequality gives $S_5\le0$, with equality if and
only if~\eqref{eq:S5eq} holds. Finally $\mathcal{C}_E\ge0$ on
$\widehat\Omega$, being a non-negative multiple of the perfect square
$(\eta_E-1)^2$, and therefore $B=\mathcal{C}_E-S_5\ge0$.
\end{proof}

\begin{theorem}\label{thm:certified-basin}
Consider the reduced model~\eqref{eq:reduced} with $\zeta=\delta_L=0$,
$\Rq>1$, and $\KE>\KE^{\ast}$, and let $\mathcal{F}$ and $\mathcal{D}_c$ be
defined as in~\eqref{eq:LyapF}--\eqref{eq:Dc}. Suppose there exist a level
$c>0$ and a constant $\theta\in[0,1)$ such that
\begin{equation}\label{eq:cond}
  (1-\eta_{M_w})Q\le \theta\bigl(\mathcal{C}_E-S_5\bigr)
  \qquad\text{for all }X\in\mathcal{D}_c.
\end{equation}
Then the sublevel set $\mathcal{D}_c$ is positively invariant, the natural
equilibrium $X_+^{**}$ is asymptotically stable relative to $\mathcal{D}_c$,
and every solution with initial condition in $\mathcal{D}_c$ converges to
$X_+^{**}$. In particular, $\mathcal{D}_c$ is contained in the basin of
attraction of $X_+^{**}$.
\end{theorem}

\begin{proof}
Let $X(t)$ be the solution with $X(0)\in\mathcal{D}_c$. We first establish
that the solution remains in $\widehat\Omega$ for as long as it exists.
Since $\mathcal{D}_c\subset\widehat\Omega\subset\widehat\Omega$,
Proposition~\ref{prop:OmegaHat} implies $X(t)\in\widehat\Omega$ while the
solution exists; in particular $0\le E(t)\le\KE$ and
$0\le M_w(t)\le M_{w,\max}$. The bound $E\le\KE$ ensures $1-E/\KE\ge0$, so the
production term in the $E$-equation is non-negative, and the bound
$M_w\le M_{w,\max}$ controls the mating loss in the $F_u$-equation.
Consequently the right-hand sides of~\eqref{eq:reduced} satisfy the
differential inequalities
\[
  \dot E\ge -k_1E,\quad \dot L\ge -k_2L,\quad \dot P\ge -k_3P,\quad
  \dot F_u\ge -\left(\frac{M_{w,\max}}{\gamma}+\mu_F\right)F_u,\quad
  \dot F_{mw}\ge -\mu_FF_{mw},\quad \dot M_w\ge -\mu_MM_w,
\]
so that, by comparison, positive initial data remain strictly positive.
Hence $X(t)\in\widehat\Omega$ for as long as the solution exists.

By Lemma~\ref{lem:Fdot} and the hypothesis~\eqref{eq:cond}, as long as
$X(t)\in\mathcal{D}_c$,
\begin{equation}\label{eq:dFbound}
\begin{aligned}
  \dot{\mathcal{F}}
  &=S_5-\mathcal{C}_E+(1-\eta_{M_w})Q
   \le S_5-\mathcal{C}_E+\theta(\mathcal{C}_E-S_5)\\
  &=-(1-\theta)(\mathcal{C}_E-S_5)
   =-(1-\theta)B\le0,
\end{aligned}
\end{equation}
since $\theta<1$ and $B\ge0$ by Lemma~\ref{lem:Fdot}. Thus $\mathcal{F}$
cannot increase while the solution lies in $\mathcal{D}_c$. A standard
first-exit argument now shows that $\mathcal{D}_c$ is positively invariant: if
the solution were to leave $\mathcal{D}_c$, then, up to its first exit time, it
would remain in $\mathcal{D}_c$, on which $\mathcal{F}$ is non-increasing;
hence $\mathcal{F}$ could not cross from a value at most $c$ to a value
exceeding $c$, a contradiction. Therefore $\mathcal{F}(X(t))\le c$ for all
$t\ge0$. Because $\mathcal{D}_c$ is compact by Lemma~\ref{lem:Fcompact}, the
solution remains in a compact set; together with the local Lipschitz
continuity of the (polynomial) vector field, this rules out finite-time
blow-up and guarantees existence for all $t\ge0$, with a nonempty
$\omega$-limit set contained in $\mathcal{D}_c$.

By the Lyapunov--LaSalle invariance principle, the solution approaches the
largest invariant subset of $\{X\in\mathcal{D}_c:\dot{\mathcal{F}}=0\}$.
From~\eqref{eq:dFbound}, $\dot{\mathcal{F}}=0$ forces $B=0$, that is,
$\mathcal{C}_E=0$ and $S_5=0$. By Lemma~\ref{lem:Fdot}, $S_5=0$ gives, via
\eqref{eq:S5eq},
\[
  \eta_E=\eta_L=\eta_P=\eta_{F_u}=\eta_{F_{mw}}=:\lambda,
\]
while $\mathcal{C}_E=0$ gives $\eta_E=1$, whence $\lambda=1$. Thus
\[
  \eta_E=\eta_L=\eta_P=\eta_{F_u}=\eta_{F_{mw}}=1.
\]
Along any complete trajectory contained in this equality set, $F_{mw}\equiv
F_{mw,+}^{\ast\ast}$ is constant, so $\dot F_{mw}=0$. Using the $F_{mw}$-equation
in~\eqref{eq:reduced} and the equilibrium identity
$k_4F_{u,+}^{\ast\ast}=\mu_FF_{mw,+}^{\ast\ast}$ gives
\[
  0=\dot F_{mw}
  =k_4\eta_{M_w}F_{u,+}^{\ast\ast}-\mu_FF_{mw,+}^{\ast\ast}
  =k_4F_{u,+}^{\ast\ast}(\eta_{M_w}-1).
\]
Since $k_4F_{u,+}^{\ast\ast}>0$, this forces $\eta_{M_w}=1$. The largest invariant
subset of $\{X\in\mathcal{D}_c:\dot{\mathcal{F}}=0\}$ is therefore the
singleton $\{X_+^{**}\}$, and LaSalle's invariance principle gives
$X(t)\to X_+^{**}$ for every solution with $X(0)\in\mathcal{D}_c$.

Finally, we verify Lyapunov stability of $X_+^{**}$ relative to
$\mathcal{D}_c$. Let $\varepsilon>0$. If
$\{X\in\mathcal{D}_c:\lvert X-X_+^{**}\rvert\ge\varepsilon\}=\varnothing$, then
every point of $\mathcal{D}_c$ already lies within $\varepsilon$ of $X_+^{**}$
and the stability requirement is met trivially. Otherwise, set
\[
  m_\varepsilon:=\min\bigl\{\mathcal{F}(X):X\in\mathcal{D}_c,\
    \lvert X-X_+^{**}\rvert\ge\varepsilon\bigr\}.
\]
The minimum is attained and is strictly positive, because the set is compact
(a closed subset of the compact set $\mathcal{D}_c$) and $\mathcal{F}$ is
continuous and vanishes only at $X_+^{**}$. By continuity of $\mathcal{F}$ and
$\mathcal{F}(X_+^{**})=0$, there is $\delta\in(0,\varepsilon)$ such that
$\lvert X(0)-X_+^{**}\rvert<\delta$ implies
$\mathcal{F}(X(0))<m_\varepsilon$. Since $\mathcal{F}$ is non-increasing along
solutions in $\mathcal{D}_c$, we have $\mathcal{F}(X(t))<m_\varepsilon$ for all
$t\ge0$, so the trajectory cannot reach
$\{\lvert X-X_+^{**}\rvert\ge\varepsilon\}$; hence
$\lvert X(t)-X_+^{**}\rvert<\varepsilon$ for all $t\ge0$. This establishes
Lyapunov stability, and together with the convergence above, asymptotic
stability of $X_+^{**}$ relative to $\mathcal{D}_c$. Consequently
$\mathcal{D}_c$ is contained in the basin of attraction of $X_+^{**}$.
\end{proof}

\begin{remark}\label{rem:Rloc-singularity}
The derivative identity~\eqref{eq:Fdot} suggests the formal local Lyapunov
ratio
\begin{equation}\label{eq:Rloc}
  \mathcal{R}_{\mathrm{loc}}(X):=
  \frac{(1-\eta_{M_w})Q}{B(X)}
  =\frac{(1-\eta_{M_w})Q}{\mathcal{C}_E-S_5},\qquad B(X)>0.
\end{equation}
At points where $B>0$, identity~\eqref{eq:Fdot} can be written as
\begin{equation}\label{eq:FdotR}
  \dot{\mathcal{F}}
  =B\bigl(\mathcal{R}_{\mathrm{loc}}-1\bigr)
  =\bigl(\mathcal{C}_E-S_5\bigr)\bigl(\mathcal{R}_{\mathrm{loc}}-1\bigr).
\end{equation}
The restriction $B>0$ is necessary because the denominator can vanish. In
particular, at $X_+^{**}$ all ratios equal one, so $S_5=0$, $\mathcal{C}_E=0$,
and $(1-\eta_{M_w})Q=0$; thus $\mathcal{R}_{\mathrm{loc}}$ has the indeterminate
form $0/0$ at $X_+^{**}$. Since $X_+^{**}$ belongs to every sublevel set
$\mathcal{D}_c$, the hypothesis of Theorem~\ref{thm:certified-basin} cannot be
stated as $\mathcal{R}_{\mathrm{loc}}\le\theta$ on all of $\mathcal{D}_c$; this
is why the theorem uses the quotient-free condition~\eqref{eq:cond}. On the
singular set $B=0$ the quotient-free condition is harmless: if $B=0$, then
$\mathcal{C}_E=0$ and $S_5=0$, hence
$\eta_E=\eta_L=\eta_P=\eta_{F_u}=\eta_{F_{mw}}=1$ while $\eta_{M_w}$ is free, so
that
\[
  Q=1-\frac{1}{\eta_{M_w}},\qquad
  (1-\eta_{M_w})Q=-\frac{(\eta_{M_w}-1)^2}{\eta_{M_w}}\le0,
\]
and~\eqref{eq:cond} holds automatically for every $\theta\ge0$. Equivalently,
away from $B=0$, condition~\eqref{eq:cond} is exactly
$\mathcal{R}_{\mathrm{loc}}\le\theta$.
\end{remark}

\begin{remark}\label{rem:theta-strict}
The strict inequality $\theta<1$ in~\eqref{eq:cond} is essential. First,
$B(X_-^{**})>0$. Indeed, the equilibrium equations for $L$, $P$, and $M_w$
give
\[
  L=\frac{\sigma_E}{k_2}E,
  \qquad
  P=\frac{\sigma_L}{k_3}L,
  \qquad
  M_w=\frac{(1-r)\sigma_P}{\mu_M}P,
\]
and those for $F_u$ and $F_{mw}$ give
\[
  F_u=\frac{r\sigma_PP}{M_w/\gamma+\mu_F},
  \qquad
  F_{mw}=\frac{(M_w/\gamma)F_u}{\mu_F},
\]
so a positive equilibrium is determined by its $E$-coordinate. Since
$X_-^{**}\ne X_+^{**}$, we have $E_-^{**}\ne E_+^{\ast\ast}$, hence
$\eta_E(X_-^{**})\ne1$ and
\[
  \mathcal{C}_E(X_-^{**})
  =\frac{\rho\,\eta_{F_{mw}}(X_-^{**})}{1-\rho}\,
  \frac{\bigl(\eta_E(X_-^{**})-1\bigr)^2}{\eta_E(X_-^{**})}>0.
\]
As $S_5\le0$, this gives
$B(X_-^{**})=\mathcal{C}_E(X_-^{**})-S_5(X_-^{**})>0$. At the Allee
equilibrium $\dot{\mathcal{F}}=0$, so~\eqref{eq:FdotR} yields
$\mathcal{R}_{\mathrm{loc}}(X_-^{**})=1$. Hence no condition with $\theta<1$
can hold on a sublevel set containing $X_-^{**}$; a non-strict condition
($\theta=1$) would fail to exclude the Allee equilibrium and would not force
the LaSalle invariant set to reduce to $X_+^{**}$.
\end{remark}

\begin{remark}\label{rem:self-certifying}
Condition~\eqref{eq:cond} is self-certifying. If it holds on $\mathcal{D}_c$,
then~\eqref{eq:dFbound} gives $\dot{\mathcal{F}}\le-(1-\theta)B\le0$, so
$\mathcal{D}_c$ is positively invariant and, by
Theorem~\ref{thm:certified-basin}, is contained in the basin of attraction of
$X_+^{**}$. Thus no separate exclusion of the extinction basin is required: if
an admissible $\mathcal{D}_c$ contained an initial condition whose solution
converged to the mosquito-free equilibrium, this would contradict the
conclusion of the theorem. Consequently the Allee equilibrium, its stable set
$\mathcal{W}^s(X_-^{**})$, and the extinction basin all lie outside every
sublevel set on which~\eqref{eq:cond} holds with $\theta<1$. The supremal
certified sublevel set---equivalently, the maximal such set when the supremum
of admissible levels is attained---provides a constructive, conservative inner
estimate of the basin of attraction of $X_+^{**}$.
\end{remark}

\begin{remark}\label{rem:control-interpretation}
Theorem~\ref{thm:certified-basin} has a direct interpretation for the control
of the wild mosquito population. Because the mosquito-free state is locally
stable under bilinear mating, the reduced system is bistable: from any initial
condition in an admissible $\mathcal{D}_c$ the population is drawn to the
natural equilibrium $X_+^{**}$, so a transient knock-down is followed by full
recovery. If the post-release state remains in $\mathcal{D}_c$, recovery to
$X_+^{**}$ is therefore guaranteed. Any strategy aiming at guaranteed
elimination from such an initial condition must consequently drive the state
outside this conservative persistence certificate. Exiting $\mathcal{D}_c$ is,
however, necessary but not sufficient for elimination: the state may leave
$\mathcal{D}_c$ and still lie in the basin of $X_+^{**}$, since $\mathcal{D}_c$
is only an inner estimate of that basin. A separate argument is needed to show
that the post-release state lies in the extinction basin. The supremal
admissible sublevel set, computed from the quotient-free bound~\eqref{eq:cond},
furnishes an explicit conservative inner estimate of the persistence basin and
hence a quantitative lower bound on the suppression target for any control
program.
\end{remark}

The quotient $\mathcal{R}_{\mathrm{loc}}$ is useful as a state-dependent sign
indicator for $\dot{\mathcal{F}}$, provided it is evaluated only where $B>0$.
Its numerator measures the sign-indefinite coupling associated with the
wild-male mating term, whereas its denominator $B=\mathcal{C}_E-S_5$ measures
the combined stabilizing effect of density-dependent egg regulation and the
Goh--Volterra-type dissipation. Equation~\eqref{eq:FdotR} shows that, at points
where $B>0$, the inequality $\mathcal{R}_{\mathrm{loc}}<1$ is exactly the
condition under which the Lyapunov functional is decreasing. This quantity
should not be confused with the quick-search reproduction number $\Rq$ defined
in~\eqref{eq:Rlm}: the latter is a parameter threshold associated with the
existence of the positive equilibria, whereas $\mathcal{R}_{\mathrm{loc}}$ is a
state-dependent sign indicator.

\begin{figure}[ht]
\centering
\includegraphics[width=0.90\textwidth]{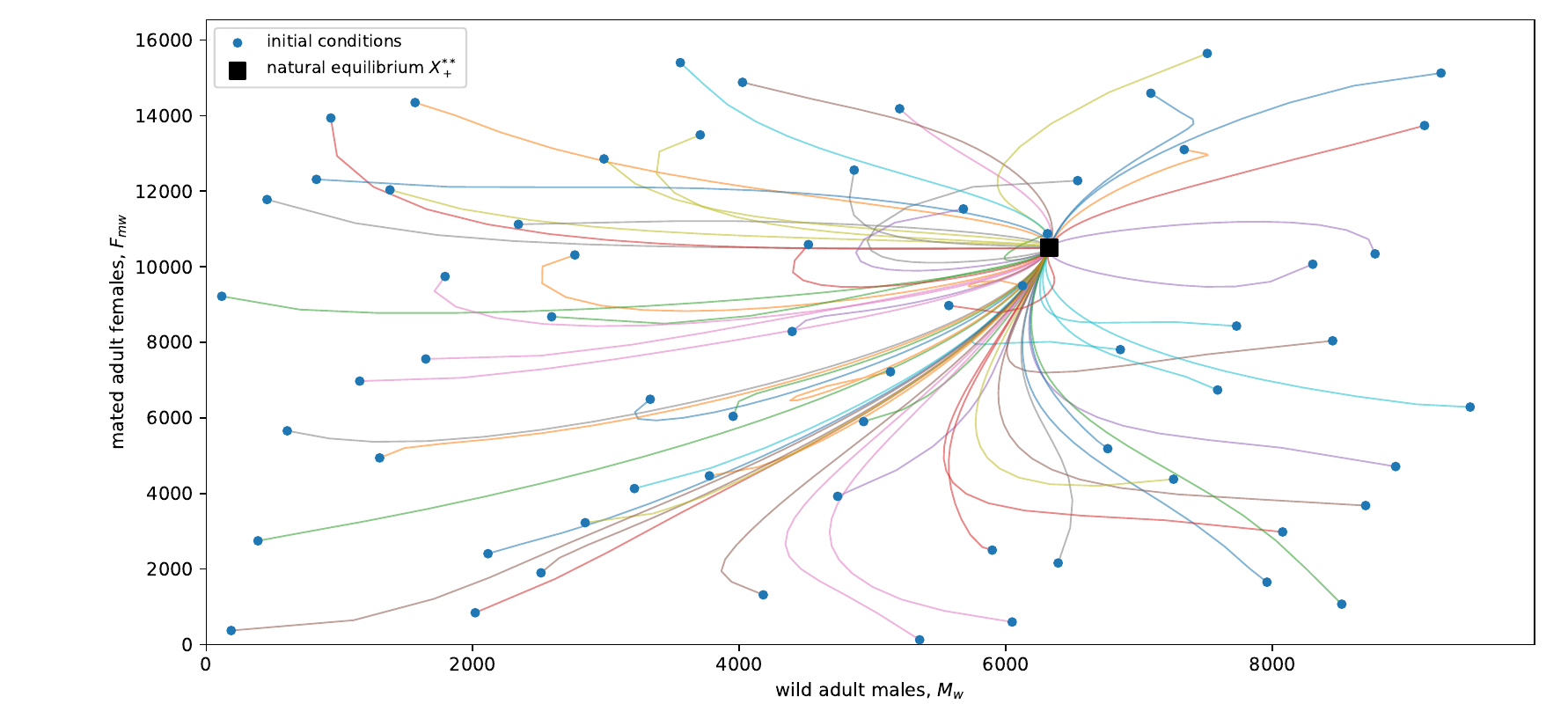}
\caption{Trajectories staring from 64 Sobol' sampled initial conditions all converging to $X_+^{\ast\ast}$ (projected into the $F_{mw}-M_w$ plane). Rigorous numerical calculation shows convergence to $X_+^{\ast\ast}$ for over $2^{19}$ initial conditions Sobol' sampled in the box with one vertex on $X_-^{\ast\ast}$ and the opposite vertex at $1.5\times X_+^{\ast\ast}$.}
\label{fig:Rloc}
\end{figure}

Using the baseline parameter values in Table~\ref{tab:parameter_values} and the
positive equilibria reported in Table~\ref{tab:equilibria}, a numerical
evaluation of the quotient-free condition~\eqref{eq:cond} on a selected
sublevel set $\mathcal{D}_{0.8c_0}$, with $c_0<\mathcal{F}(X_-^{**})\approx379.2$,
gives
\[
  \sup_{X\in\mathcal{D}_{0.8c_0}\cap\{A>0\}}
  \mathcal{R}_{\mathrm{loc}}(X)\approx0.99999999689<1,
\]
the singular set $B=0$ being covered separately as in
Remark~\ref{rem:Rloc-singularity}. Accordingly, if this numerical supremum is
certified as a uniform upper bound on $\mathcal{D}_{c_0}$---for example by
interval arithmetic, branch-and-bound, or another rigorous global-enclosure
method---then the hypothesis of Theorem~\ref{thm:certified-basin} is satisfied
on this sublevel set.

At the Allee equilibrium, we calculate $\mathcal{R}_{\mathrm{loc}}$ to be exactly $1$. While it is tempting to conclude from this that the set $\{\mathcal R_{\mathrm{loc}}<1\}$ coincides with $\mathcal D_{c_0}$, this is not the case; there are points in $\mathcal D_{c_0}$ where $\mathcal R_{\mathrm{loc}}>1$. It is therefore not possible to apply Theorem \ref{thm:certified-basin} on all of $\mathcal D_{c_0}$.


Theorem~\ref{thm:certified-basin} establishes convergence to the natural
equilibrium for the reduced SIT-free model on any Lyapunov sublevel set on
which the quotient-free condition~\eqref{eq:cond} holds. Numerical simulations
suggest that an analogous persistence-side stability statement remains true for
the SIT-free model~\eqref{eq:model_I} with the mate-search delay and larval
density-dependence retained. Such a statement cannot be posed on the whole
positive feasible region, because Theorem~\ref{thm:LAStrivial} shows that the
mosquito-free equilibrium is locally asymptotically stable; hence initial
conditions in the extinction basin must be excluded. Initial conditions on the
stable set of the Allee equilibrium must also be excluded, since they converge
to the Allee equilibrium rather than to the natural equilibrium.

Extensive numerical simulations suggest the following conjecture.

\begin{conjecture}\label{conj:SIT_free_GAS}
Consider the SIT-free model obtained from~\eqref{eq:model_I} by setting
$S(t)\equiv0$, with $\zeta>0$ and $\delta_L>0$. Suppose that $\Rq>1$ and
$\KE>\KE^{\ast}$, and let $X_-^{**}$ and $X_+^{**}$ denote, respectively, the
Allee equilibrium and the natural equilibrium of the SIT-free model. Define the
positive wild-population invariant subspace
\begin{equation*}
\Omega_{\mathrm{sf}}
=
\bigl\{(E,L,P,F_u,F_{mw},F_{ms},M_w,M_s)\in\Omega:
E,L,P,F_u,F_{mw},M_w>0,\ F_{ms}=M_s=0\bigr\}.
\end{equation*}
Let
\begin{equation*}
\mathcal{B}_0
:=
\bigl\{X_0\in\Omega_{\mathrm{sf}}:
\lim_{t\to\infty}X(t;X_0)=\mathcal E_0\bigr\}
\end{equation*}
denote the set of positive wild-population initial states whose solutions
converge to the mosquito-free equilibrium $\mathcal E_0$, and let
\begin{equation*}
\mathcal{W}^s_+(X_-^{**})
:=
\bigl\{X_0\in\Omega_{\mathrm{sf}}:
\lim_{t\to\infty}X(t;X_0)=X_-^{**}\bigr\}
\end{equation*}
denote the stable set of the Allee equilibrium in $\Omega_{\mathrm{sf}}$.
Finally, let $\mathcal{P}_+$ be the connected component containing $X_+^{**}$
of the forward-invariant set
\begin{equation*}
\Omega_{\mathrm{sf}}\setminus
\bigl(\mathcal{B}_0\cup\mathcal{W}^s_+(X_-^{**})\bigr).
\end{equation*}
Then the natural equilibrium $X_+^{**}$ is globally asymptotically stable
relative to $\mathcal{P}_+$; that is, $\mathcal{P}_+$ is positively invariant,
$X_+^{**}$ is Lyapunov stable relative to $\mathcal{P}_+$, and every solution
with initial condition in $\mathcal{P}_+$ satisfies
$\lim_{t\to\infty}X(t)=X_+^{**}$.
\end{conjecture}

\begin{remark}
The complement $\Omega_{\mathrm{sf}}\setminus(\mathcal{B}_0\cup
\mathcal{W}^s_+(X_-^{**}))$ is forward invariant: a trajectory that entered
$\mathcal{B}_0$ or $\mathcal{W}^s_+(X_-^{**})$ at some time would, by invariance
of those convergence sets, have had its initial condition in the same set. The
exclusion of $\mathcal{B}_0$ is necessary because $\mathcal E_0$ is locally
asymptotically stable by Theorem~\ref{thm:LAStrivial}, so its basin contains
positive initial conditions; the exclusion of $\mathcal{W}^s_+(X_-^{**})$ is
necessary because points on this stable set converge to the Allee equilibrium.
The conjecture therefore asserts convergence to $X_+^{**}$ on the persistence
side of the Allee threshold, not on the whole positive wild-population state
space. Equivalently, if $\mathcal{W}^s_+(X_-^{**})$ forms the boundary between
the extinction and persistence basins, then $\mathcal{P}_+$ is the persistence
component of $\Omega_{\mathrm{sf}}\setminus\mathcal{W}^s_+(X_-^{**})$
containing $X_+^{**}$.
\end{remark}

    \subsection{The effect of SIT on the equilibria of \eqref{eq:model_I}} 

\noindent 
We now consider the model~\eqref{eq:model_I} with a positive sterile
male release rate $S(t) \not\equiv 0$.  Two biologically natural
release strategies are considered: (i)~release sterile males at a
constant rate $S_0 > 0$, independent of the wild population; and
(ii)~release sterile males at a rate $S_1 M_w(t)$ proportional to
the current wild male population.  Strategy~(i) is operationally
simple but entails a fixed deployment or releases of sterile males even when the
wild population is small.  Strategy~(ii) allows the release effort
to scale down as the wild population declines, reducing unnecessary
expenditure or deployment; however, it requires continuous monitoring of the wild
male population to calibrate the release rate, and the weakening of
the intervention as the population decreases may prolong the time
required to achieve suppression.  To allow for a combination of both
strategies, we set
\begin{equation}\label{eq:release_rate}
    S(t) = S_0 + S_1(F_u(t)+F_{mw}(t)+F_{ms}(t)+M_w(t)), \qquad S_0,\, S_1 \geq 0.
\end{equation}
This form preserves the autonomous structure of model~\eqref{eq:model_I}
and introduces no additional nonlinearities, while accommodating an
adaptive release strategy that responds to fluctuations in the wild adult
mosquito population. Similar constant and proportional sterile male release rates are considered in other prior studies, such as those in \cite{AngDumLub2012, LiCaiLi2017}.  Other studies, such as those in \cite{DumTch2012,IboGumTay2020,LiAi2020}, considered impulsive sterile male releases; a feature or mechanism not considered in the current study.  A central objective of this work is to optimize the
release parameters $(S_0, S_1)$ within this framework so as to drive
the wild mosquito population below its natural Allee threshold while
minimizing the total number of sterile males deployed.  The present
section focuses on the stability and bifurcation structure of
model~\eqref{eq:model_I} under the release
rate~\eqref{eq:release_rate} for fixed, nonnegative $(S_0, S_1)$ not both zero.
 
When $S_0 > 0$, model~\eqref{eq:model_I} with release
rate~\eqref{eq:release_rate} admits a \emph{wild-mosquito-free
equilibrium} (WMFE), denoted $\mathcal{E}_s$, in which all wild
compartments are absent and sterile males are maintained solely by
the constant release:
\begin{equation}\label{eq:WMFE}
    \mathcal{E}_s
    = (0,\,0,\,0,\,0,\,0,\,0,\,0,\,S_0/\mu_M).
\end{equation}
When $S_0 = 0$, $\mathcal{E}_s$ reduces to the MFE $\mathcal E_0$.
 
\begin{theorem}\label{thm:LASsterile}
    Let $S(t) = S_0 + S_1 (F_u+F_{mw}+F_{ms}+M_w)$ in model~\eqref{eq:model_I} with
    $S_0, S_1 \geq 0$.  Then the WMFE $\mathcal{E}_s$ is locally
    asymptotically stable for all admissible parameter values.
\end{theorem}
 
\begin{proof}
The argument follows the proof of Theorem~\ref{thm:LAStrivial}
closely.  The Jacobian of~\eqref{eq:model_I} evaluated at
$\mathcal{E}_s$ differs from $A(\mathcal E_0)$ in~\eqref{eq:AI_0} in
exactly one respect: the eigenvalue $-\mu_F$ now has multiplicity two
(rather than three), and an additional eigenvalue
\begin{equation}\label{eq:new_eigenvalue}
    \lambda
    = -\frac{\eta S_0}{\zeta\eta S_0 + \gamma\mu_M} - \mu_F
    < 0
\end{equation}
appears, corresponding to the linearized dynamics of the
$F_{ms}$-compartment in the presence of a positive sterile male
population $S_0/\mu_M$.  Since this additional eigenvalue is
strictly negative and all other eigenvalues are unchanged from
Theorem~\ref{thm:LAStrivial}, the conclusion follows.
\end{proof}
 
\begin{theorem}\label{thm:GASsterile}
    Let $S(t) = S_0 + S_1 (F_u+F_{mw}+F_{ms}+M_w)$ in model~\eqref{eq:model_I} with
    $S_0, S_1 \geq 0$.  If $\Rq < 1$, then $\mathcal{E}_s$ is
    globally asymptotically stable in $\Omega$.
\end{theorem}
 
\begin{proof}
The proof is analogous to that of Theorem~\ref{thm:GAStrivial}, with
one modification: the $M_s$ term is omitted from the Lyapunov
function, since $M_s$ no longer tends to zero along trajectories
when $S_0 > 0$.  Define
\begin{equation}\label{eq:SIT_Lyapunov_function}
    \begin{split}
    V(E,L,P,F_u,F_{mw},F_{ms},M_w,M_s)
    &= E
     + \left(\frac{\sigma_E+\mu_E}{\sigma_E}\right)\,L
     + \left[\frac{(\sigma_E+\mu_E)(\sigma_L+\mu_L)}{\sigma_E\sigma_L}\right]\,P \\
    &\quad
     + \left[\frac{(\sigma_E+\mu_E)(\sigma_L+\mu_L)(\sigma_P+\mu_P)}
            {r\sigma_E\sigma_L\sigma_P}\right]\,F_u
     + \left(\frac{\phi}{\mu_F}\right)\,F_{mw}.
\end{split}
\end{equation}

The same LaSalle invariance argument as in the proof of
Theorem~\ref{thm:GAStrivial} shows that for any solution
$X(t) = (E(t), L(t), P(t), F_u(t), F_{mw}(t), F_{ms}(t),
M_w(t), M_s(t))$ in $\Omega$,
\begin{equation*}
    \lim_{t\to\infty}
    \bigl(E(t),\,L(t),\,P(t),\,F_u(t),\,F_{mw}(t),\,M_w(t)\bigr)
    = (0,\,0,\,0,\,0,\,0,\,0).
\end{equation*}
Since $F_{ms}$ satisfies $\dot{F}_{ms} =
\frac{\eta M_s}{\gamma+\zeta(M_w+\eta M_s)}F_u - \mu_F F_{ms}$
and $F_u(t) \to 0$, it follows that $F_{ms}(t) \to 0$ as well.
The variation-of-constants formula applied to the $M_s$-equation
gives
\begin{equation*}
    M_s(t)
    = e^{-\mu_M t}M_s(0)
      + \int_0^t e^{-\mu_M(t-s)}\bigl(S_0 + S_1 [F_u(s)+F_{mw}(s)+F_{ms}(s)+M_w(s)]\bigr)\,ds.
\end{equation*}
Since $F_u(s)+F_{mw}(s)+F_{ms}(s)+M_w(s)\to 0$ as $t \to \infty$, an application of
l'H\^{o}pital's rule (or a standard dominated convergence argument)
yields
\begin{equation*}
    \lim_{t\to\infty} M_s(t) = \frac{S_0}{\mu_M}.
\end{equation*}
Hence, $\dss\lim_{t\to\infty} X(t) = \mathcal{E}_s$, establishing
global asymptotic stability.
\end{proof}
 \noindent 
The proof of Theorem~\ref{thm:positive_equilibria} reveals the bifurcation mechanism by which positive equilibria appear in the
SIT-free system: a saddle--node bifurcation occurring at $\Rq +
O(\delta_L, 1/K_E) = 1$.  When sterile males are released, the same
qualitative bifurcation persists, but the release parameters $(S_0,
S_1)$ act as additional bifurcation parameters.
Theorem~\ref{thm:SIT_bifurcation} below establishes that for
sufficiently large release rates, the two positive equilibria
collide and disappear through a second saddle--node bifurcation,
leaving $\mathcal{E}_s$ as the only nonnegative equilibrium.  The
proof relies on the following lemma, which characterizes how positive
roots of a cubic polynomial are affected by the addition of a
nonnegative perturbation which is the effect of the addition of nonzero SIT terms in the model on the cubic polynomial $p(s)$.

\begin{lemma}\label{lem:cubic_roots}
    Let $p(x)$ be a cubic polynomial with positive leading
    coefficient and exactly two distinct positive roots.  Let
    $\ell(x)$ be a nonzero polynomial of degree at most two with
    nonnegative coefficients, and define $q_Z(x) = p(x) +
    Z\ell(x)$ for $Z \geq 0$.  Set
    \begin{equation}\label{eq:Zstar_def}
        Z^{\ast} := -\min_{x > 0}\frac{p(x)}{\ell(x)}.
    \end{equation}
    Then, $Z^{\ast}$ is well-defined and strictly positive. Moreover, $q_Z$ has exactly two distinct positive roots if
    $0 \leq Z < Z^{\ast}$, and no positive roots if $Z > Z^{\ast}$.
\end{lemma}
 
\begin{proof}
Let $r_1 < r_2$ denote the two positive roots of $p$.  Since the
leading coefficient of $p$ is positive, $p(x) < 0$ for $x \in
(r_1, r_2)$ and $p(x) > 0$ for $x \in [0, r_1) \cup (r_2,
+\infty)$.  Since all coefficients of $\ell$ are nonnegative and
$x > 0$, we have $\ell(x) > 0$ for all $x > 0$.  Therefore
$p(x)/\ell(x) < 0$ on $(r_1, r_2)$, and the continuous function
$p(x)/\ell(x)$ attains a strictly negative minimum $-Z^{\ast}$ at
some point $x^{\ast} \in (r_1, r_2)$.  In particular, $Z^{\ast} >
0$.
 
When $Z = 0$, $q_0 = p$ has exactly two positive roots by
hypothesis.  Suppose $0 < Z < Z^{\ast}$.  Then
\begin{equation*}
    q_Z(r_i) = p(r_i) + Z\ell(r_i) = Z\ell(r_i) > 0,
    \quad i = 1, 2,
\end{equation*}
and,
\begin{equation*}
    q_Z(x^{\ast})
    = \ell(x^{\ast})\!\left(\frac{p(x^{\ast})}{\ell(x^{\ast})} + Z\right)
    = \ell(x^{\ast})(Z - Z^{\ast}) < 0.
\end{equation*}
Since $q_Z(r_1) > 0$, $q_Z(x^{\ast}) < 0$, and $q_Z(r_2) > 0$,
the intermediate value theorem gives at least one root of $q_Z$ in
each of $(r_1, x^{\ast})$ and $(x^{\ast}, r_2)$.  Since the leading
coefficient of $q_Z$ is positive and $q_Z(0) = p(0) + Z\ell(0) \geq
0$, the cubic $q_Z$ must also have a negative root (as $q_Z(x) \to
-\infty$ as $x \to -\infty$ and $q_Z(0) \geq 0$).  This accounts
for all three roots of $q_Z$, confirming exactly two positive roots.
 
If $Z > Z^{\ast}$, then for all $x > 0$,
\begin{equation*}
    q_Z(x)
    = \ell(x)\!\left(\frac{p(x)}{\ell(x)} + Z\right)
    \geq \ell(x)(Z - Z^{\ast}) > 0,
\end{equation*}
so $q_Z$ has no positive roots.
\end{proof}

Theorem \ref{thm:SIT_bifurcation} shows that, for sufficiently small sterile male mosquito releases, the equilibria $X_\pm^{\ast\ast}$ are perturbed, but not eliminated. However, our last result in this section is that with a high enough release rate of sterile male mosquitoes, we may drive the wild population to extinction.
    
\begin{theorem}\label{thm:SIT_bifurcation}
    Suppose $\Rq > 1$ and $S(t) = S_0 + S_1 (F_u(t)+F_{mw}(t)+F_{ms}(t)+M_w(t))$
    in~\eqref{eq:model_I}.  There exist constants $K_E^{\ast},
    \delta_L^{\ast}, S_1^{\ast} > 0$ and a continuous function
    $S_0^{\ast} : [0, S_1^{\ast}] \to \mathbb{R}$ with the
    following properties.  Whenever $K_E > K_E^{\ast}$ and
    $\delta_L < \delta_L^{\ast}$:
    \begin{enumerate}
        \item[\textup{(i)}] If $0 \leq S_1 \leq S_1^{\ast}$ and
            $0 \leq S_0 < S_0^{\ast}(S_1)$, then
            model~\eqref{eq:model_I} has exactly two positive
            equilibria (in addition to $\mathcal E_0$ and
            $\mathcal{E}_s$).
        \item[\textup{(ii)}] If $S_1 > S_1^{\ast}$ or $S_0 >
            S_0^{\ast}(S_1)$, then the only nonneg\-ative equilibrium
            of~\eqref{eq:model_I} is $\mathcal{E}_s$.
    \end{enumerate}
    Moreover, in the limit $\delta_L \to 0$, $S_1^{\ast}$ and
    $S_0^{\ast}(S_1)$ have the asymptotic forms
    \begin{align}
        S_1^{\ast}
        &= 
\frac{(1-r)\mu_F\mu_M}
{\eta\bigl[(1-r)\mu_F+r\mu_M\bigr]}
(\Rq-1)
-
\frac{2\mu_F\mu_M}
{\eta\bigl[(1-r)\mu_F+r\mu_M\bigr]}
\sqrt{
\frac{
\Rq(1-r)\gamma\mu_F\mu_M(\sigma_L+\mu_L)(\sigma_P+\mu_P)
}{
K_E\sigma_E\sigma_L\sigma_P(1+\zeta\mu_F)
}
}
+
O(\delta_L),
        \label{eq:S1star}\\[6pt]
        S_0^{\ast}(S_1)
        &= \frac{\Big(\eta S_1[(1-r)\mu_F+r\mu_M]-(1-r)(\Rq-1)\mu_F\mu_M\Big)^2\sigma_E\sigma_L\sigma_P}
            {4\eta\Rq\mu_M^2\mu_F^2(1-r)
             \!\left(\sigma_L+\mu_L\right)
             \!\left(\sigma_P+\mu_P\right)}
       K_E
       - \frac{\gamma\mu_F\mu_M}{\eta(1+\zeta\mu_F)}
       + O(\delta_L).
        \label{eq:S0star}
    \end{align}
\end{theorem}
 
\begin{proof}
Let $K_E^{\ast}$ and $\delta_L^{\ast}$ be as in
Theorem~\ref{thm:positive_equilibria}.  Following the same
equilibrium analysis as in the proof of
Theorem~\ref{thm:positive_equilibria}, the $L$-component of any
equilibrium of~\eqref{eq:model_I} satisfies
\begin{equation}\label{eq:SIT_L_eqn}
    L\,q(L) = 0,
\end{equation}
where $q(L)$ is the cubic polynomial
\begin{equation}\label{eq:q_cubic}
    q(L) = p(L)
           + \bigl(\tilde{S}_0 + \tilde{S}_1 L\bigr)
             \!\left(1 + \frac{\delta_L}{\sigma_L+\mu_L}\,L\right),
\end{equation}
$p(L)$ is the cubic from the proof of
Theorem~\ref{thm:positive_equilibria}, and $\tilde{S}_0$,
$\tilde{S}_1$ are re-scalings of $S_0$, $S_1$ defined by
\begin{align}
    \tilde{S}_0
    &= 
S_0\,\eta\mu_F(1+\zeta\mu_F)
(\sigma_E+\mu_E)(\sigma_L+\mu_L)(\sigma_P+\mu_P)^2,\label{eq:S0tilde}\\
    \tilde{S}_1
       &=S_1\,
\frac{
\eta\sigma_L\sigma_P\bigl[(1-r)\mu_F+r\mu_M\bigr]
(1+\zeta\mu_F)
(\sigma_E+\mu_E)(\sigma_L+\mu_L)(\sigma_P+\mu_P)
}{\mu_M}.
    \label{eq:S1tilde}
\end{align}
As in Theorem~\ref{thm:positive_equilibria}, positive equilibria
correspond to positive roots of $q(L)$, and the equilibrium values
of all compartments other than $M_s$ are positive increasing
functions of $L$ whenever $L > 0$.
 
\medskip
\noindent\textit{Case $S_0 = 0$.}
Setting $S_0 = 0$, equation~\eqref{eq:q_cubic} becomes $q(L) =
p(L) + \tilde{S}_1\,\ell_1(L)$, where
\begin{equation*}
    \ell_1(L) = L\!\left(1 + \frac{\delta_L}{\sigma_L+\mu_L}\,L\right)
\end{equation*}
is a quadratic polynomial with nonnegative coefficients.  Since
$p(L)$ has positive leading coefficient and exactly two distinct
positive roots (by Theorem~\ref{thm:positive_equilibria}), and
$\ell_1(L) > 0$ for $L > 0$, Lemma~\ref{lem:cubic_roots} applies
with
\begin{equation*}
    \tilde{S}_1^{\ast}
    := -\min_{L > 0}\frac{p(L)}{\ell_1(L)}.
\end{equation*}
Hence $q(L)$ has exactly two positive roots for $0 \leq \tilde{S}_1
< \tilde{S}_1^{\ast}$ and no positive roots for $\tilde{S}_1 >
\tilde{S}_1^{\ast}$.  In the asymptotic limit $\delta_L \to 0$, the
minimization problem reduces to a quadratic minimization, providing the asymptotic solution:
\begin{align*}
    S_1^{\ast}
    &= \frac{\mu_M\tilde{S}_1^{\ast}}
            {\eta\sigma_L\sigma_P\left[(1-r)\mu_F+r\mu_M\right](1+\zeta\mu_F)
             (\sigma_E+\mu_E)(\sigma_L+\mu_L)(\sigma_P+\mu_P)}\\
    &=
\frac{(1-r)\mu_F\mu_M}
{\eta\bigl[(1-r)\mu_F+r\mu_M\bigr]}
(\Rq-1)
-
\frac{2\mu_F\mu_M}
{\eta\bigl[(1-r)\mu_F+r\mu_M\bigr]}
\sqrt{
\frac{
\Rq(1-r)\gamma\mu_F\mu_M(\sigma_L+\mu_L)(\sigma_P+\mu_P)
}{
K_E\sigma_E\sigma_L\sigma_P(1+\zeta\mu_F)
}
}
+
O(\delta_L).
\end{align*}
 
\medskip
\noindent\textit{Case $S_0 > 0$.}
For fixed $0 \leq S_1 < S_1^{\ast}$, define
$\tilde{p}(L) := p(L) + \tilde{S}_1\,\ell_1(L)$.  By the case
$S_0 = 0$ just established, $\tilde{p}(L)$ is a cubic with positive
leading coefficient and exactly two distinct positive roots.  Setting
$\ell_2(L) := 1 + \frac{\delta_L}{\sigma_L+\mu_L}\,L$ (a linear
polynomial with positive coefficients), equation~\eqref{eq:q_cubic}
becomes $q(L) = \tilde{p}(L) + \tilde{S}_0\,\ell_2(L)$.
Lemma~\ref{lem:cubic_roots} applies again with
\begin{equation*}
    \tilde{S}_0^{\ast}(S_1)
    := -\min_{L > 0}\frac{\tilde{p}(L)}{\ell_2(L)},
\end{equation*}
giving exactly two positive roots of $q(L)$ when $0 \leq \tilde{S}_0
< \tilde{S}_0^{\ast}(S_1)$ and none when $\tilde{S}_0 >
\tilde{S}_0^{\ast}(S_1)$.  By the implicit function theorem,
$\tilde{S}_0^{\ast}$ is a continuous function of $S_1$.  In the
asymptotic limit $\delta_L \to 0$,
\begin{align*}
    S_0^{\ast}(S_1)
    &= \frac{\tilde{S}_0^{\ast}(S_1)}
            {\eta\mu_F(1+\zeta\mu_F)
             (\sigma_E+\mu_E)(\sigma_L+\mu_L)(\sigma_P+\mu_P)^2}\\
    &= \frac{\Big(\eta S_1[(1-r)\mu_F+r\mu_M]-(1-r)(\Rq-1)\mu_F\mu_M\Big)^2\sigma_E\sigma_L\sigma_P}
            {4\eta\Rq\mu_M^2\mu_F^2(1-r)
             \!\left(\sigma_L+\mu_L\right)
             \!\left(\sigma_P+\mu_P\right)}
       K_E
       - \frac{\gamma\mu_F\mu_M}{\eta(1+\zeta\mu_F)}
       + O(\delta_L).
\end{align*}
This completes the proof of both~\eqref{eq:S1star}
and~\eqref{eq:S0star}.
\end{proof}

\begin{theorem}\label{thm:finite_cost_constant_control}
    For each $X_0\in\Omega$, there exists a constant release
    $S(t)\equiv S_0>0$ sufficiently large (depending on $X_0$) such that solution
    \begin{equation*}
        X(t)=(E(t),L(t),P(t),F_u(t),F_{mw}(t),F_{ms}(t),M_w(t),M_s(t))
    \end{equation*}
    to the model \eqref{eq:model_I} with $X(0)=X_0$ satisfies $\dss\lim_{t\to\infty} X(t)=\mathcal E_s$.
\end{theorem}

\begin{proof}
Suppose $X(t)$ is a solution to \eqref{eq:model_I} with $X(0)=X_0$ and $S(t)\equiv S_0$, a constant. By Proposition \ref{prop:bounded_solns}, each component of $X(t)$ is bounded.  Moreover, the bounds on all components
except $M_s$ may be chosen independently of the control $S(t)$.
In particular, there exists $M_{w,\max}>0$ such that
\begin{equation*}
    0 \leq M_w(t) \leq M_{w,\max}
    \qquad \text{for all } t\geq 0.
\end{equation*}

Choose $S_0>S_0^\ast$, where
\begin{equation*}
    S_0^\ast
    =
    \max\left\{
    0,\,
    \frac{\mu_M}{\eta}
    \left[
    -\frac{\gamma}{\zeta}
    + M_{w,\max}
    \left((\Rq-1)+\frac{\Rq}{\zeta\mu_F}
    \right)
    \right]
    \right\}.
\end{equation*}
For the constant release rate $S(t)\equiv S_0$, the sterile male
component is given explicitly by
\begin{equation*}
    M_s(t)
    =
    M_s(0)e^{-\mu_M t}
    + \frac{S_0}{\mu_M}\left(1-e^{-\mu_M t}\right).
\end{equation*}
Hence $M_s(t)\to S_0/\mu_M$ as $t\to\infty$.  Since
$S_0>S_0^\ast$, there exists $t_1>0$ such that
\begin{equation}\label{eq:Ms_with_constant_release}
    M_s(t)>\frac{S_0^\ast}{\mu_M}
    \qquad \text{for all } t>t_1.
\end{equation}

Consider the candidate Lyapunov function $V:\Omega\to\mathbb R$ given by \eqref{eq:SIT_Lyapunov_function}.
In Theorem \ref{thm:GASsterile}, we show that $\dot V\leq 0$ along trajectories of~\eqref{eq:model_I} provided $\Rq<1$. Now we show that even if $\Rq\geq1$, $\dot V(t)\leq 0$ whenever \eqref{eq:Ms_with_constant_release} holds. Indeed, we may show (with details in Appendix \ref{app:ext_proof}) that
\begin{equation}\label{eq:constant_control_Lyapunov_derivative}
\begin{split}
    \dot V
    &\leq
    -Q(t)F_u
    -\delta_L\left(\frac{\sigma_E+\mu_E}{\sigma_E}\right)L^2,
\end{split}
\end{equation}
where
\begin{equation*}
    Q(t)
    =
    \frac{\phi}
         {\Rq(1+\zeta\mu_F)}
    -
    \frac{\phi}{\mu_F}
    \frac{M_{w,\max}}
         {\gamma+\zeta(M_{w,\max}+\eta M_s(t))}.
\end{equation*}
We claim that $Q(t)>0$ for all $t>t_1$.  Indeed, for $t>t_1$,
\begin{equation*}
    M_s(t)
    >
    \frac{S_0^\ast}{\mu_M}
    \geq
    \frac{1}{\eta}
    \left[
    -\frac{\gamma}{\zeta}
    + M_{w,\max}
    \left(
    (\Rq-1)+\frac{\Rq}{\zeta\mu_F}
    \right)
    \right].
\end{equation*}
Rearranging this inequality gives
\begin{equation*}
    \frac{\phi}{\mu_F}
    \frac{M_{w,\max}}
         {\gamma+\zeta(M_{w,\max}+\eta M_s(t))}
    <
    \frac{\phi}
         {\Rq(1+\zeta\mu_F)},
\end{equation*}
and hence $Q(t)>0$ for all $t>t_1$.  Therefore
\eqref{eq:constant_control_Lyapunov_derivative} implies
$\dot V\leq 0$ for all $t>t_1$. We conclude that $V$ becomes a Lyapunov function for $t>t_1$.

By Proposition~\ref{prop:bounded_solns}, all components are uniformly bounded
for all $t\geq 0$, so $\{X(t):t\geq t_{1}\}$ lies in the compact
set $[0,E_{\max}]\times\cdots\subset\Omega$. Combined with
Proposition~\ref{prop:positivity} (positive invariance of $\Omega$), this provides
the compact positively invariant region required by
LaSalle~\cite{LaS1976}, and the argument that all components of $X(t)$ other than $M_s(t)$ approach $0$ as $t\to \infty$ proceeds as in
Theorems~\ref{thm:GAStrivial} and~\ref{thm:GASsterile}. Finally, from \eqref{eq:Ms_with_constant_release} we see that $M_s(t)\to S_0/\mu_M$ as $t\to\infty$. We conclude that
\begin{equation*}
    \lim_{t\to\infty} X(t)=\mathcal E_s.
\end{equation*}
\end{proof}
\noindent


\section{Numerical simulations: assessment of SIT release strategies}
\label{sec:numerics}

The model~\eqref{eq:model_I} is now simulated to assess the
population-level impact of sterile-male mosquito releases on the local
abundance of wild \emph{Anopheles} mosquitoes.  Specifically, the
simulations are used to: (i) illustrate the bistable structure of the
SIT-free model and the role of the Allee equilibrium as a threshold
between persistence and extinction; (ii) assess how sterile-male
releases alter the equilibrium structure of the model; (iii) determine
release strategies that drive the wild mosquito population below the
Allee threshold while minimizing the cumulative number of sterile males
released; and (iv) identify the biological parameters to which this
cumulative release requirement is most sensitive.

Unless otherwise stated, the simulations are carried out using the
baseline parameter values tabulated in Table~\ref{tab:parameter_values}.
These parameter values are adapted from published empirical estimates
and from the spatial scaling arguments described in
Appendix~\ref{app:spatial}; they are not obtained by fitting
model~\eqref{eq:model_I} to a site-specific field data set.  Hence,
the numerical results should be interpreted as biologically motivated
illustrations of the qualitative and semi-quantitative implications of
the model, rather than as site-specific operational forecasts.  Unless
otherwise stated, the initial condition for SIT simulations is the
natural equilibrium $X_{+}^{\ast\ast}$ of the SIT-free model,
representing the expected state of an established wild mosquito
population at the onset of control.  All time-dependent simulations
were generated using Mathematica's \texttt{NDSolve} routine.  Equilibria
were computed by solving the corresponding algebraic steady-state
equations, and local stability was assessed by evaluating the
eigenvalues of the Jacobian matrix at each computed equilibrium.

For notational convenience, define the total wild adult mosquito
population by
\begin{equation}\label{eq:Aw_definition}
    A_w(t)=F_u(t)+F_{mw}(t)+F_{ms}(t)+M_w(t).
\end{equation}
The sterile-male release function considered in the simulations is
therefore
\begin{equation}\label{eq:release_numerics}
    S(t)=S_0+S_1 A_w(t), \qquad S_0\geq 0,\quad S_1\geq 0.
\end{equation}
Here, $S_0$ represents a constant background release rate, whereas
$S_1A_w(t)$ represents a population-responsive release component scaled
to the current abundance of wild adult mosquitoes.

Most parameters in model~\eqref{eq:model_I} are intrinsic life-history
parameters and are therefore not expected to vary substantially with the
spatial scale of the population under consideration.  Three parameters,
however, are scale-dependent: the egg carrying capacity $\KE$, the
density-dependent larval mortality coefficient $\delta_L$, and the
mate-encounter coefficient $\gamma$.  In particular, $\KE$ increases
with the amount of available oviposition habitat, $\delta_L$ decreases
as the larval habitat expands, and $\gamma$ increases with the
characteristic search area because mate-finding becomes more difficult
in larger domains.  For the baseline parameter values in
Table~\ref{tab:parameter_values}, the reference domain is interpreted
as a neighborhood surrounding a small village, with area approximately
$2\,\mathrm{km}^2$; details of the spatial scaling are given in
Appendix~\ref{app:spatial}.

Since model~\eqref{eq:model_I} is deterministic and continuous-valued,
state variables below one should be interpreted as expected abundances,
or equivalently as densities after rescaling to the reference domain,
rather than as literal fractional mosquitoes.  This interpretation is
especially important near the Allee threshold, where the unstable
equilibrium may correspond to a very small expected population size.

\subsection{Assessment of the SIT-free dynamics and the Allee threshold}

The SIT-free model, obtained by setting $S(t)\equiv 0$, is first
simulated to illustrate the intrinsic population dynamics of the wild
mosquito population.  For the baseline parameter values in
Table~\ref{tab:parameter_values}, the model admits two positive
equilibria, denoted by $X_-^{\ast\ast}$ and $X_+^{\ast\ast}$, in
addition to the mosquito-free equilibrium $\mathcal E_0$.  The larger
equilibrium $X_+^{\ast\ast}$ is locally asymptotically stable and
represents the naturally maintained wild mosquito population, whereas
the smaller equilibrium $X_-^{\ast\ast}$ is unstable and represents the
Allee equilibrium.  The numerically computed values of these two
equilibria are given in Table~\ref{tab:equilibria}.

\begin{table}[!htbp]
\centering
\renewcommand{\arraystretch}{1.25}
\begin{subtable}{0.45\textwidth}
    \centering
    \caption{Natural equilibrium $X_{+}^{\ast\ast}$}
    \begin{tabular}{|l|c|}
        \hline\hline
        Component & Approx. value\\
        \hline\hline
        $E_{+}^{\ast\ast}$      & $8.66\times 10^4$ \\ \hline
        $L_{+}^{\ast\ast}$      & $2.39\times 10^4$ \\ \hline
        $P_{+}^{\ast\ast}$      & $5.13\times 10^3$ \\ \hline
        $F_{u,+}^{\ast\ast}$    & $933$ \\ \hline
        $F_{mw,+}^{\ast\ast}$   & $1.05\times 10^4$ \\ \hline
        $F_{ms,+}^{\ast\ast}$   & $0$ \\ \hline
        $M_{w,+}^{\ast\ast}$    & $6.33\times 10^3$ \\ \hline
        $M_{s,+}^{\ast\ast}$    & $0$ \\
        \hline\hline
    \end{tabular}
\end{subtable}
\hspace{0.04\textwidth}
\begin{subtable}{0.45\textwidth}
    \centering
    \caption{Allee equilibrium $X_{-}^{\ast\ast}$}
    \begin{tabular}{|l|c|}
        \hline\hline
        Component & Approx. value\\
        \hline\hline
        $E_{-}^{\ast\ast}$      & $0.751$ \\ \hline
        $L_{-}^{\ast\ast}$      & $1.91$ \\ \hline
        $P_{-}^{\ast\ast}$      & $0.411$ \\ \hline
        $F_{u,-}^{\ast\ast}$    & $0.903$ \\ \hline
        $F_{mw,-}^{\ast\ast}$   & $1.22\times 10^{-2}$ \\ \hline
        $F_{ms,-}^{\ast\ast}$   & $0$ \\ \hline
        $M_{w,-}^{\ast\ast}$    & $0.507$ \\ \hline
        $M_{s,-}^{\ast\ast}$    & $0$ \\
        \hline\hline
    \end{tabular}
\end{subtable}
\caption{Numerically computed values of the positive natural and Allee
equilibria of the SIT-free model, using the baseline parameter values
in Table~\ref{tab:parameter_values}.}
\label{tab:equilibria}
\end{table}

The Jacobian of the SIT-free model evaluated at $X_-^{\ast\ast}$ has
one eigenvalue with positive real part and seven eigenvalues with
negative real parts.  Hence, the stable manifold of $X_-^{\ast\ast}$,
denoted by $\mathcal M_A$, is seven-dimensional and acts as the Allee
threshold for the SIT-free system.  Initial conditions on one side of
$\mathcal M_A$ converge to $\mathcal E_0$, whereas initial conditions
on the other side converge to $X_+^{\ast\ast}$.

The results in Figure~\ref{fig:noSIT} illustrate this threshold
behavior.  Starting with $E(0)=1.4$ and all other compartments equal to
zero, the solution approaches the mosquito-free equilibrium.  In
contrast, starting with $E(0)=100$ and all other compartments equal to
zero, the solution crosses to the persistence side of the Allee
threshold and converges to $X_+^{\ast\ast}$.  Thus, a relatively small
change in the initial egg abundance can determine whether the wild
population collapses or recovers to its natural equilibrium.

\begin{figure}[h]
    \centering
    \begin{subfigure}{0.48\textwidth}
        \centering
        \includegraphics[width=\linewidth]{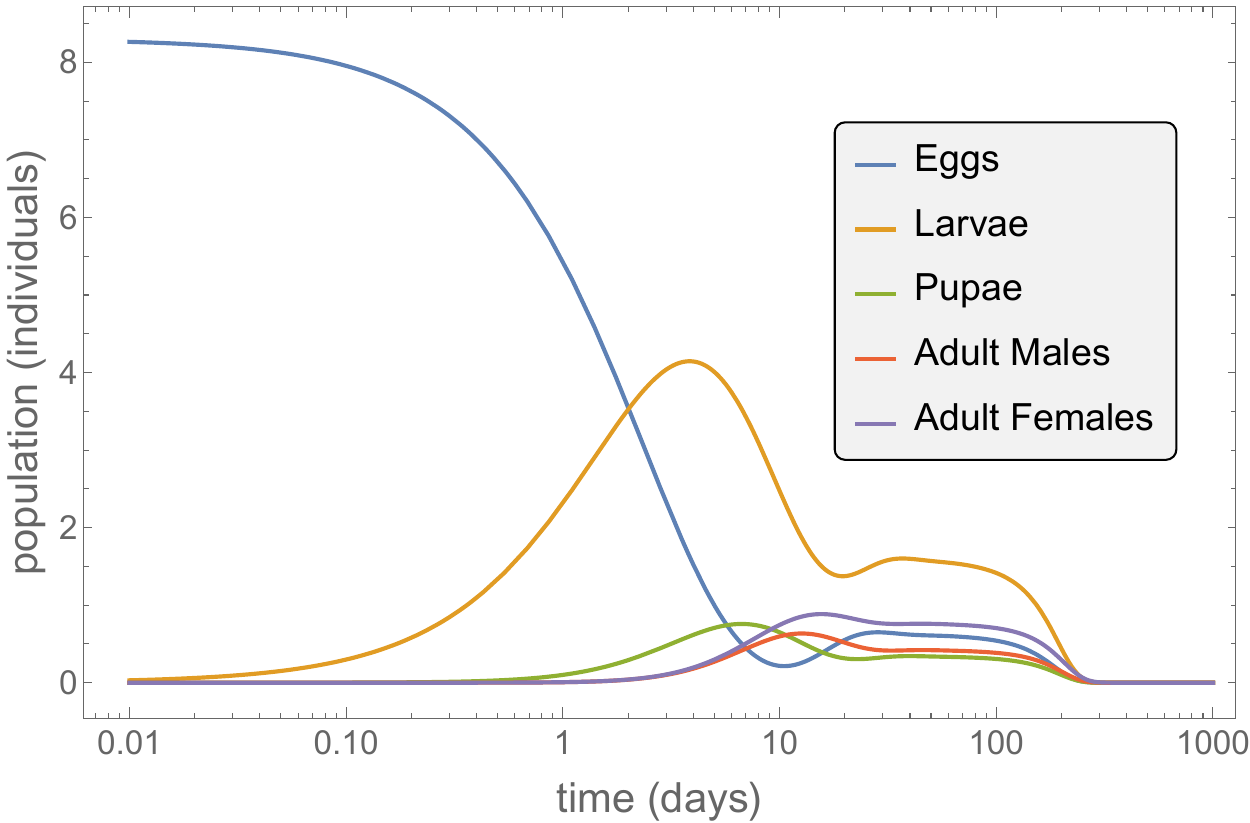}
        \caption{Convergence to the mosquito-free equilibrium.}
        \label{fig:trivial}
    \end{subfigure}\hfill
    \begin{subfigure}{0.48\textwidth}
        \centering
        \includegraphics[width=\linewidth]{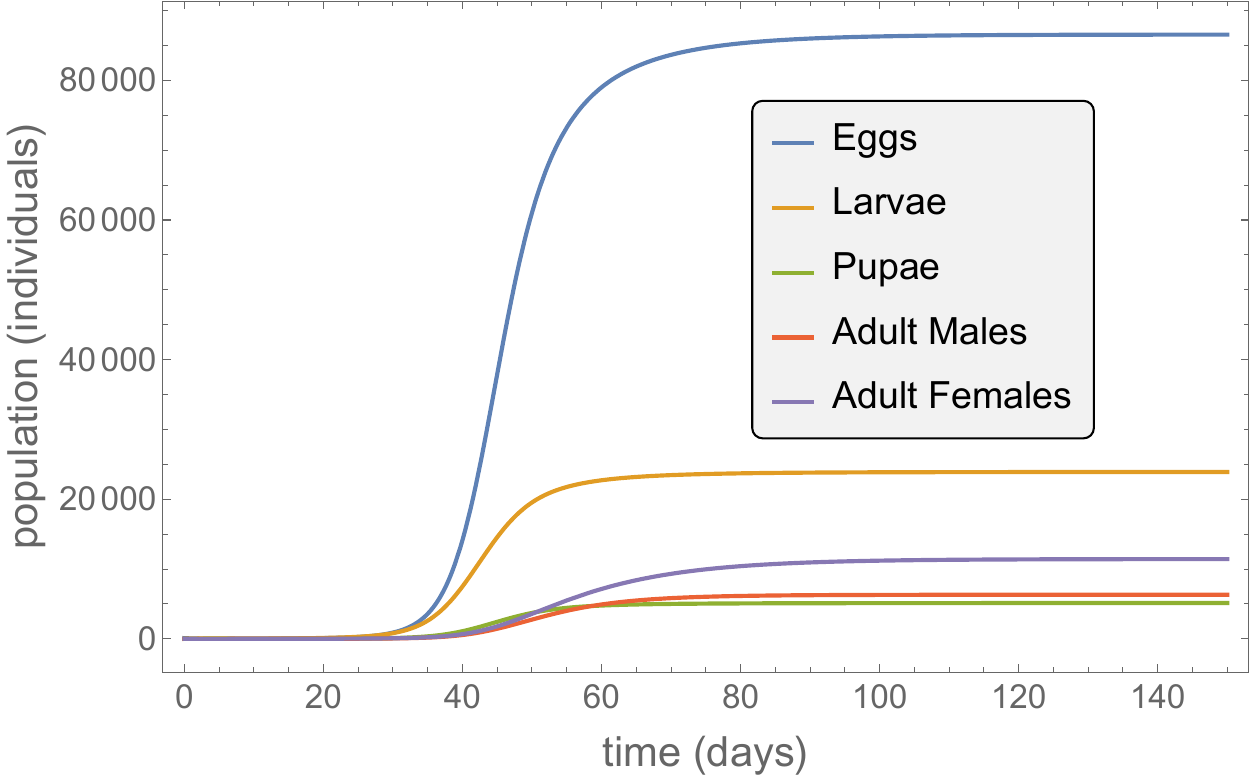}
        \caption{Convergence to the natural equilibrium $X_{+}^{\ast\ast}$.}
        \label{fig:natural}
    \end{subfigure}
    \caption{Numerical solutions of model~\eqref{eq:model_I} in the
    absence of SIT.  (a) The solution with $E(0)=1.4$ and all other
    compartments initially zero converges to the mosquito-free
    equilibrium.  (b) The solution with $E(0)=100$ and all other
    compartments initially zero converges to the natural equilibrium
    $X_{+}^{\ast\ast}$.  Parameter values used are as given in
    Table~\ref{tab:parameter_values}.}
    \label{fig:noSIT}
\end{figure}

Figure~\ref{fig:M_A slices} shows two-dimensional slices of
$\mathcal M_A$ in the $(F_u,M_w)$-plane, with
$L=P=F_{mw}=F_{ms}=M_s=0$, for several fixed values of $E$.  The
threshold curves are unbounded: when few unmated females are present,
sufficiently many wild males are needed for mating to occur before the
adult population dies out.  Initial conditions above and to the right
of the curves converge to $X_+^{\ast\ast}$, whereas those below and to
the left converge to $\mathcal E_0$.  This confirms numerically that
the mate-finding Allee effect can be exploited by SIT: releases need
only push the wild population below the Allee threshold, after which
the natural dynamics complete the decline to extinction.

\begin{figure}[h]
    \centering
    \includegraphics[width=0.6\linewidth]{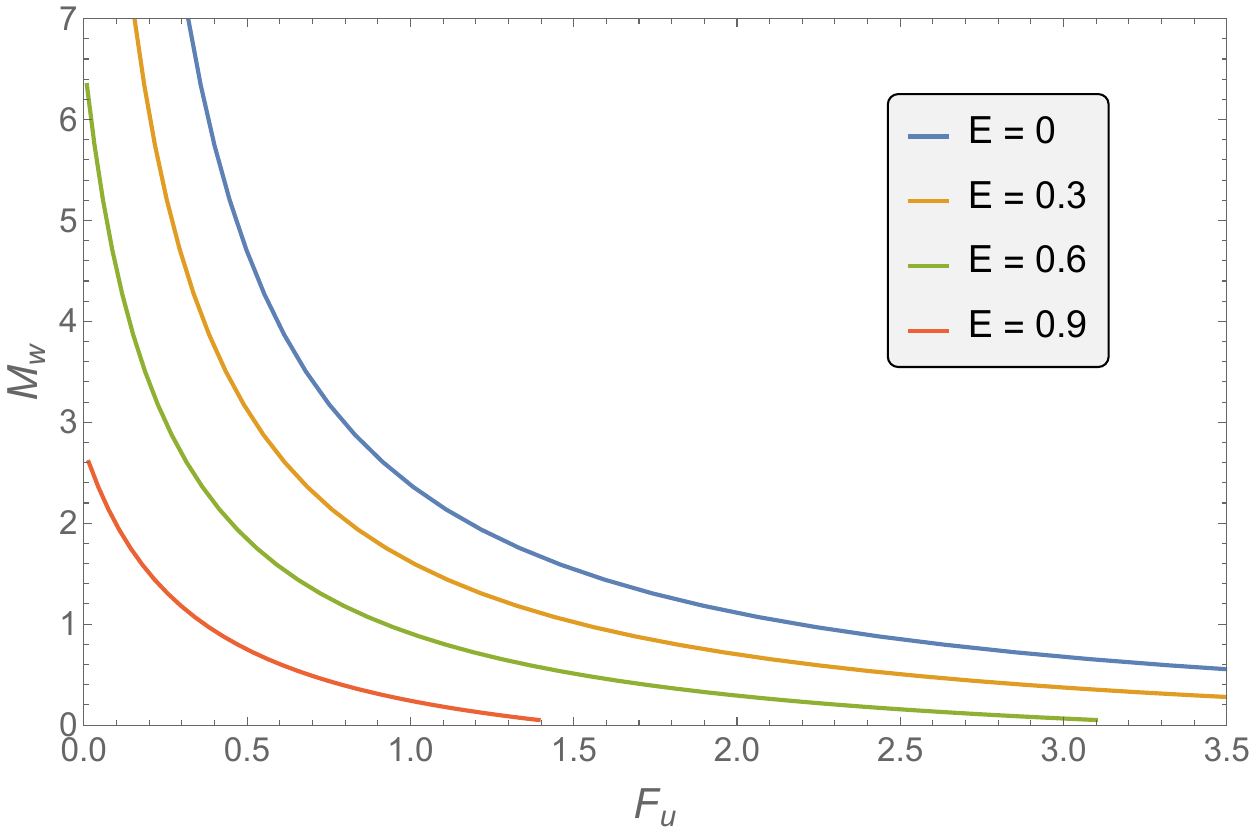}
    \caption{Slices of the Allee threshold manifold $\mathcal M_A$ in
    the $(F_u,M_w)$-plane, with $L=P=F_{mw}=F_{ms}=M_s=0$, for several
    fixed values of $E$.  Initial conditions above and to the right of
    each curve converge to $X_{+}^{\ast\ast}$, whereas initial
    conditions below and to the left converge to the mosquito-free
    equilibrium.}
    \label{fig:M_A slices}
\end{figure}

\subsection{Assessment of the impact of sterile-male releases}

The full model~\eqref{eq:model_I} is next simulated to assess the
impact of sterile-male releases.  Two representative strategies are
compared in Figure~\ref{fig:SIT}: constant release at $S_0=5000$
sterile males per day and population-responsive release at
$S(t)=0.457A_w(t)$.  The population-responsive strategy produces a
faster initial reduction in the wild population, because the release
rate is largest when the wild adult population is large.  However, this
advantage diminishes as the wild population declines, since the release
rate also declines.  Thus, proportional release is effective for early
suppression, but may become too weak near the Allee threshold.

\begin{figure}[h]
    \centering
    \begin{subfigure}{0.48\textwidth}
        \centering
        \includegraphics[width=\linewidth]{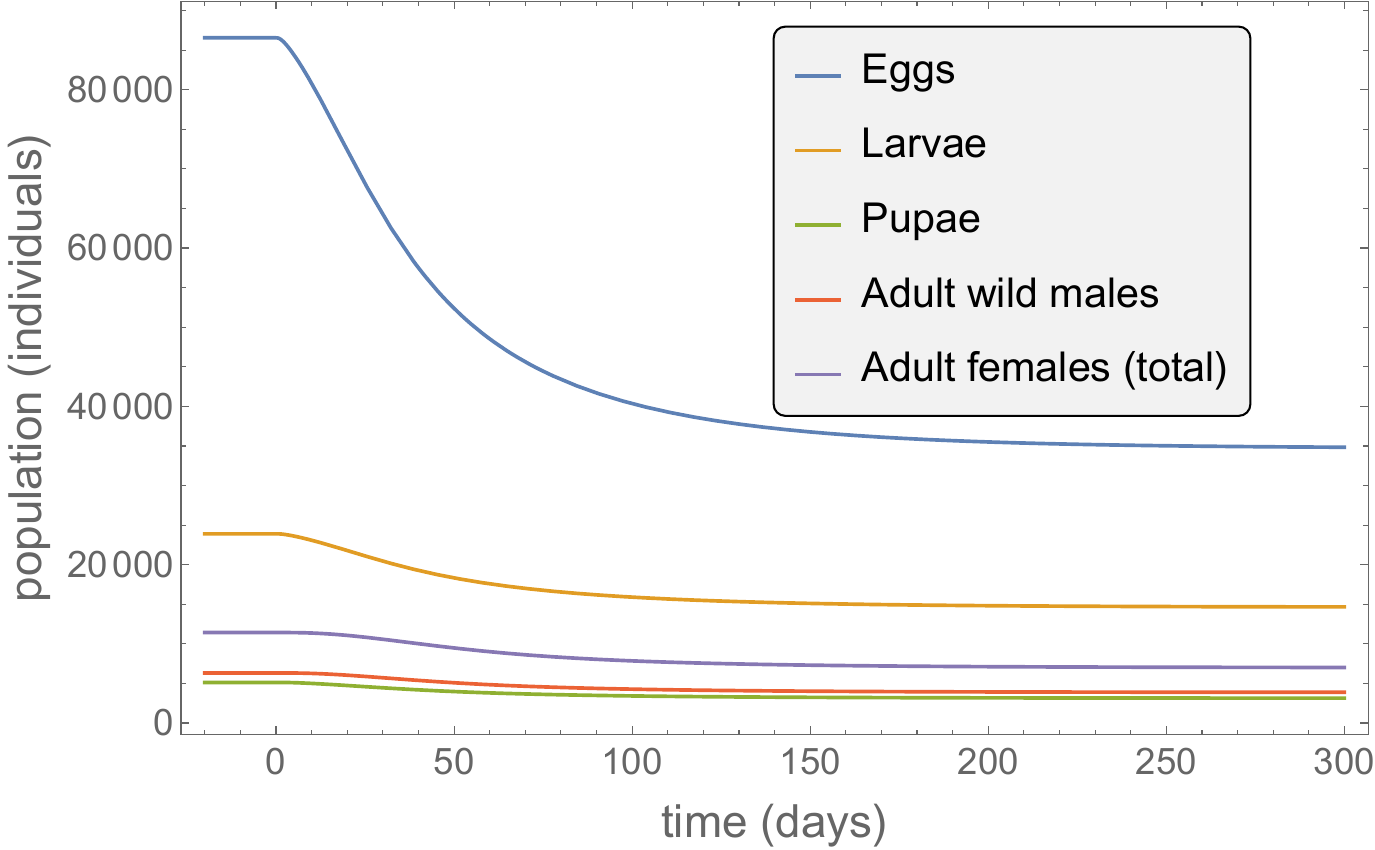}
        \caption{Constant release at $S_0=5000$ sterile males per day.}
        \label{fig:SIT_constant}
    \end{subfigure}\hfill
    \begin{subfigure}{0.48\textwidth}
        \centering
        \includegraphics[width=\linewidth]{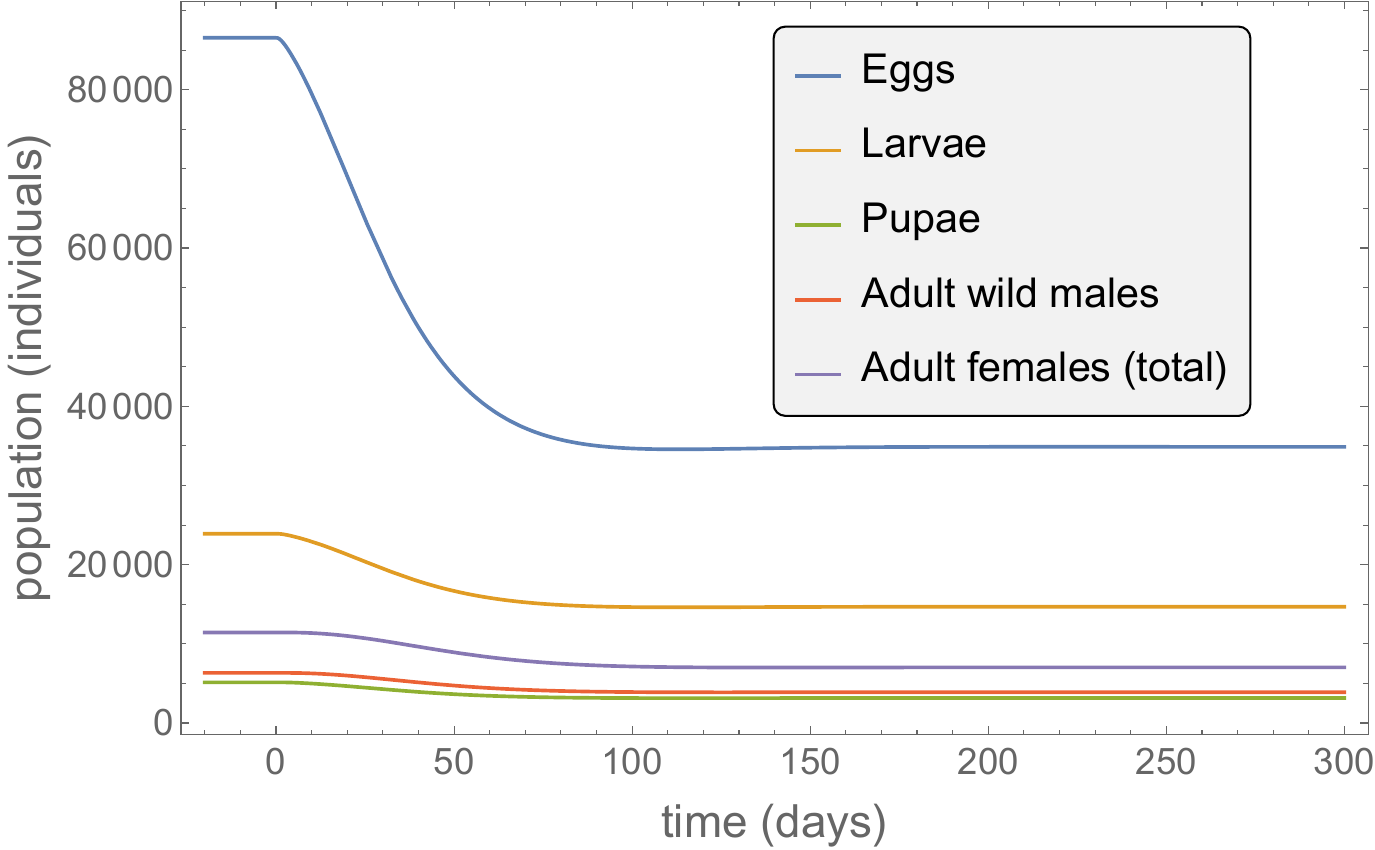}
        \caption{Population-responsive release at $S(t)=0.457A_w(t)$.}
        \label{fig:SIT_proportional}
    \end{subfigure}
    \caption{Numerical solutions of model~\eqref{eq:model_I} under
    sterile-male release.  Parameter values used are as given in
    Table~\ref{tab:parameter_values}, with initial condition
    $X_+^{\ast\ast}$.}
    \label{fig:SIT}
\end{figure}

The bifurcation results in Figures~\ref{fig:extintion/persistence}
and~\ref{fig:S0_S1_bifurcations} show how SIT changes the equilibrium
structure of the model.  For small values of $S_0$ and $S_1$, the
model retains a stable positive wild-mosquito equilibrium and an
unstable Allee-type equilibrium.  As either release parameter is
increased, these two equilibria coalesce in a saddle--node bifurcation
and disappear.  Above the bifurcation curve in the $(S_0,S_1)$-plane,
the wild mosquito compartments converge to zero; when $S_0>0$, the
sterile-male compartment approaches $S_0/\mu_M$.

\begin{figure}[h]
    \centering
    \includegraphics[width=.5\linewidth]{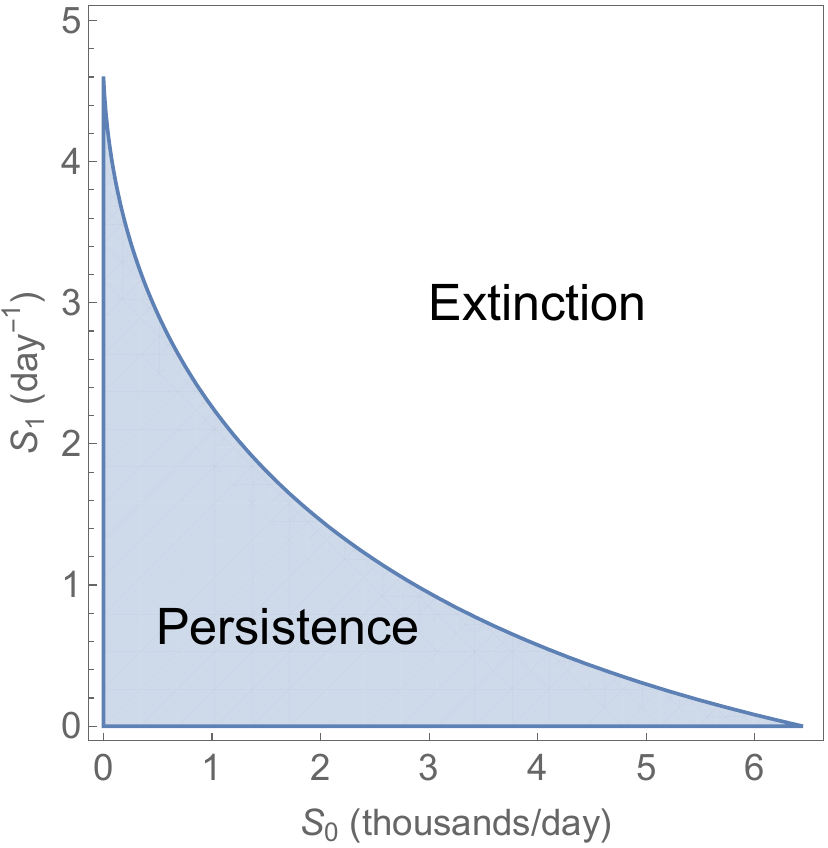}
    \caption{Numerically computed partition of the $(S_0,S_1)$-plane.
    In the persistence region, model~\eqref{eq:model_I} admits positive
    wild-mosquito equilibria.  In the extinction region, the positive
    equilibria have been removed through a saddle--node bifurcation,
    and the wild mosquito compartments decay toward zero.}
    \label{fig:extintion/persistence}
\end{figure}

Taking $S_0$ as the bifurcation parameter, Figure~\ref{fig:S0_bifurcation}
shows the collision of the stable and unstable positive equilibria for
several fixed values of $S_1$.  Similarly, Figure~\ref{fig:S1_bifurcation}
shows the corresponding bifurcation in $S_1$ for several fixed values
of $S_0$.  Increasing $S_0$ shifts both the persistence equilibrium and
the Allee equilibrium, whereas increasing $S_1$ mainly suppresses the
larger persistence equilibrium.  Hence, the proportional component is
most useful when the wild population is large, while the constant
component is more important near elimination.  This provides the
mechanistic explanation for considering a hybrid release strategy.

\begin{figure}[h]
    \centering
    \begin{subfigure}{0.5\textwidth}
        \centering
        \includegraphics[width=\linewidth]{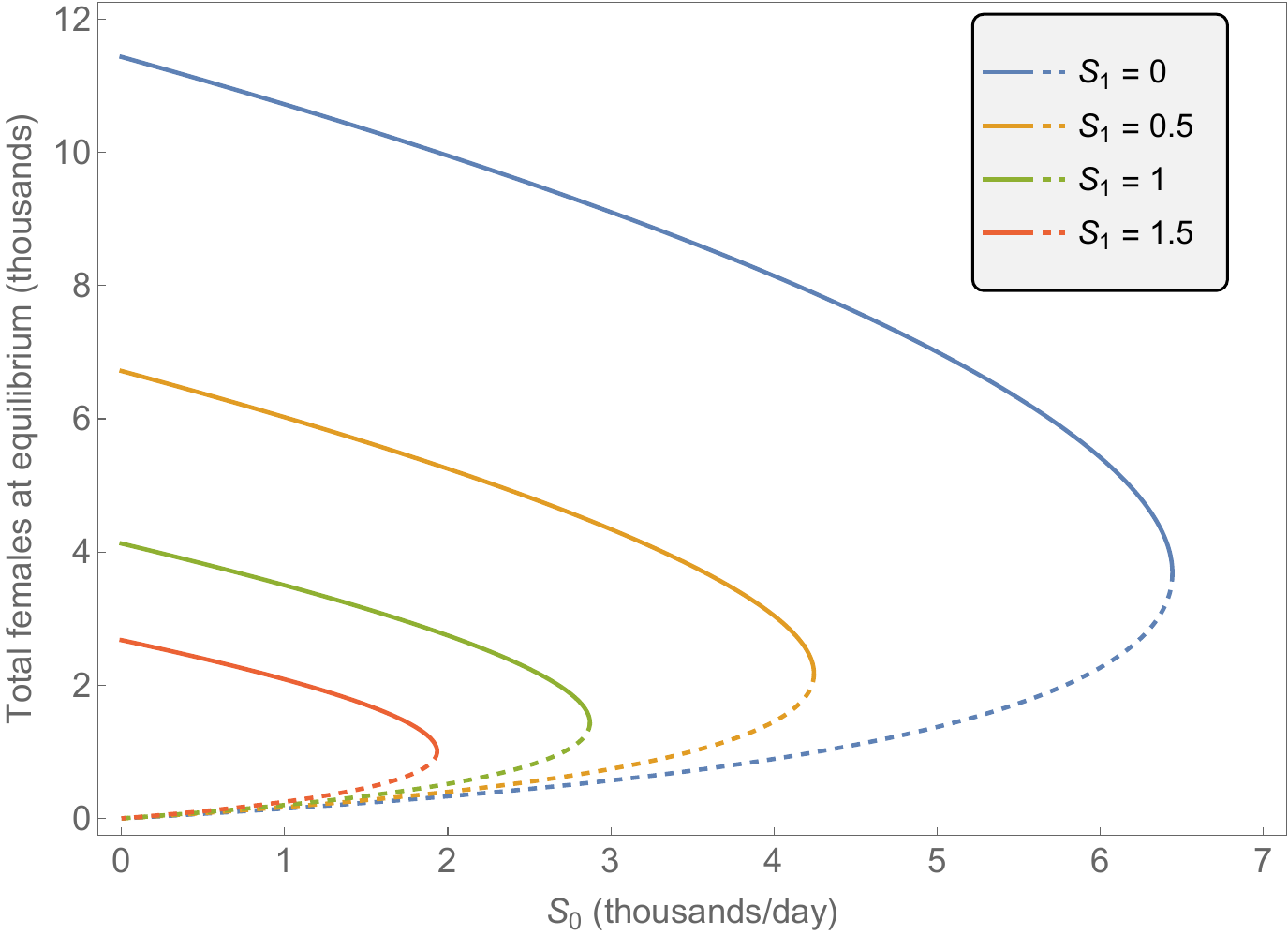}
        \caption{Saddle--node bifurcation in $S_0$.}
        \label{fig:S0_bifurcation}
    \end{subfigure}\hfill
    \begin{subfigure}{0.5\textwidth}
        \centering
        \includegraphics[width=\linewidth]{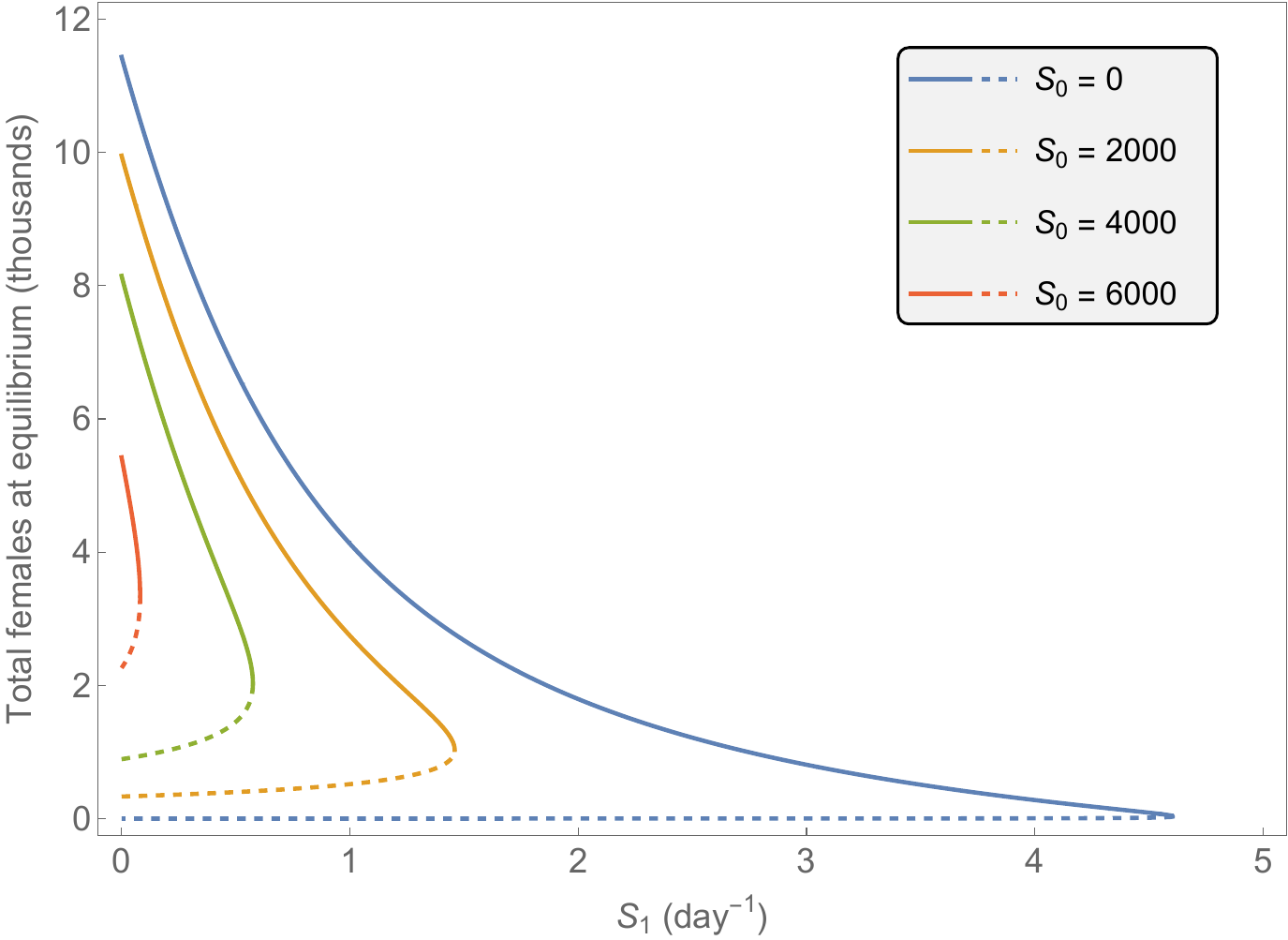}
        \caption{Saddle--node bifurcation in $S_1$.}
        \label{fig:S1_bifurcation}
    \end{subfigure}
    \caption{Bifurcation diagrams for the sterile-male release
    parameters $S_0$ and $S_1$.  Solid curves represent stable positive
    equilibria, whereas dashed curves represent unstable Allee
    equilibria.  In each case, the two positive equilibria meet and
    annihilate through a saddle--node bifurcation.}
    \label{fig:S0_S1_bifurcations}
\end{figure}

\subsection{Assessment of optimized release strategies}

We now determine the release parameters $(S_0,S_1)$ that drive the wild
mosquito population below the numerical Allee threshold while
minimizing the cumulative number of sterile males released.  Since
solutions approach extinction asymptotically, a finite-time operational
stopping criterion is required.  The criterion used here is the first
time the controlled trajectory crosses the extinction side of the
tangent approximation $T_A$ to the Allee threshold $\mathcal M_A$ at
$X_-^{\ast\ast}$.

Let $J_-$ denote the Jacobian matrix of the SIT-free model evaluated at
$X_-^{\ast\ast}$.  Since $X_-^{\ast\ast}$ has a one-dimensional
unstable manifold and a seven-dimensional stable manifold, the tangent
space to $\mathcal M_A$ at $X_-^{\ast\ast}$ is spanned by the
eigenvectors of $J_-$ corresponding to eigenvalues with negative real
parts.  Denote the corresponding affine tangent hyperplane by $T_A$.
This criterion should be viewed as a local computational approximation
of the true Allee threshold.  As a consistency check, trajectories that
crossed $T_A$ were subsequently integrated with $S(t)=0$ and were
observed to converge to the mosquito-free equilibrium.

For each admissible pair $(S_0,S_1)$, define the crossing time
\begin{equation}\label{eq:tau_definition}
    \tau(S_0,S_1)
    =
    \inf\{t>0:\; X(t;S_0,S_1)\ \hbox{has crossed the extinction side of }T_A\},
\end{equation}
where $X(t;S_0,S_1)$ denotes the solution of model~\eqref{eq:model_I}
with release rate~\eqref{eq:release_numerics} and initial condition
$X_+^{\ast\ast}$.  The cumulative release requirement is
\begin{equation}\label{eq:N_definition}
    N(S_0,S_1)
    =
    \int_0^{\tau(S_0,S_1)}
    \bigl[S_0+S_1 A_w(t)\bigr]dt .
\end{equation}
Thus, the numerical optimization problem is
\begin{equation}\label{eq:optimization_problem}
    \min_{S_0,S_1\geq 0} N(S_0,S_1),
    \qquad
    \hbox{subject to } \tau(S_0,S_1)<\infty .
\end{equation}

The efficiency map in Figure~\ref{fig:efficieny} shows a clear
low-cost region at intermediate release intensities.  Release rates
chosen only slightly beyond the saddle--node bifurcation curve require
many days of control, whereas very large release rates produce
diminishing returns.  Hence, the optimal strategy is not the largest
possible release rate, but a balanced release rate that suppresses the
population efficiently while allowing the Allee effect to contribute to
elimination.

\begin{figure}
    \centering
    \includegraphics[width=0.7\linewidth]{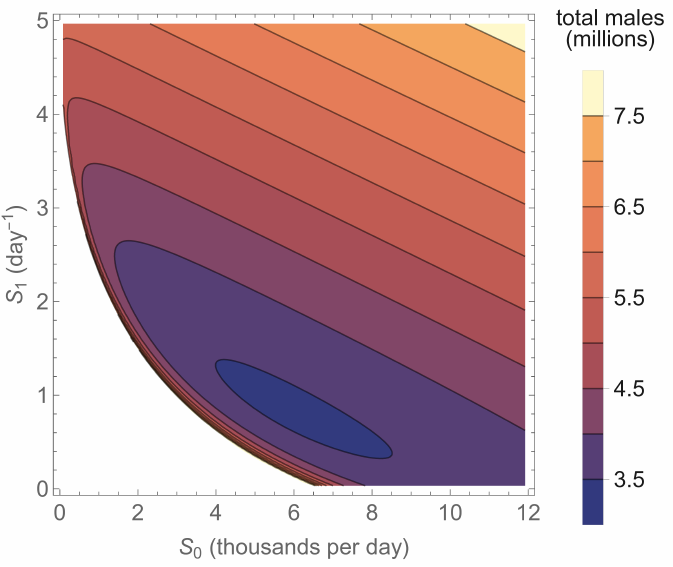}
    \caption{Contour plot of the cumulative number of sterile males
    required to drive the wild mosquito population from
    $X_{+}^{\ast\ast}$ below the numerical Allee threshold, as a
    function of the release parameters $S_0$ and $S_1$.  Parameter
    values used are as given in Table~\ref{tab:parameter_values}.}
    \label{fig:efficieny}
\end{figure}

Table~\ref{tab:sterile-male-minima} compares optimized values for three
cases: constant-only release ($S_1=0$), population-responsive release
only ($S_0=0$), and unconstrained hybrid release ($S_0,S_1>0$).  The
constant-only strategy requires $3.62\times 10^6$ sterile males over
357 days.  The population-responsive-only strategy requires
$5.65\times 10^6$ sterile males over 1150 days, despite its much larger
peak release rate ($8.56\times 10^4$ sterile males per day).  The
hybrid strategy performs best, requiring $3.44\times 10^6$ sterile
males over 355 days, with a peak release rate of $2.06\times 10^4$
sterile males per day.  Hence, the hybrid strategy reduces the
cumulative release requirement by about $5\%$ relative to the best
constant-only strategy and by about $39\%$ relative to the best
population-responsive-only strategy.

\begin{table}[h]
    \centering
    \renewcommand{\arraystretch}{1.25}
    \resizebox{\textwidth}{!}{%
    \begin{tabular}{|l|l|c|c|c|c|c|}
    \hline\hline
    Strategy
    & $S(t)$
    & \makecell{$S_0$ \\ \textnormal{(sterile males/day)}}
    & \makecell{$S_1$ \\ \textnormal{(day$^{-1}$)}}
    & \makecell{Cumulative males \\ \textnormal{(sterile males)}}
    & \makecell{Time required \\ \textnormal{(days)}}
    & \makecell{Peak release rate \\ \textnormal{(sterile males/day)}} \\
    \hline\hline
    $S_0$-only
    & $S_0$
    & $1.01\times 10^{4}$
    & $0$
    & $3.62\times 10^{6}$
    & $357$
    & $1.01\times 10^{4}$\\
    \hline
    $S_1$-only
    & $S_1 A_w(t)$
    & $0$
    & $4.82$
    & $5.65\times 10^{6}$
    & $1150$
    & $8.56\times 10^{4}$\\
    \hline
    Hybrid
    & $S_0+S_1 A_w(t)$
    & $5.98\times 10^{3}$
    & $0.826$
    & $3.44\times 10^{6}$
    & $355$
    & $2.06\times 10^{4}$\\
    \hline\hline
    \end{tabular}%
    }
    \caption{Comparison of optimized sterile-male release strategies.
    The cumulative number of sterile males is computed using
    \eqref{eq:N_definition}, with initial condition $X_+^{\ast\ast}$
    and the stopping criterion defined by the tangent approximation to
    the Allee threshold.}
    \label{tab:sterile-male-minima}
\end{table}

The hybrid optimum has a clear biological interpretation.  The
population-responsive term provides an elevated initial release when
the wild adult population is large, accelerating the early reduction in
reproductive output.  The constant term then maintains sufficient
sterile-male pressure as the wild population becomes small, preventing
the control effort from decaying too rapidly near the Allee threshold.
Thus, releases need not drive every wild compartment arbitrarily close
to zero; they need only push the trajectory across the Allee threshold,
after which the natural mate-finding Allee effect completes the decline
to extinction.

This result is consistent with previous SIT studies showing that
effective suppression generally requires sufficiently large and
sustained sterile-male releases~\cite{DumTch2012,IboGumTay2020,ZheEtAl2019}.
It also complements the findings of Iboi, Gumel and Taylor~\cite{IboGumTay2020},
who showed that larval density-dependent mortality and seasonal forcing
can strongly affect the number and frequency of releases needed for
mosquito elimination.  The new feature of the present study is not the
conclusion that large releases are needed, but rather the identification,
within a continuous two-parameter release family, of a hybrid strategy
that balances an initial population-responsive release with a sustained
constant release until the Allee threshold is crossed.  This appears to
be a relatively less-explored feature of the SIT modeling literature,
where many optimization studies are posed over fixed implementation
horizons~\cite{BliCarDumVas2024,HuaYouLiuSon2021,ThoYanEst2010}.

For the reference area of approximately $2\,\mathrm{km}^2$, the optimal
hybrid strategy requires a peak release of approximately
$2.06\times 10^4$ sterile males per day and a cumulative release of
approximately $3.44\times 10^6$ sterile males.  These values should be
interpreted as illustrative rather than site-specific, but they are
within the range of localized pilot-scale SIT implementation.  

Standard SIT protocol indicates that sterile males should be released such that they outnumber wild males at a ratio of at least $10{:}1$ \cite{CarEtAl2022,OliEtAl2021,StrBosDum2019}. Then the wild population is flooded with wild males. Taking the peak sterile male release rates $S$ from Table \ref{tab:sterile-male-minima}, we may calculate quasi-stationary sterile male population $S/\mu_M$ and compute a sterile-to-wild male ratio using the population of wild males at equilibrium. These ratios are shown in Table \ref{tab:ratios}. This table shows that the constant release strategy agrees with the $10{:}1$ rule of thumb, while the hybrid release strategy shows that it is slightly more efficient to maintain a higher ratio, at least early in the SIT program.

\begin{table}[h]
    \centering
    \renewcommand{\arraystretch}{1.25}
    \begin{tabular}{|l|l|}
    \hline\hline
    Strategy
    &  Sterile-to-wild male ratio at peak release\\
    \hline\hline
    $S_0$-only
    & 10.7\\
    \hline
    $S_1$-only
    & 90.3\\
    \hline
    Hybrid
    & 21.8\\
    \hline\hline
    \end{tabular}
    \caption{The sterile-to-wild male mosquito population ratios at the peak release rate for each strategy. That is for each strategy, this table reports the ratio of sterile male population $S/\mu_M$ to the wild male population at the start of the program.}
    \label{tab:ratios}
\end{table}



\subsection{Sensitivity analysis of the sterile-male requirement}

Finally, the model is simulated to assess the sensitivity of the
optimized cumulative release requirement to uncertainty in the
biological parameters.  For each parameter $p$, the normalized
sensitivity index
\begin{equation}\label{eq:sensitivity_index}
    \Upsilon_p^N
    =
    \frac{p}{N}\frac{\partial N}{\partial p}
\end{equation}
is computed at the baseline parameter set, where $N$ is the optimized
number of sterile males required to drive the wild population below the
numerical Allee threshold.

The results in Figure~\ref{fig:sensitivity} show that $N$ is most
sensitive to the adult female mortality rate $\mu_F$, the larval
maturation rate $\sigma_L$, and the oviposition rate $\phi$.
Biologically, this is expected.  Increasing $\phi$ increases egg
production and therefore increases the release requirement, whereas
increasing $\mu_F$ shortens the adult female lifespan and lowers the
release requirement.  The sensitivity to $\sigma_L$ reflects the role
of density-dependent larval mortality: faster larval maturation reduces
the time larvae spend exposed to density-dependent mortality, thereby
increasing adult recruitment and raising the SIT effort required. Among all model parameters, $\eta$ is the only one tied directly to the
sterilization of released males, and its normalized sensitivity index is exactly
$-1$. This value reflects an exact scaling symmetry of
model~\eqref{eq:model_I} under the sterile-male release
rate~\eqref{eq:release_numerics}. Because the sterile males influence the wild
compartments solely through the product $\eta M_s$, and because $M_s$ is slaved
to the release rate via $\dot M_s = S - \mu_M M_s$, the transformation
\[
  (\eta,\,S_0,\,S_1,\,M_s)\longmapsto
  (\alpha\eta,\,S_0/\alpha,\,S_1/\alpha,\,M_s/\alpha)
\]
leaves every wild compartment---and hence the Allee-threshold crossing
time $\tau$---unchanged, while scaling the cumulative sterile-male
requirement~$N$ by $1/\alpha$. The optimum therefore transforms covariantly,
so the optimized requirement obeys $N^{\star}(\alpha\eta)=N^{\star}(\eta)/\alpha$
exactly, yielding $\Upsilon_\eta^N=-1$. Consequently, if a fully competitive
sterile male could be produced with no fitness cost ($\eta=1$), every
sterile-male quantity in Table~\ref{tab:sterile-male-minima} would fall to
three-quarters of its baseline value---a reduction of $25\%$---relative to the
baseline $\eta = 0.75$.

By contrast, the optimized release requirement is relatively insensitive
to the adult male mortality rate $\mu_M$ and to the mating-time
parameters $\zeta$ and $\gamma$ at the baseline parameter set.  This
does not imply that mating biology is unimportant; rather, it indicates
that, near the baseline values, uncertainty in female survival, larval
development, and oviposition has the greatest effect on the predicted
number of sterile males required.  This agrees with prior SIT modeling
work emphasizing that larval density-dependence can substantially change
the release threshold for elimination~\cite{IboGumTay2020}.

\begin{figure}
    \centering
    \includegraphics[width=0.9\linewidth]{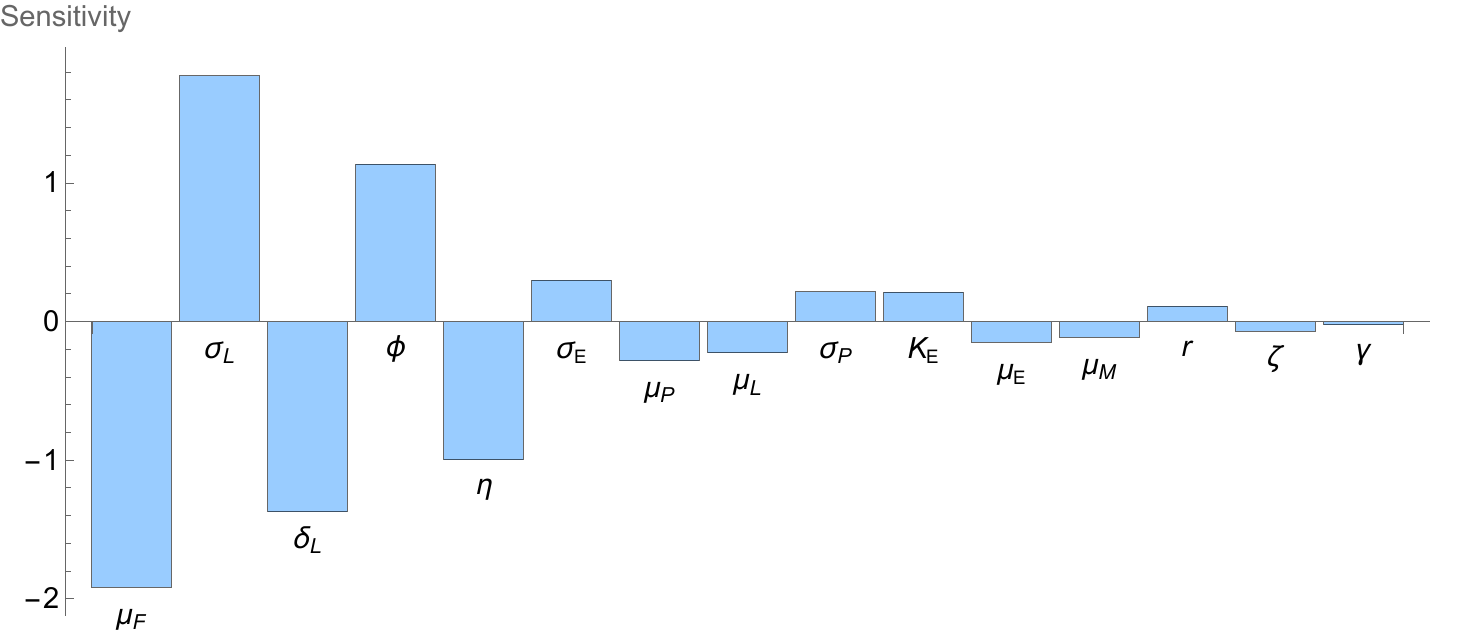}
    \caption{Normalized sensitivity indices for the optimized
    cumulative sterile-male requirement $N$.  For each parameter $p$,
    the plotted quantity is
    $\Upsilon_p^N=(p/N)(\partial N/\partial p)$, evaluated at the
    baseline parameter values in Table~\ref{tab:parameter_values}.}
    \label{fig:sensitivity}
\end{figure}

In summary, the numerical simulations support the analytical results of
Section~\ref{sec:analysis}.  The SIT-free model exhibits bistability,
with the stable manifold of the Allee equilibrium separating
persistence from extinction.  Sterile-male releases can eliminate the
positive wild-mosquito equilibria through a saddle--node bifurcation.
Among the strategies considered, the most efficient is a hybrid release
strategy, which uses a population-responsive component for rapid early
suppression and a constant component to maintain pressure near the
Allee threshold.  This finding is consistent with the broader SIT
literature showing that sustained release pressure is needed for
elimination, while adding the more specific conclusion that, under an
Allee-threshold stopping criterion, combining constant and
population-responsive releases can reduce the cumulative number of
sterile males required for local elimination.

    \section{Discussion}\label{sec:discussion}
\ \\ \noindent This study developed and rigorously analyzed a sex- and stage-structured
model for the population dynamics of malaria-vectoring \emph{Anopheles}
mosquitoes under sterile insect technique (SIT) intervention. The model
incorporates the aquatic immature stages (eggs, larvae, and pupae) and the two
adult sexes; distinguishes unmated females from females mated with wild or
sterile males; includes density-dependent regulation through a finite egg-laying
carrying capacity and density-dependent larval mortality; and uses a
mechanistic mating response in which the pre-mating interval is decomposed into
a post-emergence refractory phase and a density-dependent mate-search phase.
This mating formulation generates a mate-finding Allee effect. Sterile males
are released continuously in time according to a rate that combines a constant
component with a component proportional to the wild adult abundance. The model
is adapted from the framework of Iboi, Gumel and Taylor~\cite{IboGumTay2020},
from which it differs principally by replacing impulsive, periodic releases with
a continuous non-constant release rate and by allowing a population-responsive
release component. These choices preserve the autonomous structure of the
system while permitting, within a single tractable framework, a rigorous
treatment of well-posedness, equilibrium existence, asymptotic stability,
bifurcation structure, and release optimization. The model is mathematically
and ecologically well posed: solutions exist globally, are unique, remain
non-negative, and are uniformly bounded for every admissible initial condition
and release schedule.

A central qualitative result is the mate-finding Allee effect generated by the
strictly positive mate-search parameter~$\gamma$. Because the mating term is
quadratic in the relevant state variables near extinction---a Holling type-II
response that reduces to a mass-action term at low density---the linearized
reproductive pathway is inactive at the mosquito-free equilibrium. Consequently,
the mosquito-free equilibrium is locally asymptotically stable for all
admissible parameter values, a low-density ``superstability'' phenomenon also
identified in ecological models with Holling type-II interactions~\cite{BouBer2009}.
The basic reproduction number of the limiting quick mate-search model, denoted
$\mathcal{R}_0^q$, provides a conservative sufficient condition for global
extinction: the mosquito-free equilibrium is globally asymptotically stable
whenever $\mathcal{R}_0^q<1$. Since the quick mate-search limit overestimates
reproductive output relative to the full mate-search model, this threshold is
sufficient but not, in general, necessary for extinction. Thus, low-density
extinction is not imposed by a control assumption; it emerges from the intrinsic
mating dynamics of the vector population. This mechanism contrasts with the
Esteva--Yang model~\cite{EstYan2005}, whose effective mate-search coefficient
vanishes when sterile males are absent or fully competitive and therefore does
not generate a low-density mate-finding Allee effect. When
$\mathcal{R}_0^q>1$, and when the egg carrying capacity is sufficiently large
and larval competition sufficiently weak, the SIT-free system exhibits the
bistable structure characteristic of mate-finding Allee dynamics: a locally
asymptotically stable natural equilibrium coexists with an unstable Allee
equilibrium whose stable set separates extinction from persistence. For the
reduced SIT-free model, a Goh--Volterra Lyapunov functional gives a certified
inner estimate of the persistence basin. The corresponding persistence-side
global stability statement for the full model remains a conjecture supported by
numerical simulations. This two-positive-equilibrium bistable picture recovers,
in a more detailed setting and with explicit asymptotic thresholds, the
qualitative structure reported for related mate-finding models~\cite{AngDumYatIvr2020,DumYat2022}.

Sterile-male releases modify this bistable structure in a way that is especially
favorable for control. A positive constant release creates a
wild-mosquito-free equilibrium with a sustained sterile-male population. This
equilibrium is locally asymptotically stable for all admissible parameter values
and globally asymptotically stable when $\mathcal{R}_0^q<1$. In the biologically
important regime $\mathcal{R}_0^q>1$, where the wild population would otherwise
persist, sufficiently large releases eliminate the two positive wild-mosquito
equilibria through a saddle--node bifurcation. The analysis also gives explicit
small-$\delta_L$ asymptotic expressions for the critical proportional and
constant release rates. Complementing this equilibrium-level result, a
sufficiently large constant release drives the system to the wild-mosquito-free
equilibrium from every admissible initial condition, even when
$\mathcal{R}_0^q\ge1$. Consequently, SIT does not need to suppress every wild
compartment directly to an arbitrarily small level. It is enough to move the
wild population across the Allee separatrix; after that crossing, the natural
mate-finding Allee dynamics complete the decline to local elimination. The
analysis therefore provides a rigorous mechanism for local elimination of the
wild population, rather than only a numerical observation of suppression.

The numerical optimization quantifies this threshold-crossing mechanism. Using a
free-horizon criterion in which releases continue only until the controlled
trajectory crosses a tangent approximation to the Allee separatrix, a hybrid
release strategy was the most efficient among the strategies considered. It
required approximately $3.44\times10^{6}$ sterile males over $355$~days,
compared with approximately $3.62\times10^{6}$ sterile males over $357$~days for
the best constant-only strategy and approximately $5.65\times10^{6}$ sterile
males over $1150$~days for the best population-responsive-only strategy
(Table~\ref{tab:sterile-male-minima}). Thus, the hybrid strategy reduced the
cumulative release requirement by about $1.8\times10^{5}$ sterile males
(approximately $5\,\%$) relative to the best constant-only strategy and by about
$2.2\times10^{6}$ sterile males (approximately $39\,\%$) relative to the best
population-responsive-only strategy. The biological interpretation is direct:
the population-responsive component supplies strong early pressure when the wild
adult population is large, whereas the constant component prevents the release
effort from decaying too rapidly near the Allee separatrix, where a purely
population-responsive release rate would otherwise weaken.

For the reference domain of approximately $2\,\mathrm{km}^2$, the peak release
rate of the optimal hybrid strategy corresponds, under the quasi-steady
approximation $M_s \approx S_{\max}/\mu_M$, to a peak sterile-to-wild male
abundance ratio of the same order of magnitude as the customary $10{:}1$ SIT
overflooding rule of thumb. The optimization indicates, moreover, that
maintaining a somewhat higher ratio---closer to $20{:}1$---during the early
phase of the release program can reduce the total number of sterile males
ultimately required, a refinement consistent with the values reported in
Table~\ref{tab:ratios}.

Placed within the broader SIT modeling literature, the contribution of this
work is both analytical and strategic. Relatively few SIT models combine
detailed stage and sex structure, nonlinear mate-finding dynamics, rigorous
equilibrium and stability analysis, bifurcation theory, and release optimization
within a single tractable framework. The present study establishes the existence
and local asymptotic stability of the mosquito-free and wild-mosquito-free
equilibria; gives sufficient conditions for their global asymptotic stability;
characterizes the coexistence of a stable natural equilibrium with an unstable
Allee equilibrium; derives asymptotic expressions for the critical release
thresholds; and constructs, for a reduced SIT-free model, a Lyapunov-based inner
estimate of the persistence basin. Whereas optimal-control formulations of SIT
have been posed mostly over fixed implementation horizons~\cite{BliCarDumVas2024,HuaYouLiuSon2021,ThoYanEst2010},
and much of the suppression literature concerns impulsive or constant
releases~\cite{DumTch2012,IboGumTay2020,LiCaiLi2017,ZheEtAl2019}, the present
free-horizon formulation with an Allee-threshold stopping criterion identifies a
hybrid schedule that lowers the cumulative sterile-male requirement relative to
either pure strategy among the strategies compared. The unifying theme is that
the control strategy does not work against the nonlinear ecology of the vector
population but exploits it: the Allee effect becomes a resource for elimination,
with SIT supplying the push across the threshold.

The study has several limitations.  For instance, the model is deterministic and
assumes a well-mixed, closed, temporally constant habitat. Demographic
stochasticity---most consequential near the Allee threshold, where abundances
are small---together with spatial heterogeneity, immigration, and seasonal or
climate-driven variation are therefore not explicitly accounted for.  Furthermore, our analysis assumes
perfect sex sorting and complete sterility, and the baseline parameters are
biologically motivated rather than fitted to a specific field site. Thus, the
reported release quantities should be interpreted as illustrative rather than as
site-specific operational forecasts. In addition, the global stability of the
natural equilibrium is established only conditionally, on certified sublevel sets
for the reduced model, while the corresponding full-model statement remains a
conjecture; the numerical elimination criterion is also a local tangent
approximation of the true Allee separatrix. Future work will accordingly extend
the framework to stochastic and spatial settings---including
reaction--advection--diffusion formulations and SIT-barrier strategies designed
to prevent re-invasion~\cite{AlmEstVau2022,Bit2019,EvaBis2014,MaLiCai2026}---incorporate
immigration and seasonal forcing, introduce interspecific competition between a
high-vectorial-capacity target species and a low-vectorial-capacity
competitor~\cite{HasAfrAfrIsl2024,HayEtAl2002}, couple the vector dynamics to
\emph{Plasmodium} transmission, and examine combined SIT--insecticide strategies
under realistic resource constraints.

Taken together, these results place the mate-finding Allee effect on a rigorous
footing as a control lever rather than an obstacle: SIT need only drive the wild
population across the Allee separatrix, after which the population's own
low-density mate-finding failure completes the decline to local elimination. The
focus of the present work is local control of the vector---the sustained
suppression and, ultimately, local elimination of wild \emph{Anopheles}
populations. Such local vector control is one component of the broader malaria-control
portfolio needed to pursue the WHO Global Technical Strategy targets for
2030~\cite{WHO2015} and longer-term eradication goals~\cite{Gat2019,WilHam2018}.
Whether the mechanism analyzed here can be realized in the field depends on
factors outside the present model, including the capacity to mass-rear and
release competitive sterile males at the required scale, the retention of mating
competitiveness and survival after sterilization, and the spatial, stochastic,
and seasonal complexities of operational settings. Subject to these caveats, and
once the vector dynamics are coupled to \emph{Plasmodium} transmission, the
Allee-exploiting release strategy analyzed here---integrated with established
vector-control and case-management programs---offers a mechanistically
grounded contribution to local vector control and to the integrated strategies
through which malaria-control goals are pursued. The analysis presented here is
intended as a step toward placing such strategies on a solid quantitative, rigorous
mechanistic foundation.\ \\ \ \\ \noindent 

    \section*{Acknowledgments}

ABG acknowledges the support, in part, of the National Science Foundation (DMS-2052363;
transferred to DMS-2330801), University of Maryland Institute for Health Computing and the Brin and the State of Maryland Endowed E-Nnovate Chair Program. {\bf Generative AI Usage Statement:} During the preparation of this manuscript, the authors used the AI tools ChatGPT (OpenAI), Claude (Anthropic), and Gemini (Google) for proofreading and language editing, for assistance in checking the correctness of proofs, and for writing and checking code. The authors have reviewed and verified all AI-assisted contributions and take full responsibility for the integrity, accuracy, and originality of the model formulation, rigorous analysis, simulations, and interpretation of all results in this manuscript.


    \newpage
    \appendix

    \section{Estimation of spatial parameters}\label{app:spatial}

    As described in Section \ref{sec:numerics}, of all the parameters of the model \eqref{eq:model_I}, three are highly dependent on the spatial scale of the problem: the carrying capacity for eggs $\KE$, the larval density-dependent mortality coefficient $\delta_L$, and the mean-first-encounter time coefficient for mate finding $\gamma$. We claim that the parameter values $\KE=10^5$, $\delta_L=5\times 10^{-5}$, and $\gamma=450$ correspond to an area of about $2\,\text{km}^2$. The values of $\KE$ and $\delta_L$ are taken directly from \cite{IboGumTay2020}, so to estimate the area $A$, we need only determine the amount of territory needed to support $10^5$ mosquito eggs. Of course, the amount of land needed to support a given number of eggs varies wildly depending on the suitability of the habitat in question, but a rough estimate leads to a characteristic area of $A\approx 2\,\text{km}^2$ based on expected size and frequency of puddles in typical \emph{Anopheles} territory. In this appendix, we show how the value of $\gamma$ is estimated.

    In the absence of SIT-based interventions, the model \eqref{eq:model_I} uses the mating rate given by
    \begin{equation}
        R_{\text{mating}}=\frac{M_w F_u}{\gamma+\zeta M_w}.
    \end{equation}
    Thus, $\gamma$ is interpreted as the expected time required for a female to encounter a mate if there is exactly one male. If we assume that motion of mosquitoes is essentially Brownian with diffusion rate $D$, and an \emph{encounter distance} $r_e$ at which the two mosquitoes sense each other, then this expected time is
    \begin{equation}
        \gamma=\frac{A}{8\pi D}\log\left(\frac{L}{r_e}\right),
    \end{equation}
    where $L$ is the diameter of the domain, which we take to be $L=2\sqrt{A/\pi}$. This formulation is derived via Smoluchowski mean-first-passage-time theory \cite{LeVYusAbaDen2020,SinSchHol2006,SzaSchSch1980}. We estimate the diffusion coefficient of adult mosquitoes as $900\,m^2/\text{day}$ \cite{BeeEtAl2022}. Male and female mosquitoes are known to detect one another at a distance of about $r_e=10\,m$ \cite{MenEtAl2019}. With $A\approx 2\,\text{km}^2$, we compute $\gamma \approx 450$.

    \section{Proof Details}\label{app:ext_proof}

    \subsection{Detailed proof of Proposition \ref{prop:bounded_solns}}
    
    \begin{proof}
        Proposition \ref{prop:positivity} shows that $0$ is a lower bound for all state variables. We show that components of the solution are bounded above one-by-one. First, we observe that if $E(t)\geq \KE$ for any some $t$, then
        \begin{equation*}
            \dot E=\phi F_{mw}\left(1-\frac{E}{\KE}\right)-(\sigma_E+\mu_E)E\leq -(\sigma_E+\mu_E)\KE<0,
        \end{equation*}
        So solutions can not pass upward through $E=\KE$, and whenever $E(t)>\KE$ in a solution, the component $E(t)$ is decreasing. Thus, $E(t)$ is bounded:
        \begin{equation}
            0\leq E(t)\leq E_{\max}:=\max\{E(0),\KE\}.
        \end{equation}
        
        We will now prove the following generic fact: if $u(t)$ solves
        \begin{equation}
            \dot u\leq a-bu
        \end{equation}
        for $a,b>0$ and $u(0)\geq 0$ for all $t\in[0,T]$, then
        \begin{equation}\label{eq:u_bound}
            u(t)\leq u_{\text{max}}:=\max\{u(0),a/b\}.
        \end{equation}
        To see this, let $v=-(a-bu)$. Then
        \begin{equation}
            \dot v=b\dot u\leq b(a-bu)=-bv.
        \end{equation}
        Thus, by Gr\"onwall's inequality $v(t)\leq v(0)e^{-bt}$, from which we derive
        \begin{equation}
            u(t)\leq \frac{a}{b}+\left(u(0)-\frac{a}{b}\right)e^{-bt}.
        \end{equation}
        Therefore, if $u(0)\leq a/b$, then $u(t)<a/b$ for all $t$. On the other hand, if $u(0)>a/b$, then $u(t)\leq u(0)$ for all $t$. In either case, \eqref{eq:u_bound} holds.

        Now we can apply this fact repeatedly to the state variables of the model \eqref{eq:model_I}, starting with $L$: we have
        \begin{equation}
            \dot L=\sigma_EE-(\sigma_L+\mu_L)L-\delta_LL^2\leq \sigma_E E_{\text{max}}-(\sigma_L+\mu_L)L,
        \end{equation}
        so
        \begin{equation}
            0\leq L(t)\leq L_{\text{max}}:=\max\left\{L(0),\frac{\sigma_E E_{\text{max}}}{\sigma_L+\mu_L}\right\}.
        \end{equation}
        Next, since $0\leq S(t)\leq S_{\text{max}}$, we have
        \begin{equation}
            0\leq M_s(t)\leq M_{s,\text{max}}:=\max\left\{M_s(0),\frac{S_{\text{max}}}{\mu_M}\right\}.
        \end{equation}
        Since $L$ is bounded, we have $\dot P\leq \sigma_L L_{\text{max}}-(\sigma_P+\mu_P)P$. Therefore,
        \begin{equation}
            0\leq P(t)\leq P_{\text{max}}:=\max\left\{P(0),\frac{\sigma_L L_{\text{max}}}{\sigma_P+\mu_P}\right\}.
        \end{equation}
        Since $P$ is bounded, we have
        \begin{align*}
            \dot F_u&\leq r\sigma_PP_{\text{max}}-\mu_F F_u\\
            \dot M_w&\leq (1-r)\sigma_PP_{\text{max}}-\mu_M M_w.
        \end{align*}
        Thus,
        \begin{align}
            0&\leq F_u(t)\leq F_{u,\text{max}}:=\max\left\{F_u(0),\frac{r\sigma_PP_{\text{max}}}{\mu_F}\right\}\\
            0&\leq M_w(t)\leq M_{w,\text{max}}:=\max\left\{M_w(0),\frac{(1-r)\sigma_PP_{\text{max}}}{\mu_M}\right\}
        \end{align}
        Finally, we have
        \begin{align*}
            \dot F_{mw}&\leq \frac{M_{w,\text{max}}}{\gamma}F_{u,\text{max}}-\mu_F F_{mw}\\
            \dot F_{ms}&\leq \frac{\eta M_{s,\text{max}}}{\gamma}F_{u,\text{max}}-\mu_F F_{ms},
        \end{align*}
        so
        \begin{align}
            0&\leq F_{mw}(t)\leq F_{mw,\text{max}}:=\max\left\{F_{mw}(0),\frac{M_{w,\text{max}}F_{u,\text{max}}}{\gamma\mu_F}\right\}\\
            0&\leq F_{ms}(t)\leq F_{ms,\text{max}}:=\max\left\{F_{ms}(0),\frac{\eta M_{s,\text{max}}F_{u,\text{max}}}{\gamma\mu_F}\right\}
        \end{align}
        Thus, all state variables are bounded.
    \end{proof}

    \subsection{Derivation of \eqref{eq:Vdot_simplified}}

    We give the details leading from \eqref{eq:Lyapunov_derivative} to
\eqref{eq:Vdot_simplified}. Recall that
\[
\dot V
=
\dot E
+
\left(\frac{\sigma_E+\mu_E}{\sigma_E}\right)\dot L
+
\left(\frac{(\sigma_E+\mu_E)(\sigma_L+\mu_L)}{\sigma_E\sigma_L}\right)\dot P
+
\left(\frac{(\sigma_E+\mu_E)(\sigma_L+\mu_L)(\sigma_P+\mu_P)}
{r\sigma_E\sigma_L\sigma_P}\right)\dot F_u
+
\left(\frac{\phi}{\mu_F}\right)\dot F_{mw}.
\tag{\ref{eq:Lyapunov_derivative}}
\]
For notational convenience, define
\[
A_L=\frac{\sigma_E+\mu_E}{\sigma_E},\qquad
A_P=\frac{(\sigma_E+\mu_E)(\sigma_L+\mu_L)}{\sigma_E\sigma_L},
\]
and
\[
A_F=
\frac{(\sigma_E+\mu_E)(\sigma_L+\mu_L)(\sigma_P+\mu_P)}
{r\sigma_E\sigma_L\sigma_P}.
\]
Then
\[
\dot V=\dot E+A_L\dot L+A_P\dot P+A_F\dot F_u+\frac{\phi}{\mu_F}\dot F_{mw}.
\]
Substituting the right-hand sides of \eqref{eq:model_I}, with \(S(t)\equiv 0\), gives
\[
\begin{aligned}
\dot V
={}&
\phi\left(1-\frac{E}{K_E}\right)F_{mw}
-(\sigma_E+\mu_E)E  \\
&+A_L\left[\sigma_EE-(\sigma_L+\mu_L+\delta_LL)L\right] \\
&+A_P\left[\sigma_LL-(\sigma_P+\mu_P)P\right] \\
&+A_F\left[
r\sigma_PP
-\frac{M_w+\eta M_s}{\gamma+\zeta(M_w+\eta M_s)}F_u
-\mu_FF_u
\right]  \\
&+\frac{\phi}{\mu_F}\left[
\frac{M_w}{\gamma+\zeta(M_w+\eta M_s)}F_u
-\mu_FF_{mw}
\right].
\end{aligned}
\]
The choice of coefficients in \(V\) gives the cancellations
\[
A_L\sigma_E=\sigma_E+\mu_E,
\]
\[
A_P\sigma_L=A_L(\sigma_L+\mu_L),
\]
and
\[
A_Fr\sigma_P=A_P(\sigma_P+\mu_P).
\]
Therefore all linear \(E\), \(L\), and \(P\) terms cancel. The \(F_{mw}\) terms also simplify as
\[
\phi\left(1-\frac{E}{K_E}\right)F_{mw}-\phi F_{mw}
=
-\frac{\phi E}{K_E}F_{mw}.
\]
Thus
\[
\begin{aligned}
\dot V
={}&
-\frac{\phi E}{K_E}F_{mw}
-\delta_LA_LL^2  \\
&+
\left[
\frac{\phi}{\mu_F}
\frac{M_w}{\gamma+\zeta(M_w+\eta M_s)}
-
A_F\left(
\frac{M_w+\eta M_s}{\gamma+\zeta(M_w+\eta M_s)}
+\mu_F
\right)
\right]F_u .
\end{aligned}
\]
Now set
\[
C_m=
\frac{M_w+\eta M_s}{\gamma+\zeta(M_w+\eta M_s)}.
\]
Since \(M_w,M_s\ge 0\) and \(\eta\ge 0\),
\[
\frac{M_w}{\gamma+\zeta(M_w+\eta M_s)}
\le
\frac{M_w+\eta M_s}{\gamma+\zeta(M_w+\eta M_s)}
=C_m.
\]
Moreover,
\[
-\frac{\phi E}{K_E}F_{mw}\le 0.
\]
Hence
\[
\begin{aligned}
\dot V
&\le
-\delta_LA_LL^2
+
\left[
\frac{\phi}{\mu_F}C_m
-
A_F(C_m+\mu_F)
\right]F_u  \\
&=
-\delta_LA_LL^2
+
A_F\left[
\left(\frac{\phi}{\mu_FA_F}-1\right)C_m-\mu_F
\right]F_u .
\end{aligned}
\]
Using the definition of $\Rq$ \eqref{eq:Rlm}, we have
\[
\frac{\phi}{\mu_F}
=
A_F\mathcal R_0^q(1+\zeta\mu_F).
\]
Therefore
\[
\begin{aligned}
\dot V
&\le
-\delta_LA_LL^2
+
A_F\left[
\mathcal R_0^q(1+\zeta\mu_F)C_m
-C_m-\mu_F
\right]F_u  \\
&=
-\delta_LA_LL^2
+
A_F\left[
(\mathcal R_0^q-1)C_m
+
(\mathcal R_0^q C_m\zeta-1)\mu_F
\right]F_u .
\end{aligned}
\]
Returning to the definitions of \(A_L\) and \(A_F\), we obtain
\[
\dot V
\le
\frac{(\sigma_E+\mu_E)(\sigma_L+\mu_L)(\sigma_P+\mu_P)}
{r\sigma_E\sigma_L\sigma_P}
\left[
(\mathcal R_0^q-1)C_m
+
(\mathcal R_0^q C_m\zeta-1)\mu_F
\right]F_u
-
\delta_L
\left(\frac{\sigma_E+\mu_E}{\sigma_E}\right)L^2 .
\]

    \subsection{Derivation of \eqref{eq:Lstar} and \eqref{eq:Lstarstar}}
\label{app:Lstar_derivation}

We derive the asymptotic formulas for the two positive roots
\(L^{\ast\ast}_{+}\) and \(L^{\ast\ast}_{-}\) appearing in
\eqref{eq:Lstar} and \eqref{eq:Lstarstar}. Recall from the proof of
Theorem~\ref{thm:positive_equilibria} that, after setting
\(\delta_L=0\), the cubic polynomial \(p_{\delta_L}(s)\) reduces to the
quadratic
\[
p_0(s)=b_0s^2+c_0s+d_0,
\]
where
\[
b_0=
\frac{\mathcal R_0^q(1-r)\mu_F\sigma_L\sigma_P(1+\zeta\mu_F)
(\sigma_E+\mu_E)(\sigma_P+\mu_P)(\sigma_L+\mu_L)^2}
{K_E\sigma_E},
\]
\[
c_0=
-(\mathcal R_0^q-1)(1-r)\sigma_L\sigma_P\mu_F
(\sigma_E+\mu_E)(\sigma_L+\mu_L)(\sigma_P+\mu_P)(1+\zeta\mu_F),
\]
and
\[
d_0=
\gamma\mu_F^2\mu_M(\sigma_E+\mu_E)(\sigma_L+\mu_L)(\sigma_P+\mu_P)^2.
\]
Since \(\mathcal R_0^q>1\), we have \(b_0>0\), \(c_0<0\), and
\(d_0>0\). It is useful to write
\[
C=-c_0>0,\qquad D=d_0>0,\qquad \varepsilon=b_0>0.
\]
Then
\[
p_0(s)=\varepsilon s^2-Cs+D.
\]
The two roots of \(p_0\) are
\[
s_{\pm,0}
=
\frac{C\pm\sqrt{C^2-4\varepsilon D}}{2\varepsilon}.
\]
Here \(\varepsilon=b_0=O(K_E^{-1})\), while \(C\) and \(D\) are
independent of \(K_E\). For \(K_E\) sufficiently large, the discriminant
is positive. Expanding the square root gives
\[
\sqrt{C^2-4\varepsilon D}
=
C\sqrt{1-\frac{4\varepsilon D}{C^2}}
=
C-\frac{2\varepsilon D}{C}+O(\varepsilon^2).
\]
Therefore
\[
s_{+,0}
=
\frac{C+\sqrt{C^2-4\varepsilon D}}{2\varepsilon}
=
\frac{C}{\varepsilon}-\frac{D}{C}+O(\varepsilon),
\]
whereas
\[
s_{-,0}
=
\frac{C-\sqrt{C^2-4\varepsilon D}}{2\varepsilon}
=
\frac{D}{C}+O(\varepsilon).
\]
We now compute \(C/\varepsilon\) and \(D/C\). First, after canceling common factors, we find
\[
\frac{C}{\varepsilon}
=
\frac{K_E(\mathcal R_0^q-1)\sigma_E}
{\mathcal R_0^q(\sigma_L+\mu_L)}.
\]
Similarly,
\[
\frac{D}{C}
=
\frac{
\gamma\mu_F\mu_M(\sigma_P+\mu_P)
}{
(1-r)(\mathcal R_0^q-1)\sigma_L\sigma_P(1+\zeta\mu_F)
}.
\]
Thus we obtain
\[
s_{+,0}
=
\frac{K_E(\mathcal R_0^q-1)\sigma_E}
{\mathcal R_0^q(\sigma_L+\mu_L)}
-
\frac{
\gamma\mu_F\mu_M(\sigma_P+\mu_P)
}{
(1-r)(\mathcal R_0^q-1)\sigma_L\sigma_P(1+\zeta\mu_F)
}
+
O(K_E^{-1}),
\]
and
\[
s_{-,0}
=
\frac{
\gamma\mu_F\mu_M(\sigma_P+\mu_P)
}{
(1-r)(\mathcal R_0^q-1)\sigma_L\sigma_P(1+\zeta\mu_F)
}
+
O(K_E^{-1}).
\]

Finally, for \(\delta_L>0\) sufficiently small, the corresponding roots
of the full cubic \(p_{\delta_L}(s)\) perturb continuously from the
simple roots \(s_{\pm,0}\) of \(p_0\). Hence
\[
L^{\ast\ast}_{+}=s_{+,0}+O(\delta_L),
\qquad
L^{\ast\ast}_{-}=s_{-,0}+O(\delta_L).
\]
Combining this with the preceding expansions gives
\[
L^{\ast\ast}_{+}
=
\frac{K_E(\mathcal R_0^q-1)\sigma_E}
{\mathcal R_0^q(\sigma_L+\mu_L)}
-
\frac{
\gamma\mu_F\mu_M(\sigma_P+\mu_P)
}{
(1-r)(\mathcal R_0^q-1)\sigma_L\sigma_P(1+\zeta\mu_F)
}
+
O(K_E^{-1},\delta_L),
\]
and
\[
L^{\ast\ast}_{-}
=
\frac{
\gamma\mu_F\mu_M(\sigma_P+\mu_P)
}{
(1-r)(\mathcal R_0^q-1)\sigma_L\sigma_P(1+\zeta\mu_F)
}
+
O(K_E^{-1},\delta_L).
\]
These are precisely \eqref{eq:Lstar} and \eqref{eq:Lstarstar}.

    \subsection{Limit of the determinant \eqref{eq:determinant_of_B-}}

    To compute the limit of \eqref{eq:determinant_of_B-} as $(\KE,\delta_L)\to (\infty,0)$, we compute limits of $M_w$, $H$, and $F_u$ as $\KE\to\infty$. Using \eqref{eq:Lstarstar} and the formula for \(M_w^{\ast\ast}\) in
\eqref{eq:equil_components}, we have
\[
\lim_{\KE\to\infty}M^{\ast\ast}_{w,-}
=
\frac{\gamma\mu_F}
{(\Rq-1)(1+\zeta\mu_F)}.
\]
Hence, with \(H=\gamma+\zeta M^{\ast\ast}_{w,-}\),
\[
\lim_{\KE\to\infty}H
=
\gamma+\frac{\zeta\gamma\mu_F}
{(\Rq-1)(1+\zeta\mu_F)}
=
\frac{\gamma\bigl[\Rq(1+\zeta\mu_F)-1\bigr]}
{(\Rq-1)(1+\zeta\mu_F)}.
\]
It follows in particular that
\[
\lim_{\KE\to\infty}\frac{M^{\ast\ast}_{w,-}}{H}
=
\frac{\mu_F}{\Rq(1+\zeta\mu_F)-1}.
\]
Finally, using the equilibrium identities
\[
F_u^{\ast\ast}
=
\frac{r\sigma_P P^{\ast\ast}}
{M_w^{\ast\ast}/H+\mu_F},
\qquad
P^{\ast\ast}
=
\frac{\sigma_LL^{\ast\ast}}{\mu_P+\sigma_P},
\]
we obtain
\[
\lim_{\KE\to\infty}F^{\ast\ast}_{u,-}
=
\frac{
\dfrac{r\gamma\mu_F\mu_M}
{(1-r)(\Rq-1)(1+\zeta\mu_F)}
}{
\dfrac{\mu_F}{\Rq(1+\zeta\mu_F)-1}+\mu_F
}
=
\frac{
r\gamma\mu_M\bigl[\Rq(1+\zeta\mu_F)-1\bigr]
}{
(1-r)\Rq(\Rq-1)(1+\zeta\mu_F)^2
}.
\]
Substituting these identities into \eqref{eq:determinant_of_B-} and taking $\delta_L=0$, we obtain the limit \eqref{eq:limit_of_determinant_of_B-}.

\subsection{Derivation of \eqref{eq:constant_control_Lyapunov_derivative}}
\label{app:constant_control_Lyapunov_derivative}

We derive the differential inequality used in the proof of
Theorem~\ref{thm:finite_cost_constant_control}. Recall that \(V\) is defined by
\eqref{eq:SIT_Lyapunov_function}. Set
\[
A_L=\frac{\sigma_E+\mu_E}{\sigma_E},\qquad
A_P=\frac{(\sigma_E+\mu_E)(\sigma_L+\mu_L)}{\sigma_E\sigma_L},
\]
and
\[
A_F=
\frac{(\sigma_E+\mu_E)(\sigma_L+\mu_L)(\sigma_P+\mu_P)}
{r\sigma_E\sigma_L\sigma_P}.
\]
Then
\[
\dot V=\dot E+A_L\dot L+A_P\dot P+A_F\dot F_u
+\frac{\phi}{\mu_F}\dot F_{mw}.
\]
Using
\[
\dot E\le \phi F_{mw}-(\sigma_E+\mu_E)E,
\]
and substituting the remaining equations from \eqref{eq:model_I}, we obtain
\[
\begin{aligned}
\dot V
\le{}&
\phi F_{mw}-(\sigma_E+\mu_E)E  \\
&+A_L\left[\sigma_EE-(\sigma_L+\mu_L+\delta_LL)L\right] \\
&+A_P\left[\sigma_LL-(\sigma_P+\mu_P)P\right] \\
&+A_F\left[
r\sigma_PP
-\frac{M_w+\eta M_s}{\gamma+\zeta(M_w+\eta M_s)}F_u
-\mu_FF_u
\right] \\
&+\frac{\phi}{\mu_F}\left[
\frac{M_w}{\gamma+\zeta(M_w+\eta M_s)}F_u
-\mu_FF_{mw}
\right].
\end{aligned}
\]
The coefficients in \(V\) were chosen so that the linear \(E\), \(L\),
\(P\), and \(F_{mw}\) terms cancel. Therefore
\[
\dot V
\le
-\delta_LA_LL^2
-
A_F\left(
\frac{M_w+\eta M_s}{\gamma+\zeta(M_w+\eta M_s)}
+\mu_F
\right)F_u
+
\frac{\phi}{\mu_F}
\frac{M_w}{\gamma+\zeta(M_w+\eta M_s)}F_u .
\]
Dropping the nonpositive term involving
\((M_w+\eta M_s)/(\gamma+\zeta(M_w+\eta M_s))\), we get
\[
\dot V
\le
-\delta_LA_LL^2
-
A_F\mu_F F_u
+
\frac{\phi}{\mu_F}
\frac{M_w}{\gamma+\zeta(M_w+\eta M_s)}F_u .
\]
By the definition of \(\Rq\),
\[
A_F\mu_F=\frac{\phi}{\Rq(1+\zeta\mu_F)}.
\]
Moreover, since \(0\le M_w(t)\le M_{w,\max}\) and the function
\[
x\mapsto \frac{x}{\gamma+\zeta(x+\eta M_s(t))}
\]
is increasing for \(x\ge 0\), we have
\[
\frac{M_w(t)}{\gamma+\zeta(M_w(t)+\eta M_s(t))}
\le
\frac{M_{w,\max}}{\gamma+\zeta(M_{w,\max}+\eta M_s(t))}.
\]
Thus
\[
\dot V
\le
-\left[
\frac{\phi}{\Rq(1+\zeta\mu_F)}
-
\frac{\phi}{\mu_F}
\frac{M_{w,\max}}{\gamma+\zeta(M_{w,\max}+\eta M_s(t))}
\right]F_u
-
\delta_L\left(\frac{\sigma_E+\mu_E}{\sigma_E}\right)L^2.
\]
Defining
\[
Q(t)=
\frac{\phi}{\Rq(1+\zeta\mu_F)}
-
\frac{\phi}{\mu_F}
\frac{M_{w,\max}}{\gamma+\zeta(M_{w,\max}+\eta M_s(t))},
\]
we obtain
\[
\dot V
\le
-Q(t)F_u
-
\delta_L\left(\frac{\sigma_E+\mu_E}{\sigma_E}\right)L^2,
\]
which is \eqref{eq:constant_control_Lyapunov_derivative}.

    \section{The cause of the Allee effect in model \eqref{eq:model_I}}\label{app:linear_mating}

    The Allee effect in \eqref{eq:model_I} is induced by the nonzero time spent by unmated adult females searching for a mate. That is, the stability of the MFE is a consequence of the presence of $\gamma>0$ in \eqref{eq:model_I}. This is curious since $\gamma$ does not even appear in the linearization \eqref{eq:AI_0}. But to underscore how this is nevertheless true, we consider a simplified \emph{quick mate search} version of the model \eqref{eq:model_I} in which there is no SIT and $\gamma=0$, but all other parameters are as in \eqref{eq:model_I}. In this simplified model, mate finding is instantaneous, but there is still a post-emergence refractory period with average time $\zeta>0$. Thus, males no longer appear in the mating terms, so we do not need to explicitly represent them. The resulting simplified model is:
    \begin{equation}\label{eq:model_IV}
		\begin{cases}
			\dot E=\phi F_{m}\left(1-\frac{E}{\KE}\right)-(\sigma_E+\mu_E)E & t>0\\
			\dot L=\sigma_EE-(\sigma_L+\mu_L+\delta_L L)L & t>0\\
			\dot P=\sigma_LL-(\sigma_P+\mu_P)P & t>0\\
			\dot F_u=r\sigma_PP-\left(\frac{1}{\zeta}+\mu_F\right)F_u & t>0\\
			\dot F_m=\frac{1}{\zeta}F_u-\mu_FF_{m} & t>0.
		\end{cases}
	\end{equation}
    
    If we apply the next generation matrix method, we find that the basic reproduction number for \eqref{eq:model_IV}, which we call the \emph{quick mate search reproduction number} is in fact $\Rq$. We show that simply by setting $\gamma=0$, we eliminate the Allee effect. Indeed the following proposition shows that the MFE can be unstable:
    \begin{proposition}\label{prop:linear_mating}
        The MFE of the quick mate search model \eqref{eq:model_IV} is unstable for $\Rq>1$.
    \end{proposition}
    \begin{proof}
        We find the linearization of \eqref{eq:model_IV} about the MFE, denoted $\tilde A(\mathcal E_0)$:
        \begin{equation}
            \tilde A(\mathcal E_0)=\begin{pmatrix}
                -(\sigma_E+\mu_E) & 0 & 0 & 0 & \phi\\
                \sigma_E & -(\sigma_L+\mu_L) & 0 & 0 & 0\\
                0 & \sigma_L & -(\sigma_P+\mu_P) & 0 & 0\\
                0 & 0 & r\sigma_P & -\left(\frac{1}{\zeta}+\mu_F\right) & 0\\
                0 & 0 & 0 & \frac{1}{\zeta} & -\mu_F
            \end{pmatrix}
        \end{equation}
        Note that since $\gamma=0$, some of the partial derivatives of the simplified model \eqref{eq:model_IV} are different than those of the full model \eqref{eq:model_I}, as evidenced by the entry $1/\zeta$ in the fifth row and fourth column, which is not present in \eqref{eq:AI_0}.
        
        Though \eqref{eq:model_IV} is a simplification of \eqref{eq:model_I}, the eigenvectors of and eigenvalues of $\tilde A(\mathcal E_0)$ cannot be written explicitly as they can for $A(\mathcal E_0)$. Instead, we compute the determinant of $\tilde A(\mathcal E_0)$ which can be written in terms of $\Rq$ as
        \begin{equation}
            \text{det}\,\tilde A(\mathcal E_0)=\mu_F(\Rq-1)(1+\zeta\mu_F)(\sigma_E+\mu_E)(\sigma_L+\mu_L)(\sigma_P+\mu_P).
        \end{equation}
        We see that this determinant is positive if and only if $\Rq>1$. The sign of the determinant is $(-1)^c$ where $c$ is the number of eigenvalues (counting multiplicity) with negative real part. Since $\tilde A(\mathcal E_0)$ has five eigenvalues (counting multiplicity), they cannot all have negative real part if $\Rq>1$. We conclude that the MFE is unstable if $\Rq>1$.
    \end{proof}

    Since the MFE in the quick mate search model can be unstable, we see that simply introducing a time to find a mate is enough to induce an Allee effect. Indeed, this is explained by the observation that if the population of males is very small, then females may not be able to find a mate fast enough to reproduce before they die. Thus, it is easier to control the mosquito population when this effect is present.

	\bibliographystyle{siam}
	\bibliography{refs}
	
\end{document}